\documentclass{article}
\usepackage[T1]{fontenc}
\usepackage[margin=1in]{geometry}
\usepackage{amsfonts,amsmath,amsthm,amssymb,mathtools}  

\mathtoolsset{centercolon}
\usepackage{xfrac,nicefrac}
\usepackage{mathdots}
\usepackage{mleftright}  
\let\left\mleft
\let\right\mright

\usepackage{xspace}
\xspaceaddexceptions{]\}}  
\usepackage{regexpatch}

\usepackage{bm,bbm,dsfont}  
\usepackage{caption}
\usepackage[normalem]{ulem}
\usepackage{enumitem}

\usepackage{booktabs}
\usepackage{graphicx}
\usepackage{float}
\usepackage{subcaption}  
\usepackage{tcolorbox}
\usepackage{tikz}
\usetikzlibrary{decorations.pathreplacing}
\usetikzlibrary{calc}
\usetikzlibrary{positioning}
\usetikzlibrary{arrows.meta}
\usetikzlibrary{math}
\usetikzlibrary{patterns}
\usetikzlibrary{trees}

\usetikzlibrary{trees}

\usepackage[linesnumbered,boxed,ruled,vlined]{algorithm2e}

\SetCommentSty{mycommfont}

\usepackage{thmtools,thm-restate}
\usepackage[colorlinks,citecolor=blue,linkcolor=blue,urlcolor=red]{hyperref}

\theoremstyle{plain}

\newtheorem{theorem}{Theorem}[section]  
\newtheorem{lemma}[theorem]{Lemma}

\newtheorem{proposition}[theorem]{Proposition}
\newtheorem{corollary}[theorem]{Corollary}

\newtheorem{claim}[theorem]{Claim}

\theoremstyle{definition}  
\newtheorem{definition}[theorem]{Definition}

\newenvironment{proofof}[1]{\begin{proof}[Proof of #1]}{\end{proof}}

\AtBeginDocument{%
}

\usepackage[capitalise]{cleveref}
\crefname{algocf}{Algorithm}{Algorithms}
\Crefname{algocf}{Algorithm}{Algorithms}
\crefname{claim}{Claim}{Claims}
\Crefname{claim}{Claim}{Claims}
\crefname{statement}{Statement}{Statements}
\Crefname{statement}{Statement}{Statements}

\newfloat{Distribution}{htbp}{loa}
\crefname{Distribution}{Distribution}{Distributions}
\Crefname{Distribution}{Distribution}{Distributions}
\newfloat{Protocol}{htbp}{loa}
\crefname{Protocol}{Protocol}{Protocols}
\Crefname{Protocol}{Protocol}{Protocols}

\SetKwProg{Function}{Function}{:}{}

\SetKwProg{CodeBlock}{}{}{}
\SetKwProg{Repeat}{Repeat}{:}{}

\DeclarePairedDelimiter{\ceil}{\lceil}{\rceil}
\DeclarePairedDelimiter{\floor}{\lfloor}{\rfloor}

\DeclarePairedDelimiter{\bk}{(}{)}
\DeclarePairedDelimiter{\Bk}{[}{]}
\DeclarePairedDelimiter{\BK}{\{}{\}}

\DeclarePairedDelimiter{\abs}{\lvert}{\rvert}

\DeclarePairedDelimiterX\mysetbase[2]{\lbrace}{\rbrace}{#1\,\delimsize\vert\,#2}
\NewDocumentCommand{\myset}{sO{}m m}{%
  \IfBooleanTF{#1}
    {\mysetbase*{#3}{#4}}
    {\mysetbase[#2]{#3}{#4}}
}

\DeclareMathOperator*{\E}{\mathbb{E}}
\DeclareMathOperator*{\Var}{Var}

\let\Pr\PrAux
\DeclareMathOperator{\poly}{poly}

\DeclareMathOperator*{\ind}{\mathbbm{1}}

\renewcommand{\tilde}{\widetilde}

\newcommand{\defeq}{\coloneqq}
\newcommand{\eps}{\varepsilon}

\newcommand{\N}{\mathbb{N}}

\newcommand{\Z}{\mathbb{Z}}
\let\origl\l
\renewcommand{\l}{\ifmmode\ell\else\origl\fi}
\renewcommand{\emptyset}{\varnothing}
\renewcommand{\epsilon}{\eps}

\newcommand{\defn}[1]{\emph{\boldmath\textbf{#1}}}

\usepackage{regexpatch}
\makeatletter
\xpatchcmd\thmt@restatable{%
\csname #2\@xa\endcsname\ifx\@nx#1\@nx\else[{#1}]\fi
}{%
\ifthmt@thisistheone
\csname #2\@xa\endcsname\ifx\@nx#1\@nx\else[{#1}]\fi
\else
\csname #2\@xa\endcsname[{Restated}]
\fi}{}{}
\makeatother

\usepackage{physics}

\newcommand{\disp}{\textup{disp}}
\newcommand{\pc}{\textup{pc}}
\newcommand{\Poi}{\operatorname{Poisson}}

\usepackage[disable]{todonotes}

\usepackage{newunicodechar}
\usepackage{silence}
\newunicodechar{♢}{\tikz \node[inner sep=1.5,draw,diamond] {};}
\newunicodechar{☆}{\tikz \node[inner sep=1,draw,star,star point ratio=2] {};}
\newunicodechar{△}{\triangle}
\newunicodechar{⬜}{\kern 0.5pt\tikz \node[inner sep=1.7,draw,regular polygon,regular polygon sides=4] {};\kern 0.5pt}
\newunicodechar{◯}{\tikz[baseline=-3pt] \node[inner sep=1.7,draw,cloud,cloud puffs=4,cloud puff arc=190] {};}

\newunicodechar{⊥}{\bot}
\newunicodechar{•}{\item}
\newunicodechar{✓}{\checkmark}
\newunicodechar{✗}{\xmark}\def\xmark{\ding{55}}
\newunicodechar{…}{\dots}
\newunicodechar{≔}{\coloneqq}
\newunicodechar{⁻}{^-}
\newunicodechar{⁺}{^+}
\newunicodechar{₋}{_-}
\newunicodechar{₊}{_+}
\newunicodechar{ℓ}{\ell}
\newunicodechar{•}{\item}
\newunicodechar{…}{\dots}
\newunicodechar{≔}{\coloneqq}
\newif\ifnormopen\normopenfalse
\newunicodechar{‖}{\ifnormopen\rVert\normopenfalse\else\rVert\normopentrue\fi}
\newunicodechar{≤}{\leq}
\newunicodechar{≥}{\geq}
\newunicodechar{≰}{\nleq}
\newunicodechar{≱}{\ngeq}
\newunicodechar{⊕}{\oplus}
\newunicodechar{⊗}{\otimes}
\newunicodechar{≠}{\neq}
\newunicodechar{¬}{\neg}
\newunicodechar{≡}{\equiv}
\newunicodechar{₀}{_0}
\newunicodechar{₁}{_1}
\newunicodechar{₂}{_2}
\newunicodechar{₃}{_3}
\newunicodechar{₄}{_4}
\newunicodechar{₅}{_5}
\newunicodechar{₆}{_6}
\newunicodechar{₇}{_7}
\newunicodechar{₈}{_8}
\newunicodechar{₉}{_9}
\newunicodechar{ₚ}{_p}
\newunicodechar{ₙ}{_n}
\newunicodechar{ₐ}{_a}
\newunicodechar{ₑ}{_e}
\newunicodechar{ₕ}{_h}
\newunicodechar{ₖ}{_k}
\newunicodechar{ₗ}{_l}
\newunicodechar{ₘ}{_m}
\newunicodechar{ₛ}{_s}
\newunicodechar{ₜ}{_t}
\newunicodechar{ₓ}{_x}
\newunicodechar{⁰}{^0}
\newunicodechar{¹}{^1}
\newunicodechar{²}{^2}
\newunicodechar{³}{^3}
\newunicodechar{⁴}{^4}
\newunicodechar{⁵}{^5}
\newunicodechar{⁶}{^6}
\newunicodechar{⁷}{^7}
\newunicodechar{⁸}{^8}
\newunicodechar{⁹}{^9}
\newunicodechar{ⁿ}{^n}
\newunicodechar{∈}{\in}
\newunicodechar{∉}{\notin}
\newunicodechar{⊂}{\subset}
\newunicodechar{⊃}{\supset}
\newunicodechar{⊆}{\subseteq}
\newunicodechar{⊇}{\supseteq}
\newunicodechar{⊄}{\nsubset}
\newunicodechar{⊅}{\nsupset}
\newunicodechar{⊈}{\nsubseteq}
\newunicodechar{⊉}{\nsupseteq}
\newunicodechar{∪}{\cup}
\newunicodechar{∩}{\cap}
\newunicodechar{∀}{\forall}
\newunicodechar{∃}{\exists}
\newunicodechar{∄}{\nexists}
\newunicodechar{∨}{\vee}
\newunicodechar{∧}{\wedge}
\newunicodechar{⊼}{\bar{\wedge}}
\newunicodechar{⊽}{\bar{\vee}}
\newunicodechar{⌊}{\lfloor}
\newunicodechar{⌋}{\rfloor}
\newunicodechar{⌈}{\lceil}
\newunicodechar{⌉}{\rceil}
\newunicodechar{·}{\cdot}
\newunicodechar{∘}{\circ}
\newunicodechar{×}{\times}
\newunicodechar{↑}{\uparrow}
\newunicodechar{↓}{\downarrow}
\newunicodechar{→}{\rightarrow}
\newunicodechar{←}{\leftarrow}
\newunicodechar{⇒}{\Rightarrow}
\newunicodechar{⇐}{\Leftarrow}
\newunicodechar{↔}{\leftrightarrow}
\newunicodechar{⇔}{\Leftrightarrow}
\newunicodechar{↦}{\mapsto}
\newunicodechar{∅}{\emptyset}
\newunicodechar{∞}{\infty}
\newunicodechar{≅}{\cong}
\newunicodechar{≈}{\approx}
\newunicodechar{ℓ}{\ell}
\newunicodechar{↪}{\hookrightarrow}

\newunicodechar{ℂ}{\mathbb{C}}
\newunicodechar{ℍ}{\mathbb{H}}
\newunicodechar{ℕ}{\mathbb{N}}
\newunicodechar{ℙ}{\mathbb{P}}
\newunicodechar{ℚ}{\mathbb{Q}}
\newunicodechar{ℝ}{\mathbb{R}}
\newunicodechar{ℤ}{\mathbb{Z}}

\newunicodechar{α}{\alpha}
\newunicodechar{β}{\beta}
\newunicodechar{γ}{\gamma}
\newunicodechar{Γ}{\Gamma}
\newunicodechar{δ}{\delta}
\newunicodechar{Δ}{\Delta}
\newunicodechar{ε}{\varepsilon}
\newunicodechar{ζ}{\zeta}
\newunicodechar{η}{\eta}
\newunicodechar{θ}{\theta}
\newunicodechar{Θ}{\Theta}
\newunicodechar{ι}{\iota}
\newunicodechar{κ}{\kappa}
\newunicodechar{λ}{\lambda}
\newunicodechar{Λ}{\Lambda}
\newunicodechar{μ}{\mu}
\newunicodechar{ν}{\nu}
\newunicodechar{ξ}{\xi}
\newunicodechar{Ξ}{\Xi}
\newunicodechar{π}{\pi}
\newunicodechar{Π}{\Pi}
\newunicodechar{ρ}{\rho}
\newunicodechar{σ}{\sigma}
\newunicodechar{Σ}{\Sigma}
\newunicodechar{τ}{\tau}
\newunicodechar{υ}{\upsilon}
\newunicodechar{ϒ}{\Upsilon}
\newunicodechar{φ}{\varphi}
\newunicodechar{ϕ}{\phi}
\newunicodechar{Φ}{\Phi}
\newunicodechar{χ}{\chi}
\newunicodechar{ψ}{\psi}
\newunicodechar{Ψ}{\Psi}
\newunicodechar{ω}{\omega}
\newunicodechar{Ω}{\Omega}

\newunicodechar{ℋ}{\mathcal{H}}

\title{Quadratic Probing Revisited: \\ Smoothed Analysis and the Fall of Robin Hood}
\author{
  Yang Hu\thanks{Tsinghua University. \texttt{y-hu22@mails.tsinghua.edu.cn}}
  \and
  William Kuszmaul\thanks{Carnegie Mellon University. \texttt{kuszmaul@cmu.edu}}
  \and
  Jingxun Liang\thanks{Carnegie Mellon University. \texttt{jingxunl@andrew.cmu.edu}}
  \and
  Stefan Walzer\thanks{Karlsruhe Institute of Technology. \texttt{stefan.walzer@kit.edu}}
  \and
  Huacheng Yu\thanks{Princeton University. \texttt{hy2@cs.princeton.edu}}
  \and
  Renfei Zhou\thanks{Carnegie Mellon University. \texttt{renfeiz@andrew.cmu.edu}}
}
\date{}

\begin{document}

\maketitle

\begin{abstract}
    Quadratic probing is one of the most widely used open-addressing hash-table schemes in practice, but after more than half a century, even its most basic performance guarantees remain poorly understood. In this paper, we revisit quadratic probing through the lens of a smoothed variant in which each key follows a random probe sequence where its $k$th probe is expected at offset $\Theta(k^2)$. This is simultaneously a toy model for better understanding regular quadratic probing and a natural hashing scheme in its own right.
    We analyse smoothed quadratic probing for both Robin Hood ordering and anti-Robin Hood ordering and reveal a surprising separation: At load factor $1-\varepsilon$, anti-Robin Hood achieves an expected query time of $\Theta(\log \varepsilon^{-1})$, which matches the conjectured expected average successful query time for regular quadratic probing, while Robin Hood falls short at $\Theta(\varepsilon^{-1/2})$.
    Our analysis generalises to degree-$d$ probing for any $d \ge 1$ with expected query time $O(\max(\log \varepsilon^{-1}, \varepsilon^{1-2/d}))$ for anti-Robin Hood and $\Theta(\varepsilon^{-1/d})$ for Robin Hood.
    
    Finally, we go beyond smoothed analysis: using the probabilistic method, we show that for every $d \ge 2$, almost every random fixed-offset degree-$d$ probing sequence achieves expected query time $O(\log \varepsilon^{-1})$ under anti-Robin Hood ordering, simultaneously over all admissible table sizes and load factors. Thus, while quadratic probing itself remains elusive, we prove that essentially all quadratic-probing-like fixed-offset schemes achieve the ideal performance under the anti-Robin Hood ordering.
\end{abstract}

\section{Introduction}
\label{sec:introduction}

In 1954, a group of engineers at IBM devised a simple but highly effective way to store data in a computer \cite{KnuthVol3}. Their data structure, known today as a \defn{linear probing hash table}, worked as follows. Elements are stored in an array of some size $n$, and a new element $x$ can be inserted by placing it in the first available position from the sequence $h(x), h(x) + 1, h(x) + 2, h(x) + 3, \ldots \pmod n$. Likewise, to query an element $x$, one can search the positions $h(x) + 1, h(x) + 2, h(x) + 3, \ldots \pmod n$ until either $x$ is located, or until a free slot is found (meaning $x$ isn't present).

The linear-probing hash table quickly became the dominant data structure used at IBM \cite{bashe1986ibm} (it was even taught to users in a 1958 manual \cite{IBM1958RAMACProgrammersGuide}). In 1957, IBM researcher Peterson wrote a research paper investigating the data structure \cite{peterson1957addressing}. Peterson was especially interested in how the hash table behaves when filled to a high load factor $1 - \epsilon$ (meaning that $(1 - \epsilon)n$ elements are present). He conjectured -- somewhat contrary to his own experimental evidence -- that the expected insertion time in such a hash table should behave as $O(\epsilon^{-1})$. This conjecture was subsequently disproven in a 1962 paper by Donald Knuth \cite{knuth1962notes}, who showed that the expected insertion time is $\Theta(\epsilon^{-2})$. Knuth's result revealed a \emph{clustering effect}, in which keys form runs that are quadratically longer than one might intuitively expect.

Consider now a \emph{successful} query, i.e.\ a query for a key that is stored in the table. Such a query time matches the time that the key's insertion took, which means that the expected query time of a key in the table depends on the insertion order.
The \emph{average} expected query time of all keys stored in a table with load factor $1-ε$ is hence the average expected insertion time and therefore
\[
  \int_{ε}^1 δ^{-2} dδ = O(ε^{-1}).
\]
There is an elegant insertion policy that makes this the expected query time of all keys, even keys not stored in the table, which otherwise have insertion time $Θ(ε^{-2})$. This policy is known as Robin Hood ordering (named in \cite{CelisLaMu85,celis1986robin} but see also \cite{AmbleKn74}):
If a key $x$ probes a position that is already occupied by a key $y$, then whichever key has travelled further from its origin gets to stay while the other continues probing. Note that during queries a key can stop probing when coming across a key that it would have had priority over.
The name is a reference to the character Robin Hood's philosophy of taking from the rich and giving to the poor -- likewise the Robin Hood policy aims to equalize query times as much as possible.


To improve insertion time, Maurer proposed in 1968 a variant of linear probing called \defn{quadratic probing} \cite{maurer1968programming}, in which queries follow the probe sequence 
$$h(x), h(x) + 1^2, h(x) + 2^2, h(x) + 3^2, \ldots.$$
Researchers \cite{maurer1968programming,hopgood1972quadratic} observed that, with this probe sequence, elements were able to ``escape'' the clusters that would have formed with linear probing. 
They conjectured \cite{hopgood1972quadratic} that quadratic probing really does achieve an expected insertion time of $O(\epsilon^{-1})$, which implies (by an integration argument like above) average expected query times of $O(\log \varepsilon^{-1})$ for keys stored in the table. This conjecture, though never proven, remains widely believed today \cite{horowitz1992fundamentals, lafore2017data, necaise2010data, weiss2014dsaa_cpp}, and is the reason that hash tables like Google's Swiss Table (arguably the most widely used high-performance hash table in practice \cite{chern2017learnhashtable}) default to (variations of) quadratic probing \cite{abseil}.

From a theory perspective, the quadratic probing hash table has proven remarkably hard to analyze. After more than half a century, we still cannot say almost \emph{anything} nontrivial about how the hash table behaves. Not only can we not prove an insertion time of $O(\epsilon^{-1})$; we cannot even prove an insertion time of $f(\epsilon^{-1})$ for \emph{any} function $f$; or, even that a hash table that is 50\% full supports constant expected-time insertions. The only nontrivial result that we do know is that, when the hash table is filled to approximately 8.7\% full, it does indeed support $O(1)$ expected-time insertions \cite{kuszmaul2024towards}. 

Our inability to analyze quadratic probing's insertion time has also prevented us from answering several (perhaps less obvious) questions:
\begin{enumerate}
\item Does the Robin Hood policy improve unsuccessful query times, i.e.\ of keys not stored in the table, as it does for linear probing?
\item Does the Robin Hood policy equalize successful query times, i.e.\ of keys stored in the table, without affecting their average, as it does for linear probing?
\item What is the effect of selecting offsets using a degree-$d$ polynomial for $d \not\in \{1, 2\}$? \todo{I feel emphasizing this point more is benificial: We actually shows that $d = 2$ is the best of two worlds with best insertion time and locality.}
\end{enumerate}
Note that, whereas the conjectured insertion time is mostly a theoretical issue (practitioners already believe the bound of $O(\epsilon^{-1})$), these additional questions are not.\footnote{See, e.g., the discussion of Question (1) in \cite{chern2017learnhashtable}, or Question (2) in \cite{jagielski1976cubic,weiss2014dsaa_cpp}.} Their answers are of equal interest to practitioners as to theoreticians.

\subsection{Contribution}
\paragraph{Smoothed quadratic probing. } In this paper, we argue that, even without a complete analysis of quadratic probing, it is still possible to make significant progress on the questions above (and, in the process, to make significant progress towards the analysis of quadratic probing itself). To do this, we introduce what we call \emph{smoothed quadratic probing}, a variant of quadratic probing that ends up being more amenable to analysis. In smoothed quadratic probing, each element $x$ is assigned an \emph{offset sequence} $0  < r_1(x) < r_2(x) < \cdots$ where each $j \in \mathbb{N}_+$ is included independently with probability $1/\sqrt{j}$. Queries are then performed as in standard quadratic probing, but where each element $x$ follows its \emph{own} (approximately) quadratic sequence, examining positions $h(x),  h(x) + r_1(x), h(x) + r_2(x), \ldots \pmod n$. 

Note that smoothed quadratic probing can be viewed in one of two ways. On one hand, it offers a \emph{smoothed analysis} of quadratic probing -- allowing us to reason about how we think quadratic probing should behave, but through the lens of a more accommodating data structure. On the other hand, it is on its own a completely valid hash table, which can in principle be used in place of quadratic probing.

The main purpose of this paper is to revisit questions (1) to (3) above, but in the context of smoothed quadratic probing. Our analyses lead us to several unexpected answers (of potential interest to both theoreticians and practitioners alike). And the techniques that we introduce end up being applicable even beyond the smoothed version of the data structure -- allowing us to prove in our final result that there \emph{exist} concrete (non-smoothed) variations of quadratic probing with provable guarantees.

\paragraph{The fall of Robin Hood.}
If the average successful query time for quadratic probing has expectation $O(\log ε^{-1})$ as conjectured, then we would expect the Robin Hood policy to distribute the query costs evenly so all keys in the table have expected query time $O(\log ε^{-1})$.
Our first result says that this is \emph{not} what we get in the smoothed case.
While the expected successful query times are equalized, they become \emph{equally bad} at $Θ(ε^{-1/2})$ (and amortized expected insertion times rise to a corresponding $Θ(ε^{-3/2})$). This suggests a positive answer to Question 1 but a negative answer to Question 2.
The intuition is that the Robin Hood policy, by pressuring elements to be placed \emph{closer} to their home position, partially brings back the clustering effects of linear probing that quadratic probing was designed to overcome. In a sense the hash table behaves like a hybrid between linear and quadratic probing, exhibiting some but not all of the local clustering. 

Our analysis extends to smoothed degree-$d$ probing for any $d \ge 1$ (which samples offsets so that each element's offset sequence grows approximately as a degree-$d$ polynomial). We prove that, in general, using the Robin Hood ordering for such a hash table results in expected query time $\Theta(\epsilon^{-1/d})$ and amortized expected insertion time $\Theta(\epsilon^{-(1 + 1/d)})$. Note that, as $d$ grows, this gets closer and closer to the $O(\epsilon^{-1})$ bound one might hope for, but never reaches it.


\paragraph{What about anti-Robin Hood?} Given the poor performance of Robin Hood ordering, the next question we consider is whether a better ordering might actually be the \emph{anti-Robin Hood} ordering -- this is the ``poor get poorer'' ordering, in which we give priority to elements that have traveled \emph{less} far. 

Somewhat unexpectedly, this ordering turns out to perform much better. We prove that smoothed quadratic probing, with the anti-Robin Hood ordering, achieves expected query time $O(\log \epsilon^{-1})$ and an amortized expected insertion time $O(\epsilon^{-1})$. 

Here, again, we are able to extend the tradeoff curve to smoothed degree-$d$ probing for any $d \ge 1$. The expected query time becomes 
$O(\max(\log \epsilon^{-1}, \epsilon^{1 - 2/d}))$,
gracefully decreasing from $\epsilon^{-1}$ at $d = 1$ to $O(\log \epsilon^{-1})$ for all $d \ge 2$; likewise, the amortized expected insertion time becomes
$O(\max(\epsilon^{-1}, \epsilon^{-2/d}))$,
decreasing from $\epsilon^{-2}$ at $d = 1$ to $O(\epsilon^{-1})$ for all $d \ge 2$. 
We prove that these tradeoffs are tight for $d \le 2$, which means that $d = 2$ is, in fact, the smallest value of $d$ that is able to fully escape the clustering effects of linear probing. 

Combined, our results give a complete answer to Question (3) for both the Robin Hood and anti-Robin Hood orderings of smoothed quadratic probing. 

\paragraph{The best scheme so far.}
Even without resolving the behavior of the standard, non-smoothed quadratic-probing sequence itself, our smoothed analysis already gives a practice-relevant message: it distinguishes between natural design choices and points to the most promising one.
As an intermediate summary, among the degree-$d$ probing schemes discussed so far, smoothed quadratic anti-Robin Hood probing is the best in terms of query time:
\begin{itemize}
  • \emph{Stable} degree-$d$ probing (that does not move elements after they are inserted) takes $Ω(ε^{-1})$ time for negative queries.
  • Among the considered ordered probing schemes, anti-Robin Hood performs better than Robin Hood.
  • Anti-Robin Hood probing requires $d ≥ 2$ to avoid clustering effects and achieve query time $O(\log ε^{-1})$.
  • Among the schemes with $d ≥ 2$, quadratic probing, i.e.\ $d = 2$, is the most local, and thus supposedly the most cache friendly one.
\end{itemize}
This motivates focusing on this scheme to improve it further.

\paragraph{Beyond smoothed analysis: fixed-offset quadratic probing via the probabilistic method.} Finally, the main result of the paper is a strengthening of our anti-Robin Hood analysis to go beyond a smoothed analysis. Define a \defn{fixed-offset degree-$d$ hash table} to be a hash table in which each element follows the probe sequence
$$h(x), h(x) + r_1, h(x) + r_2, h(x) + r_3, \ldots \pmod n,$$
where the sequence $r_1, r_2, r_3, \ldots$ is the same for every element (it is not smoothed), and where $r_i = \Theta(i^d)$. Our final result proves, via the probabilistic method, that for any $d \ge 2$, there \emph{exists} a fixed-offset degree-$d$ hash table that, under the anti-Robin Hood ordering, supports queries in expected time $O(\log \epsilon^{-1})$ and insertions in amortized expected time $O(\epsilon^{-1})$. 

In fact, we prove a stronger result: that if one selects a fixed-offset degree-$d$ hash table at \emph{random} -- meaning that the offset sequence $r_1, r_2, r_3, \ldots$ is obtained by selecting each $r_i$ independently and uniformly at random from some appropriate interval -- then with probability arbitrarily close to $1$ (probability $1 - \delta$ for a positive $\delta$ of our choice), the resulting hash table supports queries in expected time $O(\log \epsilon^{-1})$ and insertions in amortized expected time $O(\epsilon^{-1})$. 

We also prove that query time $O(\log \epsilon^{-1})$ is optimal, not just for the specific hash tables we consider, but for any fixed-offset open-addressed hash table using the anti-Robin Hood ordering. Thus, we can conclude that, for tables constructed with the anti-Robin Hood ordering, almost all fixed-offset degree-$2$ (or higher) hash tables are optimal.

This puts us in the following intriguing situation. We can now prove that essentially all quadratic-probing-\emph{like} hash tables support (when implemented with the anti-Robin Hood ordering) the ideal query and insertion time that one would hope for. But, we still cannot say anything about quadratic probing itself. All we know is that, \emph{if quadratic probing fails to achieve these time bounds, it is because the specific offset sequence $0, 1^2, 2^2, \ldots$ behaves highly unlike most degree-$2$ sequences}. Whether or not this is the case remains open.

\begin{table}
  \caption{Summary of results on linear and quadratic probing hash tables using Robin Hood (RH), anti-Robin Hood (ARH) or default insertion policies.}
  \label{tab:results}
  \begin{tabular}{lrlll}
    \toprule
    probing scheme
    & policy 
    & query time 
    & insertion time$\ddagger$
    & reference
    \\
    \midrule
    
    linear
    & default
    & $Θ(ε^{-1})\dagger$
    & $Θ(ε^{-2})$
    & \cite{knuth1962notes}
    \\
    
    linear
    & RH
    & $Θ(ε^{-1})$
    & $Θ(ε^{-2})$
    & \cite{CelisLaMu85,celis1986robin} 
    \\
    
    quadratic
    & default
    & \textcolor{gray}{$Θ(\log ε^{-1})\dagger$} \footnotesize (conjectured)
    & \textcolor{gray}{$Θ(ε^{-1})$} \footnotesize (conjectured)
    & \cite{maurer1968programming}
    \\
    
    fixed-offset degree-$(d{≥}2)$
    & ARH
    & $Θ(\log ε^{-1})$
    & $Θ(ε^{-1})$
    & Thm.\,\ref{thm:anti-robinhood-expected-bounds}, Thm.\,\ref{thm:anti-rh-harmonic-probe-tail}
    \\
    
    smoothed degree-$d$
    & ARH
    & $O(\max(\log ε^{-1},ε^{1-2/d}))$
    & $O(\max(ε^{-1},ε^{-2/d}))$
    & Cor.\,\ref{cor:smoothed-arh-upper-bounds}
    \\
    
    $↪$ smoothed quadratic
    & ARH
    & $O(\log ε^{-1})$
    & $O(ε^{-1})$
    & \footnotesize (special case)
    \\
    
    smoothed degree-$d$
    & RH
    & $Θ(ε^{-1/d})$
    & $Θ(ε^{-1-1/d})$
    & Thm.\,\ref{thm:rh-expected-bounds}, Cor.\,\ref{cor:rh-displacement-lower}
    \\
    
    $↪$ smoothed quadratic
    & RH
    & $Θ(ε^{-1/2})$
    & $Θ(ε^{-3/2})$
    & \footnotesize (special case)
    \\
    \bottomrule
  \end{tabular}\\
  \footnotesize $\dagger$ For the “default” insertion policy we report average expected query times of all keys stored in a table of load factor $1-ε$. Keys that are inserted “late” or not at all may have larger query times.\\
  $\ddagger$ Amortized expected insertion time at load factor $1-ε$. The “amortized” qualifier is needed when we derive insertion times from query times by differentiating, see \cref{sec:prelim}.
\end{table}

\subsection{Other related work}
In the decades since quadratic probing was first introduced, the study of hash tables has grown to be a major subfield of data structures. In addition to research on quadratic probing, there have also been long lines of work on analyzing double hashing \cite{mitzenmacher2012peeling, guibas1976analysis, guibas1978analysis, lueker1988more, lueker1993more}, variations of cuckoo hashing \cite{pagh2004cuckoo,DW07,FriezePetti,FPSS05,polylogAlan,FPS13,wearmin,FriezeJohansson,newinsertion,Walzer,bell20241,KM25}, and other open-addressed hash-tables \cite{brent1973reducing,celis1986robin,yao1985uniform,munro1986techniques,bender2022linear,farach2024optimal,bender2024tight,braverman2024tight}. There has also been extensive work on compact and succinct dictionaries \cite{pagh1999low,raman2003succinct,arbitman2010backyard,bender2023iceberg,bender2024modern,li2023tight,li2024dynamic,liu2020succinct}, hash tables with high-probability guarantees \cite{stashcuckoo,deamortizedcuckoo,dietzfelbinger1994dynamic,goodrich2012cache,kuszmaul2022hash}, non-oblivious open-addressing \cite{fiat1988nonoblivious,fiat1993implicit,bender2025optimal}, etc. 

Besides quadratic probing, another approach that has been used to eliminate linear probing's clustering effects is \emph{graveyard hashing} \cite{bender2022linear,braverman2024tight}, a variation of linear probing that strategically rearranges elements over time in order to keep the expected insertion time at $O(\epsilon^{-1})$. This approach improves insertions, but cannot do anything for queries (indeed, \emph{all} orderings of a linear-probing hash table have expected query time at least $\Omega(\epsilon^{-1})$ \cite{peterson1957addressing,bender2022linear}).

An interesting feature of our work is that it takes a very different approach from the previous state of the art for analyzing quadratic probing \cite{kuszmaul2024towards} (which was able to analyze unordered quadratic probing up to load $\approx 0.089$). Whereas the previous analysis was witness-tree based, the analyses in this paper take a much more stochastic-processes flavor. We are optimistic that this may open the door to further progress on the analysis of (non-smoothed) quadratic probing in the future. 

The results in this paper (and in most work on open addressing) assume access to fully random hash functions. In recent decades, there has also been a long line of work on techniques for simulating such hash functions efficiently. This includes constructions that offer high independence without compromising evaluation time \cite{Siegel1989FastHighPerformanceHash,Siegel1995TR684ExtremelyRandom,PaghPagh,Siegel2004ExtremelyRandomConstantTime,OstlinPagh,dahlgaard2015hashing,thorup2013simple,bercea2023locally}, and techniqeus for applying such hash functions to hash tables (see, e.g., the splitting trick in \cite{splitting}).

\subsection{Outline}
In the body of the paper, we present our results in essentially the opposite order of how they are described above (this is to put the results which are likely of most interest first). 

We begin in Section \ref{sec:anti-robinhood} by presenting our results on anti-Robin Hood hashing. To streamline the argument, we first prove a result that is slightly different from any of those described above -- we show that for any $d \ge 1$ and fixed $\epsilon > 0$, a random degree-$d$ probing hash table (so fixed-offset rather than smoothed) achieves expected query time $O(\max(\log \epsilon^{-1}, \epsilon^{1 - 2/d}))$ (here the randomness is partly over the random offset sequence itself). This result is interesting in its own right, and sets up all the machinery we need to prove the two results discussed above in the introduction: (1) an analysis of smoothed quadratic probing for any $d \ge 1$; and (2) a proof for $d > 2$ that almost all (random) fixed-offset degree-$d$ hash tables achieve the desired time bounds (for all $\epsilon$ and large enough $n$).

Next, in Section \ref{sec:robinhood}, we present our results on smoothed degree-$d$ probing using Robin Hood ordering, proving, for any $d \ge 1$, expected query time $O(\epsilon^{-1/d})$ and amortized expected insertion time $O(\epsilon^{-(1 + 1/d)})$.

Finally, in Section \ref{sec:lower-bounds}, we prove lower bounds for both Robin Hood and anti-Robin Hood probing schemes. An interesting aspect of (several) of the lower bounds is that they actually use the \emph{upper bounds} as a starting point -- we are able to argue that if a polynomial-probing hash table achieves sufficiently good tail upper bounds, then it is \emph{forced} to incur certain expected-time lower bounds. This basic technique may also be of independent interest.

\section{Preliminaries}
\label{sec:prelim}

An \defn{open-addressing hash table} stores a set $S \subseteq [U]$ of keys in an array of $n$ slots, indexed by the residues in $[n]$.  Each key $x\in[U]$ has a \defn{probe sequence}
\[
    h_0(x),h_1(x),h_2(x),\ldots \in [n]
\]
and insertions and queries inspect these slots in order.

Throughout the paper, we will be interested in probe sequences of the form
\[
    h_i(x)=h(x)+r_i(x) \pmod n,
\]
where $h:[U] \to [n]$ is the \defn{base hash function} (assumed to be uniformly random and fully independent), and where $0 = r_0(x) < r_1(x) < r_2(x) < \cdots$ is the \defn{offset sequence} used for key $x$. For random fixed-offset hash tables (which we will define shortly), we will use $0 = r_0 < r_1 < r_2 < \cdots$ to refer to the shared offsets.


If key $x$ is stored in the slot $h_i(x)=h(x)+r_i(x) \pmod n$ after examining the probes $h_0(x), h_1(x), \ldots, h_i(x)$, we say that $x$ has \defn{displacement} $r_i(x)$ and \defn{probe complexity} $i$, denoted by $\disp(x)$ and $\pc(x)$, respectively.

\paragraph{The Robin Hood and anti-Robin Hood orderings.}
We study two priority rules for resolving collisions during insertion, the \defn{Robin Hood ordering} and the \defn{anti-Robin Hood ordering}.  Let
\[
    \gamma:[U]\to[0,1]
\]
be an independent continuous tie-breaking hash function.  Suppose keys $x$ and $y$ both probe the same slot $a$, with $x$ reaching $a$ at displacement $\Delta_x$ and $y$ reaching $a$ at displacement $\Delta_y$.  Under Robin Hood ordering, the key with larger displacement has higher priority at $a$; under anti-Robin Hood ordering, the key with smaller displacement has higher priority.  Ties are broken by the larger value of $\gamma$. 


The insertion algorithm is the same for both orderings once the priority rule is fixed.  To insert a key $x$, scan forward in the probe sequence of the current key.  If the next slot is empty, place the current key there and stop.  If the slot contains a key $y$ with higher priority at that slot, leave $y$ in place and continue scanning for the current key.  If the current key has higher priority than $y$, place the current key in the slot, evict $y$, and continue the insertion with $y$ starting from its next probe.\footnote{\label{fn:displacement-ambiguity}Strictly speaking, in the current definition of displacement, it is not straightforward to compute the displacement of $y$, which is required for comparing priorities -- given the current position of $y$ and its hash, there might still be multiple offsets $r_i, r_j$ such that $r_i = r_j \pmod n$, thus the displacement of $y$ cannot be solely determined by its position and hash. Indeed, to determine the displacement, we need to simulate the entire operation sequence to know which offset $y$ has examined. We nevertheless use this definition because it is convenient for the analysis.  In an implementation, one could instead define the displacement of a key $y$ placed at slot $a$ as $(a-h(y) \bmod n)$.  Our analysis shows that, with high probability, no key ever uses an offset $r_i>n$, and hence the two definitions agree throughout the execution.}
The process stops when some current key reaches an empty slot.

An important property of both the Robin Hood and anti-Robin Hood orderings is that the resulting hash tables are necessarily history independent. This is captured by the following lemma, which is proven formally in Section \ref{app:reductions}.

\begin{restatable}[History Independence of the Two Orderings]{lemma}{historyindependence}
  \label{lem:history-independence-of-rh-and-arh}
  For any specific probing scheme, the hash table obtained from either Robin Hood or anti-Robin Hood ordering is history independent: given the set of inserted keys, the final table state is independent of the order in which the keys were inserted.
\end{restatable}

Note that this insertion process maintains what we call the \defn{search invariant}: For any key $x$ that is in the hash table, if it is using its $j$-th probe $h_j(x)$, then all previous slots $h_1(x), h_2(x), \ldots, h_{j - 1}(x)$ in the probe sequence are occupied by keys $x' \neq x$ which have priority over $x$ at that slot.\footnote{Note that, by Lemma \ref{lem:history-independence-of-rh-and-arh}, it suffices to verify this invariant under the assumption that $x$ is the most-recently-inserted key, at which point the invariant is immediate from the insertion algorithm.}

Because of the search invariant, membership queries can be performed as follows. A membership query for a key $x$ scans the probe sequence of $x$ until it finds $x$, in which case the query returns true; or until it reaches a slot $j$ that would accept $x$ over its current contents, in which case the query returns false. More concretely, an unsuccessful query (i.e., a query to a key not present) stops at the first empty slot or at the first occupied slot whose incumbent has lower priority than $x$ at that slot. This behavior is guaranteed to be correct by the search invariant.

\paragraph{Query time versus insertion time. }Most of the analyses in this paper will focus on analyzing the time to perform a membership query on an element $x \in [U]$ at a load factor $1 - \epsilon$. 

It is therefore worth noting that any bound on query time immediately implies a bound on \emph{amortized} insertion time. Formally, we say that a hash table achieves \defn{amortized expected insertion time} $f(\epsilon)$ at load factor $1 - \epsilon$ if, for any $\delta$, the \emph{total} time taken by the first $(1 - \delta)n$ insertions is at most 
$$\sum_{i = 0}^{(1 - \delta)n - 1} f(1 - i/n)$$
in expectation. 
By the insertion algorithm, the total time of the first $(1 - \delta)n$ insertions is equal to the total query time of these $(1 - \delta)n$ inserted keys after the insertions.\footnote{For Robin Hood ordering, the more precise quantity is the \emph{probe complexity} of an existing key rather than its \emph{query time}, since a key may probe the same slot multiple times before it has high enough priority to occupy that slot.  Our Robin Hood bounds are proved for this probe-complexity notion.} Therefore, any worst-case expected query time bound of the form $g(\epsilon)$, where $g$ is a differentiable function, immediately implies an amortized expected insertion time bound of \begin{equation}f(\epsilon) := -\frac{d}{d\epsilon} g(\epsilon).
\label{eq:amortized}
\end{equation}
All of our results on amortized insertion time will follow directly from \eqref{eq:amortized}. Therefore, in our proofs, we will typically leave the amortized insertion analysis as an implicit consequence of the query-time bounds. 

\paragraph{Smoothed and randomized variations of polynomial probing.}
Let $d\ge1$ be fixed.  The deterministic \defn{degree-$d$ probing} sequence uses offsets
\[
    0,\floor{1^d},\floor{2^d},\ldots .
\]
Thus $d=1$ gives linear probing and $d=2$ gives quadratic probing.  Since tight query-time bounds for deterministic degree-$d$ probing remain open, this paper focuses on two randomized variants: \defn{smoothed degree-$d$ probing} and \defn{random fixed-offset degree-$d$ probing}.

\begin{definition}[Smoothed degree-$d$ probing hash table]
  In smoothed degree-$d$ probing, each key has its own random offset sequence.  The home offset $0$ is always included, and each positive offset $k\ge1$ is included independently with activation probability
  \[
      p(k)=k^{-(d-1)/d}.
  \]
  We say that the slot $h(x)+k$ is \defn{activated} for key $x$ when offset $k$ is included.
\end{definition}

\begin{definition}[Random Fixed-Offset Degree-$d$ Probing Hash Table]
  Let $\alpha>0$ be a fixed constant.  In random fixed-offset degree-$d$ probing, all keys share one random offset sequence.  For each $j\ge1$, the offset $r_j$ is sampled independently and uniformly from the \defn{sampling interval}
  \[
      I_j \defeq [\alpha j^d,\alpha(j+1)^d),
  \]
  and every key probes
  \[
      h(x),\ h(x)+r_1,\ h(x)+r_2,\ldots
      \pmod n .
  \]
\end{definition}

By definition, in a random fixed-offset degree-$d$ table, the probe complexity $\pc(x)$ and displacement $\disp(x)$ are related by $\pc(x) = \Theta(\disp(x)^{1/d} + 1)$. 
When discussing the smoothed variant, it will often be helpful to reference the following lemma (proven in Appendix \ref{app:reductions}), which gives an analogous relationship for smoothed degree-$d$ probing. 

\begin{restatable}[Displacement and Probe Complexity]{lemma}{displacementprobecomplexity}
  \label{lem:displacement-and-probe-complexity}
  Fix a constant $d\ge1$.  For any fixed key $x$ in a smoothed degree-$d$ probing hash table with either Robin Hood or anti-Robin Hood ordering,
  \[
    \E\Bk*{\pc(x)} = \Theta\bk*{\E\Bk*{\disp(x)^{1/d}} + 1}.
  \]
  The hidden constants depend only on $d$.
\end{restatable}

\section{Upper Bounds for Anti-Robin Hood Ordering}
\label{sec:anti-robinhood}


In this section, we derive our upper bounds for anti-Robin Hood hashing. The majority of the section (Subsections \ref{subsec:anti-robinhood-intuition}, \ref{subsec:anti-robinhood-concentration}, \ref{subsec:anti-robinhood-negative-feedback-loop}, \ref{subsec:anti-robinhood-long-displacements}, \ref{subsec:anti-robinhood-proof-of-theorem}) is spent developing the necessary technical machinery to prove the following theorem, which allows us to reason about the behavior of a random fixed-offset degree-$d$ probing hash table, using the anti-Robin Hood ordering, at any load factor $1 - \epsilon$:

\begin{restatable}{theorem}{antirobinhoodexpectedbounds}
    \label{thm:anti-robinhood-expected-bounds}
    Let $d\ge1$ be fixed.  There is a constant $C_{\mathrm{fin}}$ such that the following holds.  Suppose
    \[
        C_{\mathrm{fin}}\log n\cdot \max\{n^{-1/2},n^{-1/d}\}
        \le
        \epsilon
        \le
        \frac12.
    \]
    For a random fixed-offset degree-$d$ probing hash table using anti-Robin Hood ordering at load factor $1-\epsilon$, the expected time to perform a successful or unsuccessful query is
    \begin{align*}
      \begin{cases}
        O\bk*{\epsilon^{1-2/d}} & \text{if }1\le d<2, \\
        O\bk*{\log \epsilon^{-1}} & \text{if }d\ge2.
      \end{cases}
    \end{align*}
\end{restatable}

  Taking the derivative of the expected query time with respect to $\epsilon$, we also get the following immediate corollary for amortized expected insertion time:
  \begin{corollary}
    \label{cor:anti-robinhood-amortized-expected-insertion-time}
    With the same setup as in Theorem \ref{thm:anti-robinhood-expected-bounds}, the amortized expected insertion time at load factor $1 - \epsilon$ is
    \begin{align*}
      \begin{cases}
        O\bk*{\epsilon^{-2/d}} & \text{if }1\le d<2, \\
        O\bk*{\epsilon^{-1}} & \text{if }d\ge2.
      \end{cases}
    \end{align*}
  \end{corollary}

  Note that Theorem \ref{thm:anti-robinhood-expected-bounds} focuses only on obtaining \emph{upper bounds}. Later in the paper, however, we will also prove matching lower bounds, showing that Theorem \ref{thm:anti-robinhood-expected-bounds} is asymptotically tight for any $d \ge 1$ (Section \ref{sec:lower-bounds}). 

  After proving Theorem \ref{thm:anti-robinhood-expected-bounds}, we then set our sights on proving a more powerful result: that, when $d \ge 2$, there exists a \emph{single} fixed-offset degree-$d$ probe sequence that works for \emph{all} $\epsilon$ and $n$ simultaneously. This allows us to, in Subsection \ref{sec:one-sequence-all-eps}, establish the main technical result of the paper: 

  \begin{restatable}[One offset sequence for all  $\epsilon$ and $n$]{theorem}{onesequencealleps}
    \label{thm:one-sequence-all-eps}
    Fix a constant $d\ge2$ and $\delta\in(0,1)$.  There are constants
    \[
        C_{\mathrm{sim}}=C_{\mathrm{sim}}(\delta,d)<\infty
        \qquad\text{and}\qquad
        C_*=C_*(\delta,d)<\infty
    \]
    such that, for a random fixed-offset degree-$d$ offset sequence $r_1, r_2, \ldots$, we have with probability $1 - \delta$ (over the randomness of $r_1, r_2, \ldots$) that the following holds: for every table size $n$ and every $\epsilon$ satisfying
    \begin{equation}
        C_{\mathrm{sim}}\log n\cdot\max\{n^{-1/2},n^{-1/d}\}
        \le
        \epsilon
        \le
        \frac12,
        \label{eq:one-sequence-all-eps-epsilon-condition}
    \end{equation}
    the worst-case expected query time (either successful or unsuccessful) at load factor $1 - \epsilon$ in an $n$-slot hash table using the anti-Robin Hood ordering with offset sequence $r_1, r_2, \ldots$ is at most
    \[
        C_*\log\epsilon^{-1}.
    \]
  \end{restatable}

  As an immediate corollary, we get:
  
  \begin{corollary}
    \label{cor:one-sequence-amortized-expected-insertion-time}
    With the same setup as in Theorem \ref{thm:one-sequence-all-eps}, we have with probability $1 - \delta$ (over the randomness of $r_1, r_2, \ldots$) that, for all $\epsilon$ and $n$ satisfying \eqref{eq:one-sequence-all-eps-epsilon-condition}, the amortized expected insertion time at load factor $1 - \epsilon$ in an $n$-slot hash table using the anti-Robin Hood ordering with offset sequence $r_1, r_2, \ldots$ is at most $O(\epsilon^{-1})$ (where the hidden constant depends only on $d, \delta$).
  \end{corollary}

  Finally, after proving Theorem \ref{thm:one-sequence-all-eps}, we also establish for completeness a second (more straightforward) extension of Theorem \ref{thm:anti-robinhood-expected-bounds}, which is to show that Theorem \ref{thm:anti-robinhood-expected-bounds} also holds for \emph{smoothed} degree-$d$ probing (as a substitute for \emph{random} fixed-offset degree-$d$ probing). This is established in Subsection \ref{subsec:anti-robinhood-smoothed-analysis}.

\subsection{Intuition of the Proof}
\label{subsec:anti-robinhood-intuition}
In this subsection, we discuss the high-level ideas behind our analyses. We will focus for most of the section on analyzing smoothed quadratic probing, meaning that each key $x$ includes each offset $j$ in its probe sequence independently with probability $j^{-1/2}$. We will then briefly motivate some of the additional ideas needed to prove Theorems \ref{thm:anti-robinhood-expected-bounds} and \ref{thm:one-sequence-all-eps}.

\paragraph{Observation 1: local intervals behave in concentrated ways.} 
Define the $K$-reveal of the hash table to be what we get if we remove from the hash table any keys with displacements greater than $K$. An interesting feature of anti-Robin Hood ordering is that, for any threshold $K$ and any interval $I$ of $K$ slots, the state of $I$ in the $K$-reveal depends \emph{only} on the keys that hash to $I \cup (I - K)$. 

Moreover, many of the properties of $I$ that we might care about -- such as the number of free slots in $I$, or the number of keys that hash to $I$ but have displacement \emph{greater} than $K$ -- have low sensitivity to the addition or removal of a single key. 

So, for example, if we consider the number of free slots $F_I$ in $I$, we can argue both that:
\begin{enumerate}
    \item $F_I$ is determined by a small set of independent random variables, namely the hashes and probing choices of the keys with home positions in $I \cup (I - K)$.
    \item The addition or removal of a single key can change $F_I$ by at most $1$.
\end{enumerate}

This pair of properties is enough to obtain very strong concentration bounds on $F_I$. Even without knowing what $\E\Bk*{F_I}$ is, we can apply (an extension of) McDiarmid's inequality to argue that
$$\Pr[|F_I - \E\Bk*{F_I}| > u] \le 2 \exp\left(-\Omega\left(\frac{u^2}{K}\right)\right).$$
Likewise, a similar argument can be made for the number $O_I$ of keys that hash to $I$ but have displacement greater than $K$.

These concentration bounds allow us to dodge many of the independence issues that would normally arise in the analysis. Roughly speaking, they allow us to treat every $K$-length interval $I$ as behaving in the same way as any other, regardless of interactions between different intervals and keys. 

\paragraph{Observation 2: long displacements are rare.}
Now, let us turn our attention to the following basic question: Given a key $x$, how likely is it to have some large displacement $K$? We first observe that displacements $K \gg \epsilon^{-2}$ will be extremely rare. This is due to the following argument. 

Let $K \gg \epsilon^{-2}$, and consider the interval $I = [h(x), h(x) + K)$. By symmetry across the hash table, the expected number of free slots in $I$ is $\epsilon K$. Because $\epsilon K \gg \sqrt{K}$, our concentration bound from earlier tells us that, in the $K$-reveal, $I$ has at least $\Omega(\epsilon K)$ free slots. 

What is the probability, then, that key $x$ manages to make it through the entire interval $I$ without finding a free slot. By construction, $x$ includes every position in $I$ with probability at least $1 / \sqrt{K}$. The expected number of free slots that $x$ probes is therefore 
$$\approx \frac{1}{\sqrt{K}} \cdot \epsilon K \gg 1,$$
since $K \gg \epsilon^{-2}$. The probability that $x$ manages to ``miss all of these opportunities'' is very small. Thus, $x$ is very unlikely to have any displacement greater than $K$.

\paragraph{Observation 3: even \emph{short} displacements are surprisingly rare.}
The more interesting question is: What is the probability $P_K$ that a key $x$ has displacement at least $K$ for some $K \le \epsilon^{-2}$? 

Here, we are able to make an argument that is completely agnostic to the load factor $1 - \epsilon$. We argue that, even if $\epsilon = 0$, we have $P_K = O(1/\sqrt{K})$.

The main technical insight here is that there is a negative feedback loop: if the probability $P_K$ is very large, say $100 / \sqrt{K}$, then we can argue that $P_{2K} \le \frac{1}{10} P_K$. To see why this is the case, consider the following argument. 

Suppose $P_K \ge 100 / \sqrt{K}$, let $I_1$ be the interval $[h(x), h(x) + K)$, and let $I_2$ be the interval $[h(x) + K, h(x) + 2K)$. Notice that, for every element $x$ that does not get placed in the $K$-reveal, there is a corresponding free slot that must exist somewhere in the $K$-reveal. This means, critically, that in the $K$-reveal, the expected number of free slots in $I_2$ is at least $P_K \cdot K \ge 100 \sqrt{K}$.

How many of these $\approx 100 \sqrt{K}$ slots do we expect to still be free in the $2K$-reveal? Well, if $P_{2K} \ge \frac{1}{10} P_K \ge 10 / \sqrt{K}$, then we expect at least $10\sqrt{K}$ free slots to remain in $I_2$ in the $2K$-reveal. Since the number of such free slots is concentrated around its mean (as we argued earlier), let's assume for simplicity that there are exactly $10\sqrt{K}$ such free slots. 

Then, how likely is it that element $x$ managed to make it through the entire interval $I_2$ without finding a free slot? Here, we can make a similar argument to the one we used for large intervals. The element $x$ includes every position in $I_2$ with probability at least $1 / \sqrt{2K}$. The expected number of free slots that $x$ probes is therefore $$\frac{1}{\sqrt{2K}} \cdot 10\sqrt{K} \ge 5.$$ The probability that $x$ manages to \emph{miss} all of these opportunities is small, at most, say, $e^{-5} \le 1/10$. But this means that $P_{2K} \le \frac{1}{10} P_K$, after all. 

Thus, there is a negative feedback loop: if $P_K$ is very large, then there are many free slots in $I_2$ for the element $x$ to use, and the probability $P_{2K}$ of $x$ making it through $I_2$ without finding a free slot is very small. This negative feedback loop caps $P_K$ to $O(1/ \sqrt{K})$. 

\paragraph{Putting the pieces together: query time $O(\log \epsilon^{-1})$.}
The bound $P_K=O(1/\sqrt K)$ is exactly what is needed for query time.  Group displacements into dyadic ranges $[K,2K]$.  Conditional on the query reaching such a range, smoothed quadratic probing makes $O(\sqrt K)$ probes in that range in expectation, while the probability of reaching the range is at most $P_K$.  Hence the contribution of one dyadic range is
\[
    P_K\cdot O(\sqrt K)=O(1).
\]
There are only $O(\log \epsilon^{-1})$ dyadic ranges below $\poly(\epsilon^{-1})$, and displacements $K\gg\epsilon^{-2}$ are extremely rare.  Thus the full expected query time is $O(\log \epsilon^{-1})$.

The same basic argument also extends to smoothed degree-$d$ probing for any $d > 2$. The main difference is what happens in short intervals. Now, the negative feedback loop keeps $P_K$ not at $O(1 / \sqrt{K})$, but at $O(1 / K^{1/d})$. This is because, within interval $I_2$ (as defined earlier), the element $x$ expects to make $\Theta(K^{1/d})$ probes, so the critical threshold at which 
$$(\# \text{ probes}) \cdot (\text{fraction of free slots}) = \Theta(K^{1/d}) \cdot P_K \gg 1$$
happens at $P_K = O(1/K^{1/d})$. On the other hand, the \emph{query cost} of $x$, should it travel distance $K$, is only $O(K^{1/d})$. So the expected query time spent on displacements in $[K, 2K]$ continues to be $O(1)$. 

When $d < 2$, however, a different phenomenon occurs. The noise from our concentration bounds dominates any negative feedback loop. Thus, even though the negative-feedback loop argument from above would seem to indicate that $P_K = O(1/K^{1/d}) = o(1/\sqrt{K})$, we end up with $P_K = \Theta(1/ \sqrt{K})$ anyway. This issue turns out to be fundamental (we \emph{really do} have $P_K = \Theta(1/ \sqrt{K})$), and serves as the bottleneck all the way up to $K \approx \epsilon^{-2}$. The fact that a $P_{\epsilon^{-2}} = \Theta(\epsilon)$ fraction of keys have displacements greater than $\epsilon^{-2}$ leads to an overall query time of $\Theta(\epsilon^{1 -2/d})$. This query-time bound smoothly (and correctly) degrades from the classical $\Theta(\epsilon^{-1})$ bound for linear probing to the bound of $\epsilon^{-o(1)}$ for $d = 2$. 

\paragraph{Beyond smoothed analysis.}
Finally, one of the most interesting features of our analysis is that we can extend it to obtain an even stronger result: most sigificantly that, for any $d \ge 2$, there exists a \emph{concrete random fixed-offset degree-$d$} probe sequence that makes probes at the same asymptotic rate as degree-$d$ probing, and that achieves $O(\log \epsilon^{-1})$ expected query time for any element at any load factor $1 - \epsilon$ (Theorem \ref{thm:one-sequence-all-eps}). 

This extension is achieved by a probabilistic-method style argument. We construct the probe sequence $h(x), h(x) + r_1, h(x) + r_2, \ldots$ (where the $r_i$s are the same for all elements $x$) by constructing the sequence $r_1, r_2, \ldots$ piece by piece. We argue that, no matter how poorly some prefix of $r_1, r_2, \ldots$ performs, subsequent probes will still behave nicely, either enforcing the negative feedback loop from above (for small displacements) or ensuring that almost all elements get placed (for large displacements). We can then argue that, with probability close to $1$, every load factor $1 - \epsilon$ gets an expected query time (conditioned on the probe sequence above) that is $O(\log \epsilon^{-1})$.

\paragraph{Other technical considerations.}
We remark that, in the above discussion, we have focused on high-level analysis while omitting a number of technical difficulties. Some care is needed to handle the fact that the concentration bounds discussed above do not hold with probability $1$ (in fact, we often care about parameter regimes where they fail with constant probability), as well as the fact that different intervals $I_1, I_2$ do not behave independently of one another (we can mostly navigate this by Poissonizing the key set and focusing on intervals that have a healthy gap between them). Additionally, for $d < 2$, we will need to be careful in how we handle the transition from the small-interval parameter regime to the large-interval parameter regime, as a naive argument loses a super-constant factor in the query time here. 

Additionally, in the setting where we use a fixed offset sequence $r_1, r_2, \ldots$ for every key (rather than a different smoothed offset sequence for each key), there are several additional factors that require consideration. We must be careful to handle interactions between keys that share a home hash $h(x_1) = h(x_2)$. Such keys experience \emph{secondary-clustering} effects due to the fact that they share their entire probe sequences. Additionally, we must be careful to reason about query time in a way that allows us to, ultimately, obtain bounds for \emph{all $\epsilon$ simultaneously}. This requires a more intricate analysis than, say, simply applying a union bound over all values of $\epsilon$. 

Nonetheless, with careful handling of these additional factors, we will be able to, in the sections that follow, turn the high-level ideas above into formal analyses. 

\subsection{Concentration Bounds for Local Intervals}
\label{subsec:anti-robinhood-concentration}
Following the intuition in \cref{subsec:anti-robinhood-intuition}, we now start the proof of \cref{thm:anti-robinhood-expected-bounds} by establishing concentration bounds for local intervals.
We first work in a Poissonized model containing
\[
    \operatorname{Poisson}((1-\epsilon)n)
\]
keys.  We transfer the bounds to the fixed-load model at the end.

Given a table state, the \defn{$K$-reveal} is the partial table obtained by deleting every key whose displacement is greater than $K$.
The concentration proposition below works for any anti-Robin Hood hash table with an arbitrary fixed offset sequence.

\begin{proposition}
\label{prop:concentration}
Fix a offset sequence $0 = r_0 < r_1<r_2<\cdots$. Consider an Anti-Robin Hood hash table with this shared offset sequence and containing $\operatorname{Poisson}((1-\epsilon)n)$ keys.  Let $A$ be an interval of at most $K$ slots, where $K\le n/10$.  In the $K$-reveal, define
\[
    X=\#\{\text{free slots in }A\}
\]
and
\[
    Y=\#\{x:h(x)\in A\text{ and }x\text{ is not present in the }K\text{-reveal}\}.
\]
Equivalently, $Y$ counts keys hashing to $A$ whose final displacement is greater than $K$.  Then, for every $u\ge0$,
\[
    \Pr\left[|X-\E\Bk*{X}|>u\right]
    \le
    \exp\left(-\Omega\left(\min\left\{\frac{u^2}{K},u\right\}\right)\right)
\]
and
\[
    \Pr\left[|Y-\E\Bk*{Y}|>u\right]
    \le
    \exp\left(-\Omega\left(\min\left\{\frac{u^2}{K},u\right\}\right)\right).
\]
The probability and expectation are over the Poisson keys, with the offset sequence fixed.
\end{proposition}

The proof of the proposition relies on the following basic lemma about insertion chains.

\begin{claim}[Insertion chain]
\label{claim:inserting-an-element}
Consider the insertion of a new key $x_0$ into an anti-Robin Hood-ordered table.  There is a sequence of keys
\[
    x_0,x_1,\ldots,x_q
\]
and a sequence of slots
\[
    z_0,z_1,\ldots,z_q
\]
such that $x_0$ is placed in $z_0$; for each $i\ge1$, the key $x_i$ was previously in $z_{i-1}$ and is moved to $z_i$; and $z_q$ was empty before the insertion.

For $i\ge1$, let
\[
    \rho_i^-=\text{the displacement traveled by }x_i\text{ to reach }z_{i-1},
\]
and for $i\ge0$, let
\[
    \rho_i^+=\text{the displacement traveled by }x_i\text{ to reach }z_i.
\]
Then
\[
    \rho_0^+\le \rho_1^-\le \rho_1^+\le \rho_2^-\le \rho_2^+\le \cdots\le \rho_q^-\le \rho_q^+.
\]
\end{claim}

\begin{proof}
The insertion procedure follows a single displacement chain.  Suppose $x_i$ displaces $x_{i+1}$ from slot $z_i$.  Then $x_i$ comes earlier than $x_{i+1}$ in the anti-Robin Hood order for $z_i$.  Hence $x_i$ reaches $z_i$ at displacement no larger than the displacement at which $x_{i+1}$ reaches $z_i$, so
\[
    \rho_i^+\le \rho_{i+1}^-.
\]
After $x_i$ is displaced from its old slot, its insertion search continues forward in its own probe sequence, so
\[
    \rho_i^-\le \rho_i^+
\]
for every $i\ge1$.  Combining these inequalities gives the claim.
\end{proof}

Using the insertion-chain lemma above, we can prove that $X$ and $Y$ both have low sensitivity to the addition or removal of a single key.

\begin{lemma}
\label{lem:one-lipschitz}
For the random variables $X$ and $Y$ in Proposition~\ref{prop:concentration}, adding one key to the table changes each of $X$ and $Y$ by at most $1$.
\end{lemma}

\begin{proof}
Fix an insertion chain as in Claim~\ref{claim:inserting-an-element}.

First consider $X$, the number of free slots in $A$ in the $K$-reveal.  In the full table, every occupied slot remains occupied after the insertion; only the final empty slot of the chain becomes newly occupied.  Besides this final slot, in the $K$-reveal, a slot can newly become occupied only if its old occupant had displacement greater than $K$ and its new occupant has displacement at most $K$.

Consider the variables defined in Claim~\ref{claim:inserting-an-element}. At the slot $z_i$ where $x_i$ displaces $x_{i+1}$, the old occupant had displacement $\rho_{i+1}^-$, and the new occupant has displacement $\rho_i^+$.  By Claim~\ref{claim:inserting-an-element},
\[
    \rho_i^+\le \rho_{i+1}^-.
\]
Thus a slot occupied in the $K$-reveal before the insertion remains occupied afterward. Moreover, for any slot $z_i$ besides the final one in the chain, the slot can go from unoccupied to occupied in the $K$-reveal only if
\begin{equation}
    \rho_i^+\le K<\rho_{i+1}^-.
    \label{eq:crossing}
\end{equation}
The monotonicity of the insertion chain allows at most one such crossing.  The final slot $z_q$ may also enter the $K$-reveal, but only if $\rho_q^+\le K$, in which case \eqref{eq:crossing} did not hold for any $i<q$. Thus $X$ can decrease by at most $1$, and it cannot increase.

Now consider $Y$, the number of keys with home position in $A$ and displacement greater than $K$.  Existing keys only move forward in their probe sequences.  Thus an existing key can enter the set counted by $Y$ only if
\[
    \rho_i^-\le K<\rho_i^+
\]
for some $i\ge1$, and monotonicity permits this for at most one key.  The new key $x_0$ may itself be counted by $Y$, but only if $\rho_0^+>K$, in which case the monotonicity
\[
    \rho_0^+\le \rho_1^-\le \rho_1^+\le\cdots
\]
prevents any existing key from crossing from displacement at most $K$ to displacement greater than $K$.  Thus $Y$ can increase by at most $1$, and it cannot decrease.
\end{proof}

Finally, to complete the proof of Proposition~\ref{prop:concentration}, we will also need the following Poissonized bounded-differences inequality.

\begin{lemma}
\label{lem:McD-binomial}
Let $A$ be a set of $\operatorname{Poisson}(m)$ independent uniformly random elements from $[0,1]$.  Let
\[
    f:2^{[0,1]}\to \mathbb R
\]
satisfy
\[
    |f(B\cup\{a\})-f(B)|\le1
\]
for every finite $B\subseteq[0,1]$ and every $a\in[0,1]$.  Then, for every $u\ge0$,
\[
    \Pr\left[|f(A)-\E\Bk*{f(A)}|\ge u\right]
    \le
    2\exp\left(-\Omega\left(\min\left\{\frac{u^2}{m},u\right\}\right)\right).
\]
\end{lemma}

\begin{proof}
    Let $a₁,a₂,… ∈ [0,1]$ be independent and uniformly random, let $A_k = \{a₁,…,a_k\}$ and let $N \sim \operatorname{Poisson}(m)$. This realises $A$ as $A = A_N$. Define $g(k) := \E[f(A_k)]$ and note that $g$ is $1$-Lipschitz because $|f(A_{k+1})-f(A_k)| ≤ 1$. We divide $|f(A)-\E[f(A)]|$ into four parts as follows:
    \[
        |f(A)-\E[f(A)]|
        ≤ \underbrace{|f(A)-g(N)|}_{D₁}
        + \underbrace{|g(N)-g(⌊m⌋)|}_{D₂}
        + \underbrace{|g(⌊m⌋)-\E[g(N)]|}_{D₃}
        + \underbrace{|\E[g(N)]-\E[f(A)]|}_{D₄}.
    \]
    We show for each part separately that it is bounded by $u/4$, except with the stated error probability. We trivially have $D₄ = 0$ since $g(N) = \E[f(A) | N]$. As for $D₃$, we have
    \[
        D₃ = |g(⌊m⌋)-\E[g(N)]| ≤ \E[|g(⌊m⌋)-g(N)|] \stackrel{\text{1-Lipschitz}}{≤} \E[|⌊m⌋-N|] ≤ O(\sqrt{m})
    \]
    where the last step is a basic property of Poisson random variables. Note that the lemma's claim is trivial for $u = O(\sqrt{m})$ so by changing constants we can bound $O(\sqrt{m})$ by $u/4$.
    
    To bound $D₂$ we use a standard Chernoff bound for Poisson random variables which gives
    \[
        \Pr\left[D₂ ≥ u/4\right]
        = 
        \Pr\left[|g(N)-g(⌊m⌋)| ≥ u/4\right]
        \stackrel{\text{1-Lipschitz}}{≤} 
        \Pr\left[|N-\lfloor m\rfloor|\ge u/4\right]
        \le
        \exp\left(-\Omega\left(\min\left\{\frac{u^2}{m},u\right\}\right)\right).
    \]
    In particular we may assume $N ≤ m+u$ also when bounding $D₁$. 
    We condition on $N = k$ for any $k ≤ m+u$. We may now use McDiarmid's inequality to bound $D₁$ as $f(A) = f(A_k)$ is now a function of $k$ independent random variables and replacing one point of $A_k$ by another changes $f(A_k)$ by at most $2$. Hence:
    \begin{align*}
        \Pr[D₁\ge u/4\mid N=k]
        &= 
        \Pr[|f(A)-g(N)| \ge u/4\mid N=k]
        =  
        \Pr[|f(A_k)-\E[f(A_k)]| \ge u/4]\\
        &\le
        2\exp\left(-\Omega(u^2/k)\right)
        \stackrel{k ≤ m+u}{≤} 2\exp\left(-\Omega\left(\min\left\{\frac{u^2}{m},u\right\}\right)\right).\qedhere
    \end{align*}
\end{proof}

Putting the pieces together, we can now prove Proposition~\ref{prop:concentration}.

\begin{proof}[Proof of Proposition~\ref{prop:concentration}]
  The $K$-reveal inside $A$ depends only on keys whose home positions lie within displacement $O(K)$ of $A$.  The number of relevant keys is Poisson with mean $\Theta(K)$.  By Lemma~\ref{lem:one-lipschitz}, adding or deleting one key changes either statistic by at most one.  Applying Lemma~\ref{lem:McD-binomial} with $m=\Theta(K)$, and encoding each relevant key's home position and priority as a point in a bounded probability space, gives the claimed bounds.
\end{proof}

\subsection{The Core Negative Feedback Loop}
\label{subsec:anti-robinhood-negative-feedback-loop}

Our next step towards proving Theorem \ref{thm:anti-robinhood-expected-bounds} is to show the following lemma, which provides the core negative feedback loop that is required for analyzing the distribution of short displacements: if the expected number of keys that survive past displacement $K$ is large enough, then we expect almost all of those keys to get placed within displacement at most $5K$. This lemma holds \emph{even when the load factor is $1$}.

\begin{lemma}
\label{lem:negative-feedback-loop}
Consider a random fixed-offset degree-$d$ probing hash table containing $\operatorname{Poisson}(n)$ keys.  Let $K\le n/10$. Condition on all $r_j$ such that the sampling interval $I_j$ intersects $[0,K]$, and leave the remaining $r_j$'s random. There is a sufficiently large constant $c$, depending only on $d$ and $\alpha$, such that the following holds. If, in the $K$-reveal, every interval of length $K$ has at least $c\max\{\sqrt K,K^{(d-1)/d}\}$ free slots in expectation, then, for every interval $A$ of length $K$,
\[
    \E\Bk*{\#\{x:h(x)\in A\text{ and }\operatorname{disp}(x)>5K\}}
    \le
    10^{-4}
    \E\Bk*{\#\{x:h(x)\in A\text{ and }\operatorname{disp}(x)>K\}}.
\]
\end{lemma}

As pointed out in \cref{subsec:anti-robinhood-intuition}, the high-level idea for \cref{lem:negative-feedback-loop} is that, if there are many keys with displacement greater than $K$, then that also means that there are many free slots in the $K$-reveal to use, which further means a key is more likely to get placed within displacement at most $5K$. To formalize this idea, we first prove two lemmas, where \cref{lem:opportunity} shows formally that many free slots yield a higher probability of a key getting placed within displacement at most $5K$, and \cref{lem:heavyhomes} is the key lemma to handle the secondary-clustering effect where many keys hash to the same position and share the same probe sequence.

\begin{lemma}[Opportunity bound]
  \label{lem:opportunity}
  Fix a home position $u$ and consecutive probe indices $j_1,j_1+1,\ldots,j_2.$
  Let the interval $I \defeq \bigcup_{j=j_1}^{j_2}I_j$ be the union of the sampling intervals $I_j = [\alpha j^d,\alpha(j+1)^d)$ for the probe indices $j_1,j_1+1,\ldots,j_2$, and set $K_- = \alpha j_1^d-1$ and $K_+=\alpha(j_2+1)^d$. 

  Condition on the set of keys present, their home positions, their priorities, and the offsets $r_1,r_2,\ldots,r_{j_1-1}$.
  Suppose that fewer than $\beta$ keys with home position $u$ are absent from the $K_-$-reveal.  Then, for every $S\ge0$,
  \[
      \Pr\left[
          \begin{array}{c}
          \text{some key with home position }u\text{ is absent from the }
          K_+\text{-reveal, and}\\
          \text{at least }S\text{ slots }z\in u+I\text{ are free in the }K_+\text{-reveal}
          \end{array}
      \right]
      \le
      \exp\left(O(\beta)-\Omega\left(S/K_+^{(d-1)/d}\right)\right).
  \]
\end{lemma}

\begin{proof}
  Fix $j\in\{j_1,\ldots,j_2\}$ and suppose that $r_{j_1},\ldots,r_{j-1}$ have already been revealed.  A slot $z$ is a \defn{$j$-opportunity for $u$} if
  \[
      z\in u+I_j
      \qquad\text{and}\qquad
      z\text{ is free in the }(\alpha j^d-1)\text{-reveal}.
  \]
  Let $T_j$ be the number of $j$-opportunities.  These slots are determined before $r_j$ is sampled.

  If $u+r_j$ is a $j$-opportunity and some key with home position $u$ is still absent just before offset $r_j$ is revealed, then the next absent key with home position $u$ is placed at $u+r_j$.  Indeed, no offset of the common probe sequence lies between $\alpha j^d$ and $r_j$ before the $j$-th probe, and at offset $r_j$ the only keys that can probe $u+r_j$ are keys with home position $u$.

  Since $r_j$ is uniform in $I_j$, and since $|I_j|\le C K_+^{(d-1)/d}$ for a constant $C=C(d,\alpha)$,
  \[
      \Pr[u+r_j\text{ hits a }j\text{-opportunity}
      \mid r_{j_1},\ldots,r_{j-1}]
      =
      \frac{T_j}{|I_j|}
      \ge
      \frac{T_j}{C K_+^{(d-1)/d}} .
  \]
  Thus the opportunity hits stochastically dominate adaptive coin flips whose head probability at step $j$ is $T_j/(C K_+^{(d-1)/d})$.

  Initially, fewer than $\beta$ keys with home position $u$ are absent from the $K_-$-reveal.  Therefore, if some such key is still absent from the $K_+$-reveal, then fewer than $\beta$ opportunity hits occurred during probes $j_1,\ldots,j_2$.

  Conversely, every slot $z\in u+I$ that is free in the $K_+$-reveal lies in $u+I_j$ for a unique $j\in\{j_1,\ldots,j_2\}$, and it was already free in the $(\alpha j^d-1)$-reveal.  Hence the event in the lemma with $S$ free slots implies $\sum_{j=j_1}^{j_2}T_j\ge S$.

  Applying \Cref{lem:losinglotsofflips} below with $t=C K_+^{(d-1)/d}$ gives
  \[
      \Pr\left[
          \#\{\text{opportunity hits}\}<\beta
          \text{ and }
          \sum_{j=j_1}^{j_2}T_j\ge S
      \right]
      \le
      \exp\left(\beta-\Omega\left(S/K_+^{(d-1)/d}\right)\right),
  \]
  which proves the lemma.
\end{proof}

\begin{lemma}
  \label{lem:losinglotsofflips}
  Let $t>0$ and let $\beta\ge0$ be an integer.  A sequence of numbers
  \[
      t_1,t_2,\ldots,t_\ell\in[0,t]
  \]
  is revealed one at a time.  After $t_i$ is revealed, a coin is flipped heads with probability $t_i/t$.  The value of $t_i$ may be chosen adaptively from previous outcomes.  Then, for every $S\ge0$,
  \[
      \Pr\left[
          \#\{\text{heads}\}<\beta
          \text{ and }
          \sum_{i=1}^{\ell}t_i\ge S
      \right]
      \le
      \exp\left(\beta-\Omega(S/t)\right).
  \]
  \end{lemma}
  
  \begin{proof}
  Let $H_i$ be the number of heads among the first $i$ flips, and let
  \[
      S_i=\sum_{r=1}^i t_r.
  \]
  Let $\mathcal F_{i-1}$ be the history before $t_i$ is chosen, and let $\mathcal G_i$ be the history after $t_i$ is revealed but before the $i$-th coin is flipped.  For $c=1-e^{-1}$, define
  \[
      M_i=\exp(cS_i/t-H_i).
  \]
  We verify that $(M_i)$ is a supermartingale.  Conditional on $\mathcal G_i$, the head probability is $p_i=t_i/t$, so
  \[
      \E\Bk*{e^{-\ind\nolimits_{\BK{\text{head at }i}}}\mid\mathcal G_i}
      =
      1-p_i(1-e^{-1})
      \le e^{-cp_i}.
  \]
  Hence
  \[
      \E\Bk*{M_i\mid\mathcal F_{i-1}}
      =
      \E\Bk*{M_{i-1}e^{ct_i/t}
          \E\Bk*{e^{-\ind\nolimits_{\BK{\text{head at }i}}}\mid\mathcal G_i}
          \mid\mathcal F_{i-1}}
      \le M_{i-1}.
  \]
  Thus $\E\Bk*{M_\ell}\le1$.  On the event in the lemma,
  \[
      M_\ell\ge \exp(cS/t-\beta),
  \]
  so Markov's inequality gives the result.
  \end{proof}

\begin{lemma}[Handling slots with heavy hashes]
  \label{lem:heavyhomes}
  Fix a displacement parameter $K$ and an interval $A$ of home positions.  For each home position $u$, let
  \[
      Y_u(K)=\#\{x:h(x)=u\text{ and }\operatorname{disp}(x)>K\}.
  \]
  Define
  \[
      Y_A(K)=\sum_{u\in A}Y_u(K)
  \]
  and, for an integer $\beta\ge1$,
  \[
      Y_A^{(\beta)}(K)
      =
      \sum_{u\in A}Y_u(K)\ind\nolimits_{\BK*{Y_u(K)\ge \beta}}.
  \]
  In the Poissonized model with load at most $1$,
  \[
      \E\Bk*{Y_A^{(\beta)}(K)}
      \le
      e^{-\Omega(\beta\log\beta)}
      \E\Bk*{Y_A(K)}.
  \]
\end{lemma}

\begin{proof}
  It suffices to prove the bound for a single home position $u$, and then sum over $u\in A$.
  
  Fix $u$.  Condition on the offset sequence and on all keys with home position different from $u$.  The keys with home $u$ form a Poisson point process on the priority interval $[0,1]$ of intensity at most $1$, where larger priority means earlier placement among same-home keys.
  
  Expose the part of the $K$-reveal involving home position $u$.  Since same-home keys are processed in decreasing priority order, this exposure either shows that no same-home key survives past displacement $K$, or leaves exactly the unexposed Poisson points below some priority cutoff.  Consequently, conditional on the exposed information,
  \[
      Y_u(K)\sim \operatorname{Poisson}(\lambda_u)
      \qquad\text{for some }\lambda_u\le1,
  \]
  or $Y_u(K)=0$.  Thus the lemma follows from the fact that, if $Z\sim\operatorname{Poisson}(\lambda)$ with $\lambda\le1$, then
  \begin{align*}
      \E\Bk*{Z\ind\nolimits_{\BK*{Z\ge\beta}}}
      \le{}&
      e^{-\Omega(\beta\log\beta)}\E\Bk*{Z}.
      \qedhere
  \end{align*}
\end{proof}

With the two lemmas above, we are ready to prove \cref{lem:negative-feedback-loop}.
\begin{proofof}{\cref{lem:negative-feedback-loop}}
  Note that the lemma is trivial for any $K = O(1)$, since we can take $c$ large enough so that the if-statement in the lemma never goes into effect. We will therefore assume (WLOG) for the rest of the proof that $K$ is large enough that every sampling interval intersecting $[K,5K]$ has length at most $K$. 
  
  Let
  \[
      \delta=10^{-4}.
  \]
  Choose constants $\beta$ and $c$, where $\beta$ is sufficiently large compared to $\delta$ and $c$ is sufficiently large compared to $\beta$.
  Consider four consecutive intervals
  \[
      A,B,C,D,
  \]
  each of length $K$.  Define
  \[
      X_1=\#\{\text{free slots in }D\text{ in the }K\text{-reveal}\},
      \qquad
      X_2=\#\{\text{free slots in }D\text{ in the }5K\text{-reveal}\},
  \]
  and
  \[
      Y_1=\#\{x:h(x)\in A\text{ and }\operatorname{disp}(x)>K\},
      \qquad
      Y_2=\#\{x:h(x)\in A\text{ and }\operatorname{disp}(x)>5K\}.
  \]
  For $u\in A$, let
  \[
      Y_u=\#\{x:h(x)=u\text{ and }\operatorname{disp}(x)>K\},
  \]
  and define
  \[
      Y_1^{(\beta)}
      =
      \sum_{u\in A}Y_u\ind\nolimits_{\BK*{Y_u\ge\beta}}.
  \]
  Since the expected number of keys overall is $n$, we have by symmetry that
  \[
      \mu_1:=\E\Bk*{X_1}=\E\Bk*{Y_1}.
  \]
  By assumption,
  \[
      \mu_1\ge c\max\{\sqrt K,K^{(d-1)/d}\},
  \]
  so in particular $\mu_1\ge c\sqrt K$ and $\mu_1\ge cK^{(d-1)/d}$.
  
  Let $\mathcal R$ be the still-random offsets with sampling intervals intersecting $(K,5K]$, and write
  \[
      \nu(\mathcal R)=\E\Bk*{Y_2\mid \mathcal R},
  \]
  where the expectation is over the Poisson keys.  Conditional on $\mathcal R$, symmetry gives
  \[
      \E\Bk*{X_2\mid\mathcal R}=\E\Bk*{Y_2\mid\mathcal R}=\nu(\mathcal R).
  \]
  
  Suppose for contradiction that
  \[
      \E_{\mathcal R}\Bk*{\nu(\mathcal R)}>\delta\mu_1.
  \]
  Since $0\le Y_2\le Y_1$, since $Y_1$ is determined by offsets whose sampling intervals intersect $[0,K]$ and is therefore independent of $\mathcal R$, and since $\E\Bk*{Y_1}=\mu_1$, we have
  \[
      0\le \nu(\mathcal R)\le \mu_1.
  \]
  Hence
  \[
      \Pr_{\mathcal R}\left[\nu(\mathcal R)\ge \frac{\delta}{2}\mu_1\right]
      \ge
      \frac{\delta}{2}.
  \]
  Call such offset realizations $\mathcal R$ \defn{good}.
  
  Fix a good $\mathcal R$.  Proposition~\ref{prop:concentration}, applied with the full offset sequence fixed and reveal parameter $5K$ where needed, implies that, with probability at least $0.9$ over the Poisson keys,
  \[
      Y_1\le2\mu_1,
      \qquad
      Y_2\ge \frac{\delta}{4}\mu_1,
      \qquad
      X_2\ge \frac{\delta}{4}\mu_1,
  \]
  provided $c$ is large enough. Indeed, the proposition says that, for each of $Y_1,Y_2,X_2$ conditioned on the offset sequence, the probability of deviating from the mean by at least $\delta\mu_1/4\ge \delta c\sqrt K/4$ is at most, say, $0.001$.
  
  Choose $\beta$ so large that Lemma~\ref{lem:heavyhomes} gives
  \[
      \E\Bk*{Y_1^{(\beta)}}\le \frac{\delta^2}{1000}\mu_1.
  \]
  By Markov's inequality,
  \[
      \Pr\left[Y_1^{(\beta)}>\frac{\delta}{16}\mu_1\right]\le \frac{\delta}{50}.
  \]
  Thus the event
  \[
      \mathcal G=
      \left\{
          \mathcal R\text{ is good},\,
          Y_1\le2\mu_1,\,
          Y_2\ge\frac{\delta}{4}\mu_1,\,
          X_2\ge\frac{\delta}{4}\mu_1,\,
          Y_1^{(\beta)}\le\frac{\delta}{16}\mu_1
      \right\}
  \]
  has probability at least $\delta/4$.
  
  Condition on what we will call the \defn{augmented $K$-data} consisting of the $K$-reveal together with the counts $Y_u$ for $u\in A$.  This fixes $Y_1$, $Y_1^{(\beta)}$, and the light home positions in $A$ (the home positions with $u \in A$ with $Y_u < \beta$), while leaving the offsets whose sampling intervals are disjoint from $[0,K]$ unrevealed.
  
  Fix a light home position $u\in A$, so $1\le Y_u<\beta$. Let $I_{j_1}$ be the first sampling interval that intersects the set of displacements
  \[
      D-u=\{z-u:z\in D\},
  \]
  and let $I_{j_2}$ be the last such sampling interval. Since $D-u\subseteq[2K,4K]$ and every sampling interval intersecting $[K,5K]$ has length at most $K$, the intervals $I_{j_1}, I_{j_1 + 1}, \ldots, I_{j_2}$ are disjoint from $[1, K]$. Critically, this means that the augmented $K$-data does not reveal any information about $r_{j_1}, \ldots, r_{j_2}$.  Let $K_+=\alpha(j_2+1)^d$ be the right endpoint used in Lemma~\ref{lem:opportunity}; the preceding sentence implies $K_+\le5K$.
  
  Since $D-u\subseteq \bigcup_{j=j_1}^{j_2}I_j$, every slot of $D$ that is free in the $5K$-reveal is also free in the $K_+$-reveal.  Therefore, if $X_2\ge\delta\mu_1/4$, then Lemma~\ref{lem:opportunity} applies with $S=\delta\mu_1/4$ and $K_+^{(d-1)/d}=O(K^{(d-1)/d})$, giving
  \[
      \Pr\left[
          \text{some key with home position }u\text{ is absent from the }5K\text{-reveal}
          \text{ and }X_2\ge\frac{\delta}{4}\mu_1
          \mid \text{augmented }K\text{-data}
      \right]
      \le
      q,
  \]
  where
  \[
      q=
      \exp\left(O(\beta)-\Omega\left(\frac{\delta\mu_1}{K^{(d-1)/d}}\right)\right)
      \le
      \exp(O(\beta)-\Omega(\delta c)),
  \]
  and where the asymptotic notation hides constants that are independent of $\beta$ and $c$.
  
  Let $Z$ be the number of keys counted by $Y_2$ whose home position was light at displacement $K$, and that occur together with $X_2\ge\delta\mu_1/4$.  Conditioning on the augmented $K$-data and summing over light home positions,
  \[
      \E\Bk*{Z\ind\nolimits_{\BK*{Y_1\le2\mu_1}}
          \mid \text{augmented }K\text{-data}}
      \le
      qY_1\ind\nolimits_{\BK*{Y_1\le2\mu_1}}
      \le
      2q\mu_1.
  \]
  Thus, removing the conditioning on the augmented $K$-data, we have
  \[
      \E\Bk*{Z\ind\nolimits_{\BK*{Y_1\le2\mu_1}}}\le 2q\mu_1.
  \]
  
  On $\mathcal G$,
  \[
      Z\ind\nolimits_{\BK*{Y_1\le2\mu_1}}
      \ge
      Y_2-Y_1^{(\beta)}
      \ge
      \frac{\delta}{8}\mu_1.
  \]
  Therefore, by Markov's inequality,
  \[
      \Pr[\mathcal G]\le \frac{16q}{\delta}.
  \]
  Choosing $c$ large enough makes this smaller than $\delta/4$, contradicting $\Pr[\mathcal G]\ge\delta/4$.  Therefore
  \[
      \E\Bk*{Y_2}\le \delta\mu_1,
  \]
  as claimed.
  \end{proofof}
  
  The negative-feedback lemma immediately yields the critical tail estimate at load $1$.  This is the short-displacement regime that will later control the part of a query before the load slack $\epsilon$ becomes visible.
  
  \begin{corollary}
  \label{cor:critical-tail}
  For a random fixed-offset degree-$d$ probing hash table containing $\operatorname{Poisson}(n)$ keys, define
  \[
      P(K)=\frac1n
      \E\Bk*{\#\{x:\operatorname{disp}(x)>K\}}.
  \]
  Define $\chi_d \defeq \min\{1/2, 1/d\}$. Then, for all $K\le n/10$,
  \[
      P(K)=O(K^{-\chi_d}).
  \]
  \end{corollary}
  
  \begin{proof}
  It suffices to prove the estimate for power-of-$5$ values of $K$; the bound for intermediate values then follows from monotonicity after increasing the constant by a factor depending only on $d$.  We prove by induction that, for a sufficiently large constant $C$,
  \begin{equation}
      P(K)\le CK^{-\chi_d}
      \label{eq:critical-tail-induction}
  \end{equation}
  for every power-of-$5$ value $K\le n/10$.  The case $K=1$ holds after taking $C\ge1$.

  Fix a power-of-$5$ value $K$ with $5K\le n/10$, and assume \eqref{eq:critical-tail-induction} holds at scale $K$.  Let $\rho$ be the sequence of all offset choices $r_j$ whose sampling interval $I_j$ intersects $[0,K]$, ordered by $j$.  For a fixed value $\rho^*$ of this sequence, define the conditional tail density
  \[
      P_{\rho^*}(K)=\frac1n
      \E\Bk*{\#\{x:\operatorname{disp}(x)>K\}\mid \rho=\rho^*}.
  \]
  The $K$-reveal is already determined once these offsets are fixed, so the expectation above is only over the Poisson keys.  Moreover, conditioned on $\rho=\rho^*$, translation invariance and the identity $\E[N]=n$ imply that every interval of length $K$ has
  \[
      KP_{\rho^*}(K)
  \]
  free slots in expectation in the $K$-reveal.  Indeed, all length-$K$ intervals have the same expected number of free slots, and the expected total number of free slots in the $K$-reveal equals the expected number of keys with displacement greater than $K$.

  Set
  \[
      T_K=c\frac{\max\{\sqrt K,K^{(d-1)/d}\}}K
          =cK^{-\chi_d},
  \]
  where $c$ is the constant from Lemma~\ref{lem:negative-feedback-loop}.  If $P_{\rho^*}(K)\ge T_K$, then the conditional form of the hypothesis of Lemma~\ref{lem:negative-feedback-loop} holds for this fixed offset sequence.  Applying the lemma and summing its conclusion over all length-$K$ home intervals gives
  \[
      \frac1n\E\Bk*{\#\{x:\operatorname{disp}(x)>5K\}\mid \rho=\rho^*}
      \le 10^{-4}P_{\rho^*}(K).
  \]
  If instead $P_{\rho^*}(K)<T_K$, then monotonicity of the displacement threshold gives
  \[
      \frac1n\E\Bk*{\#\{x:\operatorname{disp}(x)>5K\}\mid \rho=\rho^*}
      \le P_{\rho^*}(K)<T_K.
  \]
  Therefore, for every fixed offset sequence $\rho^*$,
  \[
      \frac1n\E\Bk*{\#\{x:\operatorname{disp}(x)>5K\}\mid \rho=\rho^*}
      \le 10^{-4}P_{\rho^*}(K)+cK^{-\chi_d}.
  \]
  Averaging over the random sequence $\rho$ yields
  \[
      P(5K)\le 10^{-4}P(K)+cK^{-\chi_d}.
  \]
  By the inductive hypothesis,
  \[
      P(5K)\le (10^{-4}C+c)K^{-\chi_d}.
  \]
  Choosing $C$ large enough that
  \[
      10^{-4}C+c\le 5^{-\chi_d}C
  \]
  gives $P(5K)\le C(5K)^{-\chi_d}$, completing the induction.
\end{proof}

\subsection{Upper Bounds for Long Displacements}
\label{subsec:anti-robinhood-long-displacements}
Recall that, as pointed out in \cref{subsec:anti-robinhood-intuition}, the displacement distribution behaves differently for long displacements and short displacements. \cref{cor:critical-tail} gives an upper bound for short displacements. In this subsection, we will prove an upper bound for long displacements. This will serve as the third major ingredient in our proof of Theorem \ref{thm:anti-robinhood-expected-bounds}.

\begin{lemma}
  \label{lem:longreveal}
  Consider a random fixed-offset degree-$d$ probing hash table containing $\operatorname{Poisson}((1-\epsilon)n)$ keys.  There exist constants $\eta,a,C_0>0$, depending only on $d$ and $\alpha$, such that the following holds.  Let
  \[
      C_0\max\{\epsilon^{-2},\epsilon^{-d}\}\le K\le n/10.
  \]
  Condition on all offset choices whose sampling intervals intersect $[0,K]$, and leave the remaining offsets random.  Suppose that
  \[
      \frac1n
      \E\Bk*{\#\{x:\operatorname{disp}(x)>K\}}
      \le
      \eta\epsilon.
  \]
  Then
  \[
      \frac1n
      \E\Bk*{\#\{x:\operatorname{disp}(x)>5K\}}
      \le
      \epsilon\exp\left(-a\min\{\epsilon^2K,\epsilon K^{1/d}\}\right).
  \]
\end{lemma}
  
\begin{proof}
  The proof follows the same template as the proof of \cref{lem:negative-feedback-loop}; the load slack $\epsilon$ gives many genuinely empty slots, and the only additional work is to keep track of the rare event where this slack is not visible locally.  By taking $C_0$ sufficiently large, we may assume throughout that every sampling interval intersecting $[K,5K]$ has length at most $K$.
  
  Partition the table into intervals of length $K$, and fix four consecutive intervals
  \[
      A,B,C,D.
  \]
  Define
  \[
      X_1=\#\{\text{free slots in }D\text{ in the }K\text{-reveal}\},
  \]
  and
  \[
      Y_1=\#\{x:h(x)\in A,\operatorname{disp}(x)>K\},
      \qquad
      Y_2=\#\{x:h(x)\in A,\operatorname{disp}(x)>5K\}.
  \]
  For $u\in A$, write
  \[
      Y_u=\#\{x:h(x)=u,\operatorname{disp}(x)>K\}.
  \]
  Since the expected load is $1-\epsilon$,
  \[
      \E\Bk*{X_1}\ge \epsilon K.
  \]
  By assumption and by symmetry across intervals of length $K$,
  \[
      \E\Bk*{Y_1}\le \eta\epsilon K.
  \]

  Let $E$ be the length-$6K$ interval consisting of $D$ together with the five length-$K$ intervals immediately preceding it, and let
  \[
      R=\#\{x:h(x)\in E\text{ and }\operatorname{disp}(x)>K\}.
  \]
  Then
  \[
      \E\Bk*{R}\le 6\eta\epsilon K.
  \]
  Choose $\eta$ sufficiently small, and define the common-case event
  \[
      \mathcal H=
      \left\{
          X_1\ge \frac{\epsilon K}{2},\quad
          R\le \frac{\epsilon K}{8},\quad
          Y_1\le \frac{\epsilon K}{8}
      \right\}.
  \]
  We first bound the contribution from $\neg\mathcal H$.  Since $Y_2\le Y_1$, and setting $T=\epsilon K/8$,
  \[
      \E\Bk*{Y_2\ind\nolimits_{\neg\mathcal H}}
      \le
      T\Pr\left[X_1<\frac{\epsilon K}{2}\right]
      +
      T\Pr\left[R>\frac{\epsilon K}{8}\right]
      +
      \E\Bk*{Y_1\ind\nolimits_{\BK*{Y_1>T}}}.
  \]
  Proposition~\ref{prop:concentration} gives
  \[
      \Pr\left[X_1<\frac{\epsilon K}{2}\right]
      \le
      e^{-\Omega(\epsilon^2K)}.
  \]
  Applying the same proposition to the six length-$K$ intervals forming $E$, and using $\eta$ small, gives
  \[
      \Pr\left[R>\frac{\epsilon K}{8}\right]
      \le
      e^{-\Omega(\epsilon^2K)}.
  \]
  Finally, by the following tail-summation argument, we get\todo{The argument seems sound, but I don't recall seeing this in \cref{lem:negative-feedback-loop}. Jingxun: Edited.}
  \[
  \begin{aligned}
      \E\Bk*{Y_1\ind\nolimits_{\BK*{Y_1>T}}}
      &\le
      T\Pr[Y_1>T]+\sum_{\ell\ge0}2^{\ell+1}T\Pr[Y_1\ge 2^\ell T] \\
      &\le
      \epsilon K
      \sum_{\ell\ge0}
      2^{\ell+1}
      \exp\left(
          -\Omega\left(
              \min\left\{\frac{(2^\ell\epsilon K)^2}{K},2^\ell\epsilon K\right\}
          \right)
      \right) \\
      &\le
      \epsilon K e^{-\Omega(\epsilon^2K)}.
  \end{aligned}
  \]
  Hence
  \[
      \E\Bk*{Y_2\ind\nolimits_{\neg\mathcal H}}
      \le
      \epsilon K e^{-\Omega(\epsilon^2K)}.
  \]
  
  We now bound $\E\Bk*{Y_2\ind\nolimits_{\mathcal H}}$.  Set
  \[
      \beta=\left\lceil b_0\epsilon K^{1/d}\right\rceil,
  \]
  where $b_0>0$ is a sufficiently small constant.  Since $K\ge C_0\max\{\epsilon^{-2},\epsilon^{-d}\}$, choosing $C_0$ large ensures $\beta\ge1$.  Define
  \[
      Y_1^{(\beta)}
      =
      \sum_{u\in A}Y_u\ind\nolimits_{\BK*{Y_u\ge\beta}}.
  \]
  By Lemma~\ref{lem:heavyhomes},
  \[
      \E\Bk*{Y_1^{(\beta)}}
      \le
      e^{-\Omega(\beta)}\E\Bk*{Y_1}
      \le
      \eta\epsilon K e^{-\Omega(\epsilon K^{1/d})}.
  \]
  Thus heavy home positions contribute negligibly to $Y_2$, and it remains to control light home positions.
  
  Suppose $\mathcal H$ occurs, and condition on the augmented $K$-data for the intervals involved: the $K$-reveal, the counts $Y_u$ for $u\in A$, and the count $R$.  This fixes $X_1,Y_1,R,Y_1^{(\beta)}$, and $\mathcal H$, while leaving the offsets whose sampling intervals are disjoint from $[0,K]$ unrevealed.  This is the same deferred-randomness setup used in the proof of \cref{lem:negative-feedback-loop}: it fixes the surviving keys and their home-position multiplicities without exposing the future offsets that may place them.  The difference is that the earlier proof worked with slots already known to be free in the $5K$-reveal, whereas here the load slack is visible only in the $K$-reveal.  The additional count $R$ controls how many of those currently free slots can be filled before the relevant future offsets are exposed.
  
  Fix a light home position $u\in A$, meaning that $1\le Y_u<\beta$.  Let $I_{j_1}$ be the first sampling interval that intersects $D - u$ 
  and let $I_{j_2}$ be the last such sampling interval.  Since $D-u\subseteq[2K,4K]$ and every sampling interval intersecting $[K,5K]$ has length at most $K$, the intervals $I_{j_1},I_{j_1+1},\ldots,I_{j_2}$ are disjoint from $[1,K]$.  Thus the augmented $K$-data does not reveal $r_{j_1},\ldots,r_{j_2}$.  Let $K_-=\alpha j_1^d-1$ and $K_+=\alpha(j_2+1)^d$ be as in Lemma~\ref{lem:opportunity}; then $K_-\ge K$ and $K_+\le5K$.

  Among the $X_1$ slots of $D$ that are free in the $K$-reveal, any slot that is not free in the $K_+$-reveal must be occupied by a key with home position in $E$ and displacement greater than $K$; such keys are counted by $R$.  On $\mathcal H$, at least
  \[
      X_1-R\ge \frac{3\epsilon K}{8}
  \]
  slots of $D$ are therefore free in the $K_+$-reveal.  Since all of $D-u$ is covered by $I_{j_1},\ldots,I_{j_2}$, Lemma~\ref{lem:opportunity} gives
  \[
      \Pr\left[
          \begin{array}{c}
          \text{some home-}u\text{ key is absent from the }5K\text{-reveal}\\
          \mid \text{augmented }K\text{-data}
          \end{array}
      \right]
      \le
      \exp\left(O(\beta)-\Omega\left(\frac{\epsilon K}{K^{(d-1)/d}}\right)\right).
  \]
  Since $\epsilon K/K^{(d-1)/d}=\epsilon K^{1/d}$, and since $b_0$ is sufficiently small, the last display is
  \[
      q=\exp(-\Omega(\epsilon K^{1/d})).
  \]
  Linearity of expectation over light home positions gives
  \[
      \E\Bk*{(Y_2-Y_1^{(\beta)})_+\ind\nolimits_{\mathcal H}
          \mid \text{augmented }K\text{-data}}
      \le
      qY_1\ind\nolimits_{\mathcal H}.
  \]
  Taking expectations and using $\E\Bk*{Y_1}\le \eta\epsilon K$ yields
  \[
      \E\Bk*{(Y_2-Y_1^{(\beta)})_+\ind\nolimits_{\mathcal H}}
      \le
      \eta\epsilon K e^{-\Omega(\epsilon K^{1/d})}.
  \]
  Combining this with the bound on $\E\Bk*{Y_1^{(\beta)}}$ gives
  \[
      \E\Bk*{Y_2\ind\nolimits_{\mathcal H}}
      \le
      \epsilon K e^{-\Omega(\epsilon K^{1/d})}.
  \]
  Combining the estimates on $\mathcal H$ and $\neg\mathcal H$, we obtain
  \[
      \E\Bk*{Y_2}
      \le
      \epsilon K\exp\left(-a\min\{\epsilon^2K,\epsilon K^{1/d}\}\right)
  \]
  for some constant $a>0$.  Since $\E\Bk*{Y_2} = \frac{K}{n}\E\Bk*{\#\{x:\operatorname{disp}(x)>5K\}}$, the lemma follows.
\end{proof}

\subsection{Proof of Theorem~\ref{thm:anti-robinhood-expected-bounds}}
\label{subsec:anti-robinhood-proof-of-theorem}
Combining the results from the previous subsections, we can now complete the proof of \cref{thm:anti-robinhood-expected-bounds}, restated below for convenience:

\antirobinhoodexpectedbounds*

\begin{proofof}{\cref{thm:anti-robinhood-expected-bounds}}
  We first prove the displacement bounds that will be used for queries.  It is convenient to do this with a generic load parameter $\theta$, and later we will take $\theta=\epsilon/2$.

For a Poissonized table containing $\operatorname{Poisson}((1-\theta)n)$ keys, define
\[
    P_\theta(K)
    =
    \frac1n
    \E\Bk*{\#\{x:\operatorname{disp}(x)>K\}}.
\]
By monotonicity in the number of keys and Corollary~\ref{cor:critical-tail},
\[
    P_\theta(K)\le P_0(K)\le C_{\mathrm{crit}}K^{-\chi_d}
\]
for all $K\le n/10$, where $C_{\mathrm{crit}}$ is an absolute constant depending only on $d$ and $\alpha$.

Let $\eta$ and $C_0$ be the constants from Lemma~\ref{lem:longreveal}.  Choose
\[
    K_\theta
    =
    L\max\{\theta^{-2},\theta^{-d}\},
\]
where $L$ is a sufficiently large constant, depending only on $d,\alpha,\eta,C_0$, so that
\[
    C_{\mathrm{crit}}K_\theta^{-\chi_d}\le \eta\theta
    \qquad\text{and}\qquad
    K_\theta\ge C_0\max\{\theta^{-2},\theta^{-d}\}.
\]
Thus, whenever $K_\theta\le n/10$,
\[
    P_\theta(K_\theta)\le \eta\theta.
\]
The assumptions on $\epsilon$ in the theorem, with $C_{\mathrm{fin}}$ large enough, ensure that this condition holds for $\theta=\epsilon/2$.

Since $P_\theta(K)$ is nonincreasing in $K$, the hypothesis of Lemma~\ref{lem:longreveal} holds at every scale $K\ge K_\theta$.  Therefore, applying Lemma~\ref{lem:longreveal} at scale $K$ gives
\[
    P_\theta(5K)
    \le
    \theta\exp\left(-a\min\{\theta^2K,\theta K^{1/d}\}\right)
\]
for every $K_\theta\le K\le n/10$.  Equivalently, after changing the constant in the exponent, for every
\[
    K\ge 5K_\theta
\]
in the range under consideration,
\begin{equation}
    P_\theta(K)
    \le
    \theta\exp\left(-\Omega\left(\min\{\theta^2K,\theta K^{1/d}\}\right)\right).
    \label{eq:ptheta1}
\end{equation}
In addition to this, for all $K$, we can always use the load-one bound
\begin{equation}
    P_\theta(K)=O(K^{-\chi_d})
    \label{eq:ptheta2}
\end{equation}
that comes from Corollary \ref{cor:critical-tail}.

\paragraph{Transferring to fixed-load tables.}
We next transfer these bounds \eqref{eq:ptheta1} and \eqref{eq:ptheta2} to a fixed-load table.  Let $P_\theta^{\mathrm{fix}}(K)$ denote the analogue of $P_\theta(K)$ for a table with $(1-\theta)n$ keys (rather than a Poisson number of keys).  Let $N\sim\operatorname{Poisson}((1-\theta)n)$.  There is an absolute constant $\gamma>0$ such that
\[
    \Pr[N\ge (1-\theta)n]\ge \gamma.
\]
On the event $N\ge(1-\theta)n$, delete uniformly chosen keys from the Poissonized table until exactly $(1-\theta)n$ keys remain.  The resulting key set has the fixed-load distribution, and deleting keys cannot increase displacements.  Hence
\[
    P_\theta^{\mathrm{fix}}(K)\le \gamma^{-1}P_\theta(K).
\]
Thus the two displacement-tail bounds \eqref{eq:ptheta1} and \eqref{eq:ptheta2} also hold in the fixed-load model.

\paragraph{Establishing query time bounds.}
Finally, we translate our bound on displacement tails to bounds on query time.  Set
\[
    \theta=\epsilon/2.
\]
We will bound both successful and unsuccessful queries using the fixed-load tail $P_\theta^{\mathrm{fix}}$.

First consider a successful query to a fixed stored key in the table of load $1-\epsilon$.  By history independence, and symmetry across keys, the query time has the same distribution as if we queried a uniformly random stored key.  The probability that a uniformly random stored key has displacement greater than $K$ is
\[
    \frac{n}{(1-\epsilon)n}\,P_\epsilon^{\mathrm{fix}}(K)
    =
    O(P_\epsilon^{\mathrm{fix}}(K))
    \le
    O(P_\theta^{\mathrm{fix}}(K)).
\]

Now consider an unsuccessful query for a key $y$ (that is not present).  The unsuccessful query stops no later than the position where $y$ would be placed if inserted. The expected time is therefore at most the expected \emph{successful-query} time in a table with $(1 - \epsilon)n + 1$ keys. For all sufficiently large $n$ in the theorem's parameter range,
\[
    (1-\epsilon)n+1\le (1-\epsilon/2)n=(1-\theta)n.
\]
Therefore, just as for successful queries, we have that the probability of the query searching to a displacement greater than $K$ is at most $O(P_\theta^{\mathrm{fix}}(K))$.

Let $i_\star=\Theta(n^{1/d})$ be the largest probe index whose sampling interval intersects $[0,n/10]$.  To control the remaining probes, we use the following expected completion bound.
\begin{restatable}[Polynomial Expected Completion]{lemma}{polycompletion}
    \label{lem:polynomial-completion}
    Fix $d\ge1$.  For random fixed-offset degree-$d$ probing, define
    \[
        \tau_n
        \defeq
        \inf\left\{
            t\in\N:
            \{r_j\bmod n:j\ge n,\ r_j\le t\}=[n]
        \right\}.
    \]
    In words, $\tau_n$ is the least true displacement by which the tail of the common offset sequence beginning at probe index $n$ has visited every physical slot.  For smoothed degree-$d$ probing and a fixed key $x$, define
    \[
        \tau_n(x)
        \defeq
        \inf\left\{
            t\in\N:
            \{k\bmod n:n\le k\le t,\ k\text{ is activated for }x\}=[n]
        \right\};
    \]
    this is the corresponding completion threshold for the activated offsets of $x$.

    There is a constant $B<\infty$, depending only on the constants defining the probing scheme, such that for every table size $n$,
    \[
        \E\Bk*{\tau_n}\le O(n^B)
    \]
    for random fixed-offset degree-$d$ probing, while
    \[
        \E\Bk*{\tau_n(x)}\le O(n^B)
    \]
    for smoothed degree-$d$ probing and every fixed key $x$.
\end{restatable}

The proof is deferred to \cref{app:completion}.  By definition, $\tau_n$ depends only on offset variables whose sampling intervals are disjoint from $[0,n/10]$; it is therefore independent of the keys and of all offset variables that determine whether a query reaches displacement $n/10$.  The same independence holds for $\tau_n(x)$ because it uses only activation bits at offsets at least $n$.  Once the relevant completion threshold is reached, the query has probed every physical slot, and the number of probes made before true displacement $t$ is at most $t+1$.  Consequently, \cref{lem:polynomial-completion} gives
\[
    \E\Bk*{\text{query time}}
    \le
    O\left(
        1+
        \sum_{i=1}^{i_\star}
            P_\theta^{\mathrm{fix}}(\Theta(i^d))
        +
        n^B P_\theta^{\mathrm{fix}}(n/10)
    \right).
\]
Let
\[
    i_0
    =
    \Theta(K_\theta^{1/d})
    =
    \Theta\left(\max\{\theta^{-2/d},\theta^{-1}\}\right)
    =
    \Theta\left(\max\{\epsilon^{-2/d},\epsilon^{-1}\}\right),
\]
where the hidden constant is chosen large enough that $\Theta(i^d)\ge 5K_\theta$ for all $i\ge i_0$.

We split the sum at $i_0$.  For $i<i_0$, the load-one tail bound \eqref{eq:ptheta2} gives
\[
    P_\theta^{\mathrm{fix}}(\Theta(i^d))
    =
    O\left((i^d)^{-\chi_d}\right)
    =
    O(i^{-d\chi_d}).
\]
If $1\le d<2$, then $\chi_d=1/2$, so
\[
    \sum_{i<i_0} i^{-d\chi_d}
    =
    \sum_{i<i_0} i^{-d/2}
    =
    O(i_0^{1-d/2}).
\]
In this range $i_0=\Theta(\epsilon^{-2/d})$, and hence
\[
    i_0^{1-d/2}
    =
    O(\epsilon^{1-2/d}).
\]
If $d=2$, then $d\chi_d=1$ and $i_0=\Theta(\epsilon^{-1})$, so the same sum is
\[
   \sum_{i<i_0} i^{-d\chi_d} = \sum_{i<O(\epsilon^{-1})} i^{-1} = O(\log\epsilon^{-1}).
\]
If $d>2$, then $\chi_d=1/d$, so again $d\chi_d=1$, while $i_0=\Theta(\epsilon^{-1})$, so the sum is again
\[
    O(\log\epsilon^{-1}).
\]

For $i\ge i_0$, on the other hand, we may use the bound \eqref{eq:ptheta1} to deduce that
\[
    P_\theta^{\mathrm{fix}}(\Theta(i^d))
    \le
    O\left(
        \epsilon
        \exp\left(
            -\Omega\left(\min\{\epsilon^2 i^d,\epsilon i\}\right)
        \right)
    \right).
\]
Therefore
\[
\begin{aligned}
    \sum_{i\ge i_0}
        P_\theta^{\mathrm{fix}}(\Theta(i^d))
    &\le
    O\left(
        \epsilon\sum_{i\ge i_0}
            e^{-\Omega(\epsilon^2 i^d)}
    \right)
    +
    O\left(
        \epsilon\sum_{i\ge i_0}
            e^{-\Omega(\epsilon i)}
    \right).
\end{aligned}
\]
The final sum is at most
\[
    \epsilon\sum_{i\ge0}e^{-\Omega(\epsilon i)}=O(1).
\]
To bound the second to final sum, comparison with an integral gives
\[
    \sum_{i\ge0}e^{-\Omega(\epsilon^2 i^d)}
    =
    O(\epsilon^{-2/d}),
\]
so
\[
    \epsilon\sum_{i\ge i_0}e^{-\Omega(\epsilon^2 i^d)}
    =
    O(\epsilon^{1-2/d}).
\]
This is the desired bound when $d<2$, and it is $O(1)$ when $d\ge2$.

Finally, we bound the polynomial-completion contribution.  The tail bound at displacement $n/10$ gives
\[
    n^B P_\theta^{\mathrm{fix}}(n/10)
    \le
    n^B\cdot
    O\left(
        \epsilon
        \exp\left(
            -\Omega\left(\min\{\epsilon^2 n,\epsilon n^{1/d}\}\right)
        \right)
    \right).
\]
The theorem assumes
\[
    \epsilon
    \ge
    C_{\mathrm{fin}}\log n\cdot \max\{n^{-1/2},n^{-1/d}\}.
\]
Thus both
\[
    \epsilon^2n
    \qquad\text{and}\qquad
    \epsilon n^{1/d}
\]
are at least a sufficiently large constant times $\log n$, once $C_{\mathrm{fin}}$ is chosen large enough in terms of $B,d,\alpha$.  Hence the polynomial-completion contribution is $O(1)$.

Combining the estimates,
\[
    \E\Bk*{\text{query time}}
    =
    O(\epsilon^{1-2/d})
    \qquad (1\le d<2),
\]
and
\[
    \E\Bk*{\text{query time}}
    =
    O(\log\epsilon^{-1})
    \qquad (d\ge2).
\]
This proves the theorem.
\end{proofof}

\subsection{For \texorpdfstring{$d \ge 2$}{d >= 2}, A Single Offset Sequence Suffices}\label{sec:one-sequence-all-eps}

In this subsection, for $d \ge 2$, we strengthen the preceding analysis to show that almost every random fixed-offset degree-$d$ offset sequence achieves $O(\log \epsilon^{-1})$ expected query time simultaneously for all admissible load factors and table sizes (Theorem \ref{thm:one-sequence-all-eps}).

The main additional tool that we need is the following concentration inequality.

\begin{lemma}[A drift bound]
\label{lem:scale-drift}
Let $X_1,\ldots,X_j$ be nonnegative random variables adapted to a filtration $(\mathcal F_i)$.  Suppose $X_1\le1$, and for every $i\ge2$,
\[
    0\le X_i\le3X_{i-1}.
\]
Suppose also that whenever $X_{i-1}\ge1$,
\[
    \E\Bk*{X_i\mid\mathcal F_{i-1}}
    \le
    10^{-3}X_{i-1}.
\]
There are absolute constants $R_0,C_0<\infty$ such that, for every $R\ge R_0$,
\[
    \Pr\left[\sum_{i=1}^j X_i>Rj\right]\le C_0 R^{-4}j^{-2}.
\]
\end{lemma}

\begin{proof}
The indices with $X_i<1$ contribute at most $j$ to the sum.  We now control the maximal contiguous blocks of indices on which $X_i\ge1$.  Each such block begins with value at most $3$: either the block starts at $i=1$, or the previous value was below $1$, and the growth condition gives $X_i\le3X_{i-1}<3$.

Fix one block, and let $M$ be the sum of the values of $X_i$ in that block.  If the block begins at the stopping time $\tau$, let $Y_r$ be the value of the process $r$ steps after the beginning of the block, with $Y_r=0$ after the block ends, and let $\mathcal H_r=\mathcal F_{\tau+r}$.  Thus
\[
    M=\sum_{r\ge0}Y_r
    \qquad\text{and}\qquad
    Y_0\le3.
\]
For any $r$ such that $Y_r > 0$, the assumptions give $0\le Y_{r+1}\le3Y_r$ and
\[
    \E\Bk*{Y_{r+1}\mid\mathcal H_r}
    \le
    10^{-3}Y_r.
\]
Therefore
\[
    \E\Bk*{Y_{r+1}^4\mid\mathcal H_r}
    \le
    (3Y_r)^3\E\Bk*{Y_{r+1}\mid\mathcal H_r}
    \le
    27\cdot10^{-3}Y_r^4.
\]
Hence, for some absolute constant $\rho<1$, and for all $r \ge 0$, we have
\[
    \E\Bk*{Y_r^4}\le O(\rho^r).
\]
By Minkowski's inequality,
\[
    \left(\E\Bk*{M^4}\right)^{1/4}
    \le
    \sum_{r\ge0}\left(\E\Bk*{Y_r^4}\right)^{1/4}
    =
    O(1).
\]
Equivalently,
\[
    \E\Bk*{M^4}=O(1).
\]
The same argument holds after conditioning on the entire history before the block begins.

There are at most $j$ such blocks.  Let their sums be
\[
    M_1,M_2,\ldots,M_s,
    \qquad s\le j.
\]
Let $\mathcal G_\ell$ be the history through the end of the $\ell$-th block.  The conditional fourth-moment bound above gives
\[
    \E\Bk*{M_\ell^4\mid\mathcal G_{\ell-1}}\le C_1
\]
for an absolute constant $C_1$. 
Define
\[
    D_\ell=M_\ell-\E\Bk*{M_\ell\mid\mathcal G_{\ell-1}}.
\]
Then the $D_\ell$'s are martingale differences.  The conditional fourth-moment bound also gives the needed lower moments: Jensen's inequality gives $\E\Bk*{M_\ell\mid\mathcal G_{\ell-1}}^4\le C_1$, and Cauchy--Schwarz gives $\E\Bk*{M_\ell^2\mid\mathcal G_{\ell-1}}\le C_1^{1/2}$.  Using $(a-b)^2\le2a^2+2b^2$ and $(a-b)^4\le8a^4+8b^4$, we get
\[
    \E\Bk*{D_\ell^2\mid\mathcal G_{\ell-1}}=O(1),
    \qquad
    \E\Bk*{D_\ell^4\mid\mathcal G_{\ell-1}}=O(1).
\]

We claim that
\[
    \E\Bk*{\left(\sum_{\ell=1}^sD_\ell\right)^4}=O(j^2).
\]
To see this, expand the fourth power.  Since the $D_\ell$'s are martingale differences, every term in which some $D_\ell$ appears to the first power has expectation zero. The only remaining terms are of the forms
\[
    \E\Bk*{D_\ell^4}
    \qquad\text{and}\qquad
    \E\Bk*{D_a^2D_b^2}\quad(a<b).
\]
The first type contributes $O(j)$.  For the second type,
\[
    \E\Bk*{D_a^2D_b^2}
    =
    \E\Bk*{D_a^2\E\Bk*{D_b^2\mid\mathcal G_{b-1}}}
    =
    O(\E\Bk*{D_a^2})
    =
    O(1),
\]
so the total contribution is $O(j^2)$.  This proves the claim.

Also,
\[
    \sum_{\ell=1}^s\E\Bk*{M_\ell\mid\mathcal G_{\ell-1}}
    =
    O(j),
\]
because $s\le j$.  Thus, for $R$ sufficiently large, we can apply Markov's inequality to get
\[
\begin{aligned}
    \Pr\left[\sum_{\ell=1}^sM_\ell>Rj\right]
    &\le
    \Pr\left[\sum_{\ell=1}^sD_\ell>\Omega(Rj)\right] \\
    &\le
    \frac{\E\Bk*{(\sum_{\ell=1}^sD_\ell)^4}}{\Omega(R^4j^4)}
    =
    O(R^{-4}j^{-2}).
\end{aligned}
\]
Adding the deterministic contribution at most $j$ from the indices with $X_i<1$, and adjusting the constants as necessary proves the lemma.
\end{proof}

With the help of Lemma \ref{lem:longreveal}, we can now prove Theorem \ref{thm:one-sequence-all-eps}, restated below for convenience:

\onesequencealleps*

\begin{proof}
As in the proof of \cref{thm:anti-robinhood-expected-bounds}, we begin by considering a Poissonized table, i.e., a table with a Poisson random-variable number of keys. We will also, until the end of the proof, ignore the query-time contributions of keys with displacements greater than $n / 10$, since these must be handled with a separate argument.

It suffices to consider load parameters of the form
\[
    \epsilon_j=2^{-j},
    \qquad j=1,2,\ldots .
\]
Let
\[
    K_i=5^iK_{\min},
\]
where $K_{\min}$ is a sufficiently large constant.  For each $i$, define
\[
    A_i
    =
    K_i^{1/d}
    \cdot
    \frac1n
    \E\Bk*{\#\{x:\operatorname{disp}(x)>K_i\}},
\]
where the expectation is for the Poissonized table at expected load $1$. Note that $A_i$ is still a random variable, since it depends on the (still random) offset sequence being used. One can think of $A_i$ as capturing, for a given offset sequence at load factor $1$, the expected query cost spent traveling from displacement $K_{i - 1}$ to $K_i$.  Whenever $K_i\le n/10$, this quantity $A_i$ is determined by the Poisson arrivals in a length-$(2K_i)$ interval and therefore does not depend on the choice of $n$.  Thus $A_i$ may be viewed as a random variable depending only on the offset sequence, and not on $n$.

We will now argue that the sequence $A_i$ satisfies the conditions of Lemma~\ref{lem:scale-drift}, after division by a sufficiently large absolute constant.  First,
\[
    A_{i+1}\le5^{1/d}A_i\le3A_i,
\]
because $d\ge2$ and the quantity
$
    \frac1n
    \E\Bk*{\#\{x:\operatorname{disp}(x)>K\}}
$
is nonincreasing in $K$.  Second, if $A_i$ is larger than a sufficiently large constant, then the expected number of free slots in an interval of length $K_i$ in the $K_i$-reveal is
\[
    K_i\cdot \frac1n
    \E\Bk*{\#\{x:\operatorname{disp}(x)>K_i\}}
    =
    K_i^{(d-1)/d}A_i.
\]
Thus \cref{lem:negative-feedback-loop} applies at threshold $K_i$, giving
\[
    \E\Bk*{A_{i+1}\mid\text{offsets whose sampling intervals intersect }[0,K_i]}
    \le
    10^{-4}5^{1/d}A_i
    \le
    10^{-3}A_i.
\]
Therefore, by Lemma~\ref{lem:scale-drift}, for every $R\ge R_0$ and every $m\ge1$,
\begin{equation}
    \Pr\left[\sum_{i=1}^m A_i>Rm\right]\le O(R^{-4}m^{-2}).
    \label{eq:drift}
\end{equation}

Choose constants $R\ge R_0$ and $M\ge1$, both depending only on $\delta,d,\alpha$, large enough for the union bounds below; as a slight abuse of notation, we will sometimes include $R$ and $M$ in asymptotic notation to indicate that the hidden constants are independent of $R$ and $M$.  

Fix $j$, and consider the Poissonized table at load factor $1-\epsilon_j$.  We show that, except with probability $O(j^{-2})$ over the offset sequence, its expected query time is $O(j)$.  We first control displacements up to a transition scale using the drift event, then use Lemma~\ref{lem:longreveal} for larger displacements up to $n/10$, and finally handle displacements beyond $n/10$ by the polynomial-completion bound.

\paragraph{Displacements up to the transition scale.}
Since
\[
    \epsilon_j=2^{-j}
    \qquad\text{and}\qquad
    K_m^{1/d}=K_{\min}^{1/d}5^{m/d},
\]
let $m_j$ be the least positive integer such that
\begin{equation}
    K_{m_j}^{1/d}\ge MR\frac{j}{\epsilon_j}.
    \label{eq:mj-transition-scale}
\end{equation}
By minimality,
\[
    K_{m_j}^{1/d}
    =
    \Theta\left(MR\frac{j}{\epsilon_j}\right),
\]
and taking logarithms shows that $m_j=\Theta(j)$, where the implicit constants may depend on $M,R,d$, and $K_{\min}$, but not on $j$ or $n$.  Since $d\ge2$, choosing $M$ and $R$ sufficiently large also ensures
\[
    K_{m_j}
    \ge
    C_0\max\{\epsilon_j^{-2},\epsilon_j^{-d}\},
\]
as required by Lemma~\ref{lem:longreveal}.

Define
\[
    \mathcal C_j
    =
    \left\{
        \sum_{i=1}^{m_j}A_i\le Rm_j
    \right\}.
\]
The bound \eqref{eq:drift} gives
\begin{equation}
    \Pr[\neg\mathcal C_j]
    \le
    O(R^{-4}m_j^{-2})
    =
    O(R^{-4}j^{-2}).
    \label{eq:Cjprob}
\end{equation}

Fix an offset sequence for which $\mathcal C_j$ holds.  Since lowering the load can only decrease displacement tails, the expected query-time contribution from displacement scales at most $K_{m_j}$ in the table of load $1-\epsilon_j$ is bounded by the corresponding load-one quantities.  Hence it is
\begin{equation}
    O\left(1+\sum_{i=1}^{m_j}A_i\right)
    \le
    O(1+Rm_j)
    =
    O(Rj),
    \label{eq:Cjholds}
\end{equation}
where the expectation is over the keys and their hashes only.

\paragraph{Displacements from the transition scale to $n/10$.}
The event $\mathcal C_j$ also supplies the starting condition for Lemma~\ref{lem:longreveal}.  To see this, fix the offsets whose sampling intervals intersect $[0,K_{m_j}]$ and suppose that they satisfy $\mathcal C_j$.  The definition of $A_{m_j}$ is the load-one displacement tail for this fixed revealed prefix, multiplied by $K_{m_j}^{1/d}$.  Therefore, by monotonicity in the load,
\[
    \frac1n
    \E\Bk*{\#\{x:\operatorname{disp}(x)>K_{m_j}\}}
    \le
    \frac{A_{m_j}}{K_{m_j}^{1/d}}
    \le
    \frac{\sum_{i=1}^{m_j}A_i}{K_{m_j}^{1/d}}
    \le
    O\left(\frac{Rj}{K_{m_j}^{1/d}}\right)
    =
    O\left(\frac{\epsilon_j}{M}\right).
\]
Here the expectation is over the keys; offsets beyond the revealed prefix do not affect whether a displacement exceeds $K_{m_j}$.  Taking $M$ sufficiently large makes the last expression at most $\eta\epsilon_j$.

This bound is pointwise in every extension of the revealed prefix.  Since the displacement tail is nonincreasing in its threshold, after revealing any additional offsets the hypothesis of Lemma~\ref{lem:longreveal} continues to hold at every scale $K\ge K_{m_j}$.  Moreover, $K_{m_j}\ge\epsilon_j^{-d}$ implies
\[
    \epsilon_j^2K
    \ge
    \epsilon_jK^{1/d}
    \qquad\text{for every }K\ge K_{m_j},
\]
so the second term is the minimum in the exponent of Lemma~\ref{lem:longreveal}.  Consequently, for every $s\ge0$ such that $K_{m_j+s}\le n/10$, that lemma, followed by averaging over the revealed offsets conditioned on $\mathcal C_j$, gives
\[
    \E\Bk*{\frac1n
        \#\{x:\operatorname{disp}(x)>K_{m_j+s+1}\}
        \mid
        \mathcal C_j}
    \le
    \epsilon_j
    \exp\left(-\Omega\left(\epsilon_jK_{m_j+s}^{1/d}\right)\right).
\]
Because $K_{m_j+s}=5^sK_{m_j}$ and \eqref{eq:mj-transition-scale} holds,
\[
    \epsilon_jK_{m_j+s}^{1/d}
    =
    \Theta(5^{s/d}MRj).
\]
It follows that the conditional expected query-time contribution from these scales is at most
\[
\begin{aligned}
    &\sum_{\substack{s\ge0:\\K_{m_j+s}\le n/10}}
        K_{m_j+s+1}^{1/d}
        \E\Bk*{\frac1n
            \#\{x:\operatorname{disp}(x)>K_{m_j+s+1}\}
            \mid
            \mathcal C_j} \\
    &\qquad\le
    \sum_{s\ge0}
        O(5^{s/d}MRj)
        \exp\left(-\Omega(5^{s/d}MRj)\right)
    \le
    \exp(-\Omega(MRj)).
\end{aligned}
\]
The infinite sum on the right dominates the finite sum for every table size $n$, so this estimate holds simultaneously for all $n$ for which $K_{m_j}\le n/10$.  By Markov's inequality, conditioned on $\mathcal C_j$, the probability over the remaining random offsets that the expected contribution of these displacement scales exceeds $1$ is at most
\[
    \exp(-\Omega(MRj)).
\]

Combining this estimate with \eqref{eq:Cjprob} and \eqref{eq:Cjholds}, for a sufficiently large constant $C'$ we obtain
\begin{equation}
    \Pr\left[
        \E\Bk*{\text{truncated query time}\mid\text{offset sequence}}
        >
        C'Rj
    \right]
    \le
    O(R^{-4}j^{-2})+\exp(-\Omega(MRj)).
    \label{eq:prexpj}
\end{equation}
Here ``truncated'' means that only the contribution from displacements at most $n/10$ is counted.  The probability is over the random fixed-offset degree-$d$ offset sequence, while the inner expectation is over the keys and their hashes.

\paragraph{Displacements beyond $n/10$.}
Suppose, as in the theorem statement, that the table size $n$ satisfies
\begin{equation}
    \epsilon_j
    \ge
    C_{\mathrm{sim}}\log n\cdot\max\{n^{-1/2},n^{-1/d}\}.
    \label{eq:epsjlower}
\end{equation}
Under this assumption, $j=O(\log n)$.  Together with
$K_{m_j}^{1/d}=O(MRj/\epsilon_j)$, this shows that, after increasing
$C_{\mathrm{sim}}$ in terms of $M,R,d$, we have $K_{m_j}\le n/10$.
Thus the preceding short- and intermediate-displacement estimates apply to every admissible $n$.

Let $\tau_n$ be the completion threshold from \cref{lem:polynomial-completion}.  It has expectation $n^{O(1)}$ and, by its definition, is independent of the displacement tail at $n/10$, which is determined by the keys and by offset variables whose sampling intervals intersect $[0,n/10]$.  Moreover, every query has stopped by true displacement $\tau_n$, after at most $\tau_n+1$ probes.  Therefore, averaging over both the offset sequence and the keys, the expected contribution of keys with displacement greater than $n/10$ is at most
\begin{equation}
    n^B\cdot
    O\left(
        \epsilon_j
        \exp\left(
            -\Omega\left(
                \min\{\epsilon_j^2 n,\epsilon_j n^{1/d}\}
            \right)
        \right)
    \right).
    \label{eq:largedisp}
\end{equation}
The inequality \eqref{eq:epsjlower} implies
\[
    \min\{\epsilon_j^2 n,\epsilon_j n^{1/d}\}
    \ge
    \Omega(C_{\mathrm{sim}}\log n).
\]
Fix an integer $q\ge3$, to be chosen below.  Taking $C_{\mathrm{sim}}$ sufficiently large in terms of $q$ makes \eqref{eq:largedisp} at most $n^{-q}$.  Markov's inequality then shows that, for this pair $(j,n)$, the probability that the completion contribution exceeds $1$ is at most $n^{-q}$.  Every $n$ satisfying \eqref{eq:epsjlower} is at least $\epsilon_j^{-1}$, so a union bound over all admissible table sizes gives failure probability at most
\[
    \sum_{n\ge\epsilon_j^{-1}}n^{-q}
    =
    O(\epsilon_j^{q-1})
    =
    O(2^{-(q-1)j}).
\]
Combining this bound with \eqref{eq:prexpj}, the probability that there exists an admissible $n$ for which the expected query time at load $1-\epsilon_j$ exceeds $C'Rj+1$ is at most
\[
    O(R^{-4}j^{-2})
    +
    \exp(-\Omega(MRj))
    +
    O(2^{-(q-1)j}).
\]
Choose $R$, $M$, and $q$, and then $C_{\mathrm{sim}}$, sufficiently large in terms of $\delta,d,\alpha$.  Summing the preceding display over all $j\ge1$ then gives a total failure probability at most $\delta$.  Thus, with probability at least $1-\delta$, every dyadic slack $\epsilon_j$ and every $n$ satisfying \eqref{eq:epsjlower} have expected query time
\[
    O(j)=O(\log\epsilon_j^{-1}).
\]

To pass from dyadic to arbitrary slack, choose $j$ so that
\[
    \epsilon_j\le\epsilon<2\epsilon_j.
\]
The table at slack $\epsilon$ has no larger expected query time than the table at the smaller slack $\epsilon_j$.  After doubling the constant in the theorem's lower bound on $\epsilon$, the admissibility of $(\epsilon,n)$ implies the admissibility of $(\epsilon_j,n)$.  Since $j=\Theta(\log\epsilon^{-1})$, this proves the theorem for Poissonized tables.

Finally, since a table with $\operatorname{Poisson}((1 - \epsilon)n)$ keys has at least $\lfloor(1 - \epsilon)n\rfloor$ keys with probability $\Omega(1)$, we can conclude that any expected-query-time bound on the Poissonized table also applies (up to constant factors) to a table with exactly $\lfloor(1 - \epsilon)n\rfloor$ keys.  This completes the proof. 
\end{proof}

\subsection{Extending to Smoothed \texorpdfstring{Degree-$d$}{Degree-d} Probing}\label{subsec:anti-robinhood-smoothed-analysis}

Finally, we observe that the proof of \cref{thm:anti-robinhood-expected-bounds} extends to smoothed degree-$d$ probing. Recall that, in such a hash table, for each key $x$, each offset $k$ is included in \emph{that key's} offset sequence independently with probability $p(k) = 1/k^{1 - 1/d}$. Two local changes are needed. First, we use \cref{lem:displacement-and-probe-complexity} to conclude that, when bounding the expected time to query an element with displacement $k$, we may still bound it by $O(k^{1/d})$. Second, we replace \cref{lem:opportunity} with the following variation.  The final completion step is supplied directly by the smoothed clause of \cref{lem:polynomial-completion}.

\begin{lemma}[Smoothed opportunity bound]
  \label{lem:opportunitysmoothed}
  Fix a constant $d\ge1$ and consider a smoothed degree-$d$ probing hash table.  Fix a home position $a$ and integers $0\le K_-<K_+$.  Condition on the set of keys present, their home positions, their priorities, and every activation bit at an offset at most $K_-$.  Suppose that fewer than $\beta$ keys with home position $a$ are absent from the $K_-$-reveal.  Then, for every $S\ge0$,
  \[
      \Pr\left[
          \begin{array}{c}
          \text{some key with home position }a\text{ is absent from the }
          K_+\text{-reveal, and}\\
          \text{at least }S\text{ slots }z\in a+[K_-+1,K_+]\text{ are free in the }K_+\text{-reveal}
          \end{array}
      \right]
      \le
      \exp\left(O(\beta)-\Omega\left(S/K_+^{(d-1)/d}\right)\right).
  \]
\end{lemma}

\begin{proof}
  Reveal the activation bits one offset at a time, for $k=K_-+1,K_-+2,\ldots,K_+$.  Just before the bits at offset $k$ are exposed, call $a+k$ a \defn{$k$-opportunity for $a$} if it is free in the $(k-1)$-reveal.  Whether $a+k$ is a $k$-opportunity is determined by the information already exposed.

  Whenever a $k$-opportunity exists and at least one key with home position $a$ is still absent, choose the highest-priority such key and reveal its activation bit at offset $k$ first.  This bit is independent of the past and equals $1$ with probability
  \[
      p(k)\ge cK_+^{-(d-1)/d}
  \]
  for a constant $c=c(d)>0$.  If the bit equals $1$, the chosen key is placed at $a+k$: before offset $k$ the slot is free, and at true offset $k$ only keys with home position $a$ can probe it.  We may then reveal the remaining activation bits at offset $k$ in any order.

  Initially fewer than $\beta$ same-home keys are absent.  Hence, if one of them is still absent from the $K_+$-reveal, fewer than $\beta$ successful opportunity activations occurred.  Conversely, every slot in $a+[K_-+1,K_+]$ that is free in the $K_+$-reveal was a $k$-opportunity at its corresponding offset $k$.  Thus the event in the lemma with $S$ final free slots entails at least $S$ adaptive trials, each with head probability at least $cK_+^{-(d-1)/d}$, but fewer than $\beta$ heads.  Applying \cref{lem:losinglotsofflips} with $t=1$ and with $t_k$ equal to the corresponding conditional head probability proves the claim.
\end{proof}

Using \cref{lem:opportunitysmoothed} in place of \cref{lem:opportunity}, using \cref{lem:displacement-and-probe-complexity} to relate displacements to query times, and using the independent completion threshold $\tau_n(x)$ from \cref{lem:polynomial-completion} for the final tail, the rest of the proof of \cref{thm:anti-robinhood-expected-bounds} extends to smoothed degree-$d$ probing.\todo{The concentration bound proposition 3.5 works only for fixed sequence. We also need to change that by pack the probe sequence into the randomness of the keys to apply McDiarmid's inequality.}
\begin{corollary}
\label{cor:smoothed-arh-upper-bounds}
The conclusions of \cref{thm:anti-robinhood-expected-bounds} also hold for smoothed degree-$d$ probing.
\end{corollary}

\section{Upper Bounds for Robin Hood Ordering}
\label{sec:robinhood}
\label{sec:rh-analysis}

In this section, we prove our upper bound for smoothed degree-$d$ probing with Robin Hood ordering (which we also show to be asymptotically tight in \cref{sec:lower-bounds}).  Our main result is the following theorem, which establishes the expected displacement and expected probe complexity of every fixed key in the table.

\begin{restatable}[Robin Hood displacement and probe complexity]{theorem}{rhexpectedbounds}
  \label{thm:rh-expected-bounds}
  Fix a constant $d\ge1$ and a parameter $0<\eps\le1/2$, and let $n$ be an integer satisfying $n\ge100\eps^{-10d}$.  In the smoothed degree-$d$ probing hash table using Robin Hood ordering with $n$ slots at load $1-\eps$, every fixed key $x$ in the table satisfies
  \[
    \E\Bk*{\disp(x)}=O(\eps^{-1}), \qquad \E\Bk*{\pc(x)}=O(\eps^{-\frac{1}{d}}),
  \]
  where the hidden constants depend only on $d$.
\end{restatable}

As in the previous section, this also gives us an immediate corollary for amortized expected insertion time:
\begin{corollary}
  \label{cor:rh-amortized-expected-insertion-time}
  With the same setup as in Theorem \ref{thm:rh-expected-bounds}, the amortized expected insertion time at load factor $1 - \epsilon$ is
  $
    O(\eps^{-1}),
  $ where again the hidden constant depends only on $d$.
\end{corollary}

The section is organized into three subsections.  In \cref{subsec:rh-switching-to-poissonized-right-unbounded}, we set up two reductions that simplify the model under analysis: a Poissonization step that decouples the number of keys hashed to each slot, and a comparison step that replaces the standard cyclic table with a \emph{right-unbounded} table in which keys near slot $n$ continue probing past the end of the array instead of wrapping around.  Together these reductions let us work entirely in the more tractable Poissonized right-unbounded model.

In \cref{subsec:rh-frontier}, we analyze the left-to-right insertion process in that model via the \emph{Robin Hood frontier process}, which tracks the multiset of displacements of keys that have been hashed before the current slot but not yet been placed.  The main estimate of the subsection (\cref{lem:rh-frontier-tail}) is a stretched-exponential tail bound on both the size and the maximum element of the frontier; it is proved by an induction that couples a large frontier to a random walk with negative drift, while simultaneously controlling the largest surviving displacement.

Finally, in \cref{subsec:rh-displacement-upper}, we use the frontier tail bound to control the displacement of a distinguished key in the right-unbounded model (\cref{claim:rh-displacement-tail}), and then transfer the resulting tail estimate back to the standard cyclic table via the reductions of \cref{subsec:rh-switching-to-poissonized-right-unbounded}.  Integrating the tail yields the bounds on $\E\Bk*{\disp(x)}$ and $\E\Bk*{\pc(x)}$ stated in \cref{thm:rh-expected-bounds}.

\subsection{Reductions to the Poissonized Right-Unbounded Model}
\label{subsec:rh-switching-to-poissonized-right-unbounded}

According to \cref{lem:history-independence-of-rh-and-arh}, the hash table is history-independent, so the order of insertion does not matter. Therefore, one way to characterize the set of keys is to group them by their hash values, and represent the set of keys to be inserted as a vector of random variables $\bk{X_1, \ldots, X_n}$, where $X_i$ is the number of keys that are hashed to the slot $i$. The random variables $X_1, \ldots, X_n$ are not independent a priori, since they sum up to the total number of keys $(1 - \eps)n$. It is convenient to make them independent by \defn{Poissonizing} the number of keys.

\begin{definition}[Poissonized Set of Keys]
  \label{def:poissonized-set-of-keys}
  Fix a parameter $0<\eps\le1/2$ and a distinguished key $x$ whose displacement and probe complexity we wish to analyze. In the Poissonized model, the number of keys $N$ besides $x$ to be inserted is a Poisson random variable with mean $(1-\eps)n$. Equivalently, the number of other keys hashed to each slot $i$ is an independent Poisson random variable $X_i \sim \Poi(1-\eps)$. We obtain the displacement of $x$ by inserting all $N+1$ keys; if $N+1>n$, we define its displacement to be $n$.
\end{definition}

The following lemma directly reduces the standard fixed-load model to the Poissonized model.
\begin{lemma}[Poissonization reduction]
  \label{cor:reduction-to-poissonized}
  Fix a parameter $0<\eps\le1/2$.  Use $\disp^{(\textup{Poi})}(x)$ and $\disp^{(\textup{std})}(x)$ to denote the displacement of $x$ in the Poissonized and standard models, respectively.  For every integer $n\ge10\eps^{-2}$ and every $L>0$,
  \[
      \Pr\Bk*{\disp^{(\textup{std})}(x)\ge L}
      \le
      3\Pr\Bk*{\disp^{(\textup{Poi})}(x)\ge L}.
  \]
\end{lemma}
\begin{proof}
  Write $\lambda=(1-\eps)n$, $m=\floor*{\lambda}$, and $\Delta=\eps n/2$, and let
  \[
      \mathcal E=\{m\le N\le\lambda+\Delta\}.
  \]
  The assumption $n\ge10\eps^{-2}$ gives
  $\Delta/\sqrt{\lambda}\ge\sqrt{5/2}$.  A standard uniform local central-limit estimate for the Poisson distribution therefore gives
  \[
      \Pr[\mathcal E]\ge\frac13,
  \]
  with the integer endpoints rounded as above.

  We next couple the two hash tables on $\mathcal E$.  The standard instance contains at most $m$ keys other than $x$.  Fix a value of $N$ in this event, generate the $N$ other keys and all their probe sequences, and use the first required number of these keys for the standard instance.  The Poissonized instance is obtained from the standard instance by adding keys.  In the deferred-acceptance view from the proof of \cref{lem:history-independence-of-rh-and-arh}, adding proposals can only make slots more selective; it cannot move an old key to an earlier proposal in its probe sequence.  Thus, under this coupling,
  \[
      \disp^{(\textup{Poi})}(x)
      \ge
      \disp^{(\textup{std})}(x).
  \]
  The same domination holds after averaging over $N$ conditional on $\mathcal E$.  Therefore
  \[
  \begin{aligned}
      \Pr\Bk*{\disp^{(\textup{Poi})}(x)\ge L}
      &\ge
      \Pr[\mathcal E]\,
      \Pr\Bk*{\disp^{(\textup{Poi})}(x)\ge L\mid\mathcal E}\\
      &\ge
      \frac13\Pr\Bk*{\disp^{(\textup{std})}(x)\ge L},
  \end{aligned}
  \]
  which proves the claim.
\end{proof}


In smoothed degree-$d$ probing, the array inside is cyclic, in the sense that when a key reaches its right end (i.e., the slot $n$), it will wrap around to continue probing from the left end (i.e., the slot $1$). This is not convenient for the analysis.
A more natural way is to extend the array infinitely to the right, and when a key reaches the slot $n$, it will keep probing to the right until it finds a free slot. We call this the \defn{right-unbounded model}.

Below, we want to bridge the displacement of a key $x$ in the right-unbounded model and the standard model for \emph{Robin Hood} smoothed degree-$d$ probing. Before changing to the right-unbounded model, since the slots are symmetric in the standard model, the expected displacement of $x$ will keep the same even if we \emph{condition on} the event that $x$ is hashed to a specific slot, say $3n/4$. The reduction below to the right-unbounded model will make use of this condition that $x$ is hashed to the slot $3n/4$, which is not close to both the left end and the right end.

\begin{lemma}
  \label{lem:reduction-to-right-unbounded}
  Fix a constant $d\ge1$, a parameter $0<\eps\le1/2$, and an integer $n\ge100\eps^{-10d}$.  Use $\disp^{(\textup{Ru})}(x)$ and $\disp^{(\textup{std})}(x)$ to denote the displacement random variables of a key $x$ in the right-unbounded and standard models, respectively.  Condition on $h(x)=3n/4$ in a Robin Hood smoothed degree-$d$ probing table of load $1-\eps$.  There is a coupling between the two models and an event $\mathcal G$ such that, on $\mathcal G$ and on the following two events,
  \begin{itemize}
    \item In the right-unbounded model, there are at most $\bk*{1 - \frac{2\eps}{3}}\cdot \frac{n}{2}$ keys hashed to the slots $\Bk*{1, \frac{n}{2}}$. 
    \item In the right-unbounded model, there are at most $\frac{\eps \cdot n}{6}$ keys traveling beyond the slot $n$,
  \end{itemize}
  we have
  \begin{align*}
    \disp^{(\textup{std})}(x) \le 4 \cdot \disp^{(\textup{Ru})}(x).
  \end{align*}
  Moreover, conditioned on the two events above,
  \[
    \Pr\Bk*{\mathcal G} \ge 1 - \exp\bk*{-\Omega\bk*{\eps n^{\frac{1}{d}}}}.
  \]
\end{lemma}
\begin{proof}
  By \cref{lem:history-independence-of-rh-and-arh}, we may insert the keys in increasing order of their hash values.  Write $\mathcal E_1$ and $\mathcal E_2$ for the first and second events in the statement.

  \paragraph{The coupling.}
  Up to the moment when a key first crosses slot $n$, the standard cyclic table and the right-unbounded table use the same activation bits.  In the right-unbounded table, once a key crosses slot $n$, we mark it as \defn{overflowed} and do not expose its later probes; those later probes cannot affect any slot in $[n]$ in the right-unbounded model.

  We next specify the fresh randomness used by the standard table after a key wraps around to the left.  Let
  \[
    p_0 \defeq p(3n/2) = (3n/2)^{-\frac{d-1}{d}}.
  \]
  For each wrap-around kicking chain, draw a fresh block $(B_1,\ldots,B_{n/2})$ of independent Bernoulli$(p_0)$ random variables, one for each slot in $[1,n/2]$.  When the current key is tested at a slot $j\le n/2$ during this chain, its cyclic displacement is at most $3n/2$, so its true activation probability at that occurrence is at least $p_0$.  We use the monotone coupling in which $B_j=1$ forces the true activation bit for the current key at slot $j$ to be $1$; if $B_j=0$, the remaining activation randomness is filled in so that the marginal activation probability is correct.

  \paragraph{The good event.}
  During a wrap-around chain, let $E\subseteq[1,n/2]$ be the set of empty slots in $[1,n/2]$ just before the chain starts.  The set $E$ may depend on the right-unbounded process and on the pool blocks used by earlier wrap-around chains.  Define $\mathcal G$ to be the event that, for every wrap-around chain for which $\abs{E}\ge \frac{\eps n}{6}$, the fresh pool block for that chain has $B_j=1$ for at least one $j\in E$.

  We bound the probability of $\mathcal G$ conditional on $\mathcal E_1\cap \mathcal E_2$.  Fix any realization of the right-unbounded process satisfying $\mathcal E_1\cap\mathcal E_2$, and expose the wrap-around pool blocks adaptively.  Conditional on the past and on the current set $E$ with $\abs{E}\ge \frac{\eps n}{6}$, the probability that the current chain fails the requirement in $\mathcal G$ is at most
  \[
    (1-p_0)^{\frac{\eps n}{6}}
      \le
      \exp\bk*{-\Omega\bk*{\eps n^{\frac{1}{d}}}}.
  \]
  On $\mathcal E_2$, there are at most $\frac{\eps n}{6}$ wrap-around chains.  Hence a sequential union bound gives
  \[
    \Pr\Bk*{\mathcal G^c \mid \mathcal E_1\cap\mathcal E_2}
      \le
      \frac{\eps n}{6}\exp\bk*{-\Omega\bk*{\eps n^{\frac{1}{d}}}}
      \le
      \exp\bk*{-\Omega\bk*{\eps n^{\frac{1}{d}}}},
  \]
  where the last inequality uses $n\ge100\eps^{-10d}$.  This proves the stated conditional probability bound.

  \paragraph{Termination before $n/2$.}
  We now work deterministically on $\mathcal E_1\cap\mathcal E_2\cap\mathcal G$. We claim that every wrap-around kicking chain in the standard table terminates before reaching slot $n/2$.  Proceed by induction over the increasing-home insertion order.  Suppose all previous wrap-around chains terminated inside $[1,n/2]$.  Then the part of the standard table in $[n/2,n]$ is identical to the right-unbounded table, and every key that newly wraps around in the standard table is one of the overflowed keys in the right-unbounded table.

  At the start of the current wrap-around chain, every occupied slot in $[1,n/2]$ contains either a key whose hash location lies in $[1,n/2]$, or one of the keys that overflowed past slot $n$ in the right-unbounded table.  By $\mathcal E_1$ and $\mathcal E_2$, the number of such occupied slots is at most
  \[
    \bk*{1-\frac{2\eps}{3}}\frac n2 + \frac{\eps n}{6} = \frac n2 - \frac{\eps n}{6}.
  \]
  Thus there are at least $\frac{\eps n}{6}$ empty slots in $[1,n/2]$ at the start of the chain.  Since $\mathcal G$ occurs, the fresh pool block has $B_j=1$ at one of these empty slots.  That slot is activated for the current key when the chain reaches it, so the chain stops there.  This proves the induction claim.

  \paragraph{Comparing the displacement of $x$.}
  If $x$ does not overflow in the right-unbounded table, then all slots relevant to the placement of $x$ lie in $[n/2,n]$, where the coupled tables are identical.  Hence $\disp^{(\textup{std})}(x)=\disp^{(\textup{Ru})}(x)$.  If $x$ overflows in the right-unbounded table, then $h(x)=3n/4$ implies $\disp^{(\textup{Ru})}(x)\ge \frac n4$.  By the termination claim, the standard wrap-around chain of $x$ terminates before slot $n/2$, and in particular $\disp^{(\textup{std})}(x)\le n$.  Therefore
  \[
    \disp^{(\textup{std})}(x) \le n \le 4\disp^{(\textup{Ru})}(x).
  \]
  The desired comparison follows in both cases.
\end{proof}

Combining with Poissonization, it is helpful to study Robin Hood smoothed degree-$d$ probing in the following Poissonized right-unbounded model, where we have an array which is infinite to the right, labelled by all positive integers. There are $\Poi(1-\eps)$ keys ``hashed'' to each slot $i$ for $i \in \Z^+$, and they are inserted in increasing order of the slot index. The analysis of Robin Hood smoothed degree-$d$ probing in \cref{subsec:rh-frontier} just starts from this Poissonized right-unbounded model.



\subsection{The Robin Hood Frontier Process}
\label{subsec:rh-frontier}

To prove \cref{thm:rh-expected-bounds}, we begin in the \emph{Poissonized right-unbounded model}.  Thus the array is processed from left to right, there is no wraparound, and the number of keys with home slot $i$ is an independent random variable $X_i\sim\Poi(1-\eps)$. By history-independence, this left-to-right insertion order gives the same final state as any other insertion order in this right-unbounded instance.

To study the left-to-right insertion process in the Poissonized right-unbounded model, we keep track of the set of frontier of keys that have already passed the current slot but have not yet been placed.  After slot $t$ has been processed, let $S_t$ be the multiset of displacements of these unsettled keys.  Equivalently, the keys represented by $S_t$ are precisely the keys with home at most $t$ that are still unplaced.  The transition from $S_t$ to $S_{t+1}$ is:
\begin{enumerate}[label=\textup{(\roman*)}]
  \item increase every displacement in $S_t$ by one;
  \item add $X_{t+1}$ new copies of $0$;
  \item among the keys in the resulting multiset whose current offset is activated, delete one with largest displacement, if such a key exists.
\end{enumerate}
The deleted key is placed into slot $t+1$. We assume initially $S_0 = \emptyset$, and we use the convention that $\max(S_t)=0$ when $S_t=\emptyset$, where $\max(S_t)$ denotes the largest element from $S_t$.

The following lemma is the main estimate in the right-unbounded model.  Its two conclusions reinforce one another: a small frontier at an earlier time helps rule out very old surviving keys, while a bound on the largest displacement gives a uniform service rate whenever the frontier is large.

\begin{lemma}[Tail bound for the Robin Hood frontier]
  \label{lem:rh-frontier-tail}
  Fix a constant $d\ge1$ and a parameter $0<\eps\le1/2$.  There is a sufficiently large constant $k_0=k_0(d)$ such that, for every slot $t$ and every $k\ge k_0$ in the Poissonized right-unbounded model,
  \begin{align}
    \Pr\Bk*{\abs{S_t} > k\eps^{-1}}
      &\le \exp\bk*{-k^{\frac{1}{d}}/100}, \label{eq:rh-frontier-size-tail} \\
    \Pr\Bk*{\max(S_t) > 4k\eps^{-1}}
      &\le \exp\bk*{-k^{\frac{1}{d}}/100}. \label{eq:rh-frontier-max-tail}
  \end{align}
\end{lemma}

\begin{proof}
  Fix $t$ and $k\ge k_0$.  We prove \eqref{eq:rh-frontier-size-tail} and \eqref{eq:rh-frontier-max-tail} simultaneously by induction on $t$, assuming both estimates at every earlier slot and for every parameter at least $k_0$.  The base case $t=0$ is immediate.

  The proof follows the two mechanisms described above.  First, the induction hypothesis bounds the number and displacement of the keys present sufficiently far in the past; each such key then has many independent opportunities to be placed, which controls the current maximum displacement.  Second, once all frontier displacements are moderate, a large frontier supplies a probe at almost every slot.  Since the mean number of new arrivals per slot is only $1-\eps$, the frontier size is then dominated by a random walk with drift $-\eps/2$.

  \paragraph{Step 1: bounding the maximum displacement.}
  Let
  \[
    M\defeq \ceil*{4k\eps^{-1}}.
  \]
  If $t<M$, then every key in $S_t$ has displacement less than $M$, so the desired maximum bound is immediate.  Suppose henceforth that $t\ge M$.  Since displacements are integers, the event $\max(S_t)>4k\eps^{-1}$ implies $\max(S_t)\ge M$.  Any key witnessing this event was already present in $S_{t-M}$ and survived all of the next $M$ slots.  By the induction hypothesis with parameter $2k$,
  \begin{align}
    \Pr\Bk*{\abs{S_{t-M}}>2k\eps^{-1}}
      &\le \exp\bk*{-(2k)^{\frac{1}{d}}/100}, \label{eq:max-proof-good-size-revised}\\
    \Pr\Bk*{\max(S_{t-M})>8k\eps^{-1}}
      &\le \exp\bk*{-(2k)^{\frac{1}{d}}/100}. \label{eq:max-proof-good-max-revised}
  \end{align}
  For $k_0$ sufficiently large, each right-hand side is at most $\frac14\exp(-k^{\frac{1}{d}}/100)$.  Condition on the complementary good event and on the whole multiset $S_{t-M}$.

  Fix a key $x$ represented in $S_{t-M}$.  Throughout the interval $\{t-M+1,\ldots,t\}$, the relative Robin Hood priority of two keys already in $S_{t-M}$ is fixed: the key with the smaller home slot has the larger displacement.  Moreover, every key arriving after time $t-M$ has lower priority than $x$.  Expose the probes in this interval of the keys in $S_{t-M}$ that have higher priority than $x$.  Each such key can occupy at most one slot, so these keys block at most $\abs{S_{t-M}}\le2k\eps^{-1}$ of the $M$ slots.  Consequently, at least
  \[
    M-2k\eps^{-1}\ge2k\eps^{-1}
  \]
  slots remain at which $x$ would be placed if it probed: each such slot is either empty or occupied only by a lower-priority key.

  At every one of these available slots, the displacement of $x$ is at most
  \[
    \max(S_{t-M})+M \le 16k\eps^{-1},
  \]
  after increasing $k_0$ to absorb rounding.  The activation decisions of $x$ at the available slots are independent of the probes exposed above, and each has probability at least $p(16k\eps^{-1})$.  Hence, conditionally on the good event and on the higher-priority probes,
  \begin{align*}
    \Pr\Bk*{x\in S_t}
      &\le \bk*{1-p(16k\eps^{-1})}^{2k\eps^{-1}} \\
      &\le \exp\bk*{-(k\eps^{-1})^{\frac{1}{d}}/16}.
  \end{align*}
  A union bound over the at most $2k\eps^{-1}$ keys in $S_{t-M}$ makes this conditional probability at most
  \[
      \frac12\exp\bk*{-k^{\frac{1}{d}}/100}
  \]
  when $k_0$ is sufficiently large.  Combining this estimate with \eqref{eq:max-proof-good-size-revised} and \eqref{eq:max-proof-good-max-revised} proves \eqref{eq:rh-frontier-max-tail} at time $t$.

  \paragraph{Step 2: coupling the size process to a negative-drift walk.}
  It remains to prove \eqref{eq:rh-frontier-size-tail} at time $t$.  Put
  \[
    L\defeq \ceil*{10k\eps^{-2}},
    \qquad
    B\defeq 8k\eps^{-1}\log^{2d}(k\eps^{-1}).
  \]
  The guiding idea is that, while $\max(S_{i-1})\le B$, a frontier of size at least $\frac12k\eps^{-1}$ contains many keys, each with a nonnegligible chance to probe slot $i$.  With probability at least $1-\eps/2$, some old key is therefore placed.  This service rate exceeds the mean arrival rate $1-\eps$ by $\eps/2$, creating negative drift.

  For each slot $i$, let $Y_i\in\{0,1\}$ indicate that some key is placed into slot $i$ during the transition from $S_{i-1}$ to $S_i$.  Then
  \begin{equation}
    \label{eq:frontier-size-transition-revised}
    \abs{S_i}=\abs{S_{i-1}}+X_i-Y_i.
  \end{equation}

  Conditional on the past, let $q_i$ be the probability that none of the keys already in $S_{i-1}$ probes slot $i$.  Draw independent random variables $U_i\sim\operatorname{Unif}[0,1]$, also independent of the arrivals, and use a monotone coupling in which the event that at least one old key probes slot $i$ is realized as $\{U_i\le1-q_i\}$.  Define
  \[
    \tilde Y_i\defeq \ind\nolimits_{\BK*{U_i\le1-\eps/2}},
    \qquad
    Z_i\defeq \ind\nolimits_{\BK*{q_i>\eps/2}},
    \qquad
    W_i\defeq X_i-\tilde Y_i.
  \]
  If $q_i\le\eps/2$, then $1-q_i\ge1-\eps/2$, so $\tilde Y_i=1$ forces an old key to probe and therefore forces $Y_i=1$.  If $q_i>\eps/2$, subtracting $Z_i=1$ makes the same lower bound trivial.  Thus, in all cases,
  \begin{equation}
    \label{eq:Yi-coupling-revised}
    Y_i\ge \tilde Y_i-Z_i.
  \end{equation}
  The variables $W_i$ are independent and satisfy
  \begin{equation}
    \label{eq:Wi-mean-var-revised}
    \E\Bk*{W_i}=-\eps/2,
    \qquad
    \Var(W_i)\le2.
  \end{equation}

  We next verify that the error term $Z_i$ vanishes whenever the frontier is large and its maximum displacement is controlled.  If
  \[
    \abs{S_{i-1}}\ge \frac12 k\eps^{-1}
    \qquad\text{and}\qquad
    \max(S_{i-1})\le B,
  \]
  then every old key probes slot $i$ with probability at least $p(B+1)$, independently of the others.  Consequently,
  \begin{align*}
    q_i
      &\le \bk*{1-p(B+1)}^{\frac12 k\eps^{-1}} \\
      &\le \exp\bk*{-\frac12 k\eps^{-1}(B+1)^{-\frac{d-1}{d}}} \\
      &\le \exp\bk*{-(k\eps^{-1})^{\frac{1}{2d}}}
       \le \eps/2,
  \end{align*}
  where the last two inequalities hold for $k_0$ sufficiently large as a function of $d$.  Hence $Z_i=0$, and \eqref{eq:frontier-size-transition-revised}--\eqref{eq:Yi-coupling-revised} give
  \[
    \abs{S_i}-\abs{S_{i-1}}\le W_i.
  \]
  Thus $W_i$ is the promised negative-drift comparison increment whenever the frontier is both large and low-maximum.

  \paragraph{Step 3: reducing a large frontier to four exceptional events.}
  Let
  \[
    a\defeq\max\{0,t-L\},
    \qquad
    T\defeq t-a=\min\{t,L\}.
  \]
  We claim that $\abs{S_t}>k\eps^{-1}$ implies at least one of the following events:
  \begin{align*}
    \mathcal F_1 &\defeq
      \BK*{\abs{S_a}>2k\eps^{-1}}, \\
    \mathcal F_2 &\defeq
      \BK*{\exists i\in\BK{a,\ldots,t-1}:
        \max(S_i)>B}, \\
    \mathcal F_3 &\defeq
      \BK*{T=L\ \text{ and }\
        \sum_{i=t-L+1}^{t} W_i>-k\eps^{-1}}, \\
    \mathcal F_4 &\defeq
      \BK*{\exists m\in\{1,\ldots,T\}:
        \sum_{i=t-m+1}^{t} W_i>\frac13 k\eps^{-1}}.
  \end{align*}
  The four events have a direct interpretation: the frontier is already too large at the start of the window, the maximum-displacement control fails somewhere in the window, the comparison walk does not realize its expected decrease over a full window, or a terminal segment of that walk has an unusually large upward excursion.

  Suppose none of $\mathcal F_1,\ldots,\mathcal F_4$ occurs.  First assume that
  \[
    \abs{S_i}\ge \frac12k\eps^{-1}
    \qquad\text{for every }i\in\{a,\ldots,t-1\}.
  \]
  Since $S_0=\emptyset$, this case forces $a>0$ and hence $T=L$.  The failure of $\mathcal F_2$ makes the comparison from Step 2 valid on every transition from $a+1$ through $t$.  Therefore, using the failures of $\mathcal F_1$ and $\mathcal F_3$,
  \[
    \abs{S_t}
      \le \abs{S_a}+\sum_{i=a+1}^{t}W_i
      \le 2k\eps^{-1}-k\eps^{-1}
      =k\eps^{-1}.
  \]

  Otherwise, let $s$ be the last index in $\{a,\ldots,t-1\}$ with $\abs{S_s}\le\frac12k\eps^{-1}$.  For each transition $i\in\{s+2,\ldots,t\}$, the previous frontier is larger than $\frac12k\eps^{-1}$; the failure of $\mathcal F_2$ again makes the comparison from Step 2 valid.  At the first transition after $s$, we use only
  \[
    X_{s+1}=W_{s+1}+\tilde Y_{s+1}\le W_{s+1}+1.
  \]
  Since $t-s\le T$, the failure of $\mathcal F_4$ now gives
  \begin{align*}
    \abs{S_t}
      &\le \abs{S_s}+1+\sum_{i=s+1}^{t}W_i \\
      &\le \frac12k\eps^{-1}+1+\frac13k\eps^{-1}
       \le k\eps^{-1},
  \end{align*}
  where the last inequality holds for $k\eps^{-1}\ge6$.  This proves the claimed implication.  In particular, it formalizes the drift intuition: a frontier that remains large for a whole window can stay large only if the walk fails to drift downward, while a frontier that became large more recently requires a large upward fluctuation after its last visit below the half-threshold.

  \paragraph{Step 4: bounding the four exceptional events.}
  If $a=0$, then $\mathcal F_1$ is impossible.  Otherwise, the induction hypothesis with parameter $2k$ gives
  \[
    \Pr\Bk*{\mathcal F_1}
      \le \exp\bk*{-(2k)^{\frac{1}{d}}/100}
      \le \frac14\exp\bk*{-k^{\frac{1}{d}}/100}.
  \]
  For $\mathcal F_2$, apply the induction hypothesis for the maximum tail at each fixed $i<t$ with parameter
  \[
    2k\log^{2d}(k\eps^{-1}).
  \]
  The corresponding maximum threshold is $B$.  Hence
  \[
    \Pr\Bk*{\max(S_i)>B}
      \le \exp\bk*{-(2k)^{\frac{1}{d}}\log^2(k\eps^{-1})/100}.
  \]
  A union bound over at most $L$ indices gives
  \[
    \Pr\Bk*{\mathcal F_2}
      \le \frac14\exp\bk*{-k^{\frac{1}{d}}/100}
  \]
  by the choice of $k_0$.

  If $T<L$, then $\mathcal F_3$ is empty.  If $T=L$, the sum in $\mathcal F_3$ has mean at most $-5k\eps^{-1}$ and variance at most $2L$.  The standard Chernoff bound for a centered Poisson-minus-Bernoulli sum gives, for any interval $I$ and any $u\ge0$,
  \begin{equation}
    \label{eq:poisson-bernoulli-chernoff-revised}
    \Pr\Bk*{\sum_{i\in I}(W_i-\E\Bk*{W_i})>u}
      \le \exp\bk*{-\Omega\bk*{\min\{u^2/\abs{I},u\}}}.
  \end{equation}
  Taking $I=\{t-L+1,\ldots,t\}$ and $u\ge4k\eps^{-1}$ yields
  \[
    \Pr\Bk*{\mathcal F_3}
      \le \exp(-\Omega(k))
      \le \frac14\exp\bk*{-k^{\frac{1}{d}}/100}.
  \]

  Finally, we bound $\mathcal F_4$ by an exponential supermartingale.  Let $\theta\defeq\eps/16$ and let $W=X-\tilde Y$, where $X\sim\Poi(1-\eps)$ and $\tilde Y\sim\operatorname{Bernoulli}(1-\eps/2)$ are independent.  Then
  \begin{align*}
    \log \E\Bk*{e^{\theta W}}
      &=(1-\eps)(e^\theta-1)
        +\log\bk*{\eps/2+(1-\eps/2)e^{-\theta}} \\
      &\le -\eps\theta/2+2\theta^2
       \le0.
  \end{align*}
  For $m=0,1,\ldots,T$, define
  \[
    M_m\defeq
      \exp\bk*{\theta\sum_{j=1}^{m}W_{t-j+1}}
  \]
  and let $\mathcal H_m=\sigma(W_t,W_{t-1},\ldots,W_{t-m+1})$.  Independence of the $W_i$'s and the preceding moment bound show that $(M_m)_{m=0}^{T}$ is a nonnegative supermartingale with respect to $(\mathcal H_m)_{m=0}^{T}$, with $M_0=1$. 
  The bound of $\mathcal F_4$ follows from the standard maximal inequality below.
\begin{theorem}[Ville's maximal inequality \cite{ville1939etude}]
  \label{thm:ville-maximal-inequality}
  Let $(M_s)_{s=0}^{T}$ be a nonnegative supermartingale with respect to a filtration $(\mathcal F_s)_{s=0}^{T}$.  Then, for every $a>0$,
  \[
    \Pr\Bk*{\max_{0\le s\le T} M_s\ge a}
      \le \frac{\E\Bk*{M_0}}{a}.
  \]
\end{theorem}
Applying \cref{thm:ville-maximal-inequality} at level $\exp(\theta k\eps^{-1}/3)$ gives
  \[
    \Pr\Bk*{\mathcal F_4}
      \le \exp\bk*{-\frac{\theta k\eps^{-1}}{3}}
      =\exp(-k/48)
      \le \frac14\exp\bk*{-k^{\frac{1}{d}}/100}.
  \]

  A union bound over $\mathcal F_1,\ldots,\mathcal F_4$ proves \eqref{eq:rh-frontier-size-tail} at time $t$, completing the induction and the proof of the lemma.
\end{proof}

\subsection{Displacement Upper Bound}
\label{subsec:rh-displacement-upper}

We continue in the Poissonized right-unbounded model of the previous subsection. The frontier estimate gives a tail bound for a distinguished key.  Only after this claim do we transfer the result back to the standard cyclic table.

\begin{claim}[Right-unbounded displacement tail]
  \label{claim:rh-displacement-tail}
  Fix a constant $d\ge1$ and a parameter $0<\eps\le1/2$.  Let $x$ be a distinguished key with home slot $t$ in the Poissonized right-unbounded model.  For every $k\ge k_0(d)$,
  \[
    \Pr\Bk*{\disp^{(\textup{Ru})}(x)>2k\eps^{-1}} \le 2\exp\bk*{-k^{\frac{1}{d}}/100}.
  \]
\end{claim}

\begin{proof}
  Expose the table formed by all keys except $x$, and then insert $x$ with an independent probe sequence.  Same-home ties may be broken against $x$, which can only increase its displacement and is therefore harmless for an upper bound.  By \cref{lem:rh-frontier-tail},
  \[
    \Pr\Bk*{\abs{S_t}>k\eps^{-1}} \le \exp\bk*{-k^{\frac{1}{d}}/100}.
  \]
  Condition on the complementary event and on the whole frontier $S_t$.  In the next $\ceil*{2k\eps^{-1}}$ slots, at most $k\eps^{-1}$ slots can be occupied by keys that have higher Robin Hood priority than $x$.  Thus there are at least $k\eps^{-1}$ slots in this interval where $x$ would be placed if it probed the slot.  At each such slot, the displacement of $x$ is at most $2k\eps^{-1}$ up to rounding, so the activation probability is at least $p(2k\eps^{-1})$ after changing constants.  Therefore
  \begin{align*}
    \Pr\Bk*{\disp^{(\textup{Ru})}(x)>2k\eps^{-1}\mid \abs{S_t}\le k\eps^{-1},S_t}
      &\le \bk*{1-p(2k\eps^{-1})}^{k\eps^{-1}} \\
      &\le \exp\bk*{-(k\eps^{-1})^{\frac{1}{d}}/2}
       \le \exp\bk*{-k^{\frac{1}{d}}/100},
  \end{align*}
  for $k_0$ sufficiently large.  Combining the two estimates proves the claim.
\end{proof}

We now transfer \cref{claim:rh-displacement-tail} from the Poissonized right-unbounded model to the standard cyclic table using the reductions from \cref{subsec:rh-switching-to-poissonized-right-unbounded}.  The resulting corollary controls both the mean displacement and its $1/d$-moment; by \cref{lem:displacement-and-probe-complexity}, the latter is the quantity that determines expected probe complexity.

\begin{corollary}[Standard-model displacement upper bounds]
  \label{cor:rh-displacement-upper}
  Fix a constant $d\ge1$, a parameter $0<\eps\le1/2$, and an integer $n\ge100\eps^{-10d}$.  In a standard Robin Hood smoothed degree-$d$ probing hash table with $n$ slots and load $1-\eps$, every fixed key $x$ satisfies
  \[
    \E\Bk*{\disp(x)}=O(\eps^{-1}), \qquad \E\Bk*{\disp(x)^{\frac{1}{d}}}=O(\eps^{-\frac{1}{d}}).
  \]
\end{corollary}

\begin{proof}
  By symmetry of the cyclic table, we may condition on $h(x)=3n/4$.  Applying the Poissonization reduction \cref{cor:reduction-to-poissonized} conditionally on this home slot, it is enough to control the corresponding Poissonized cyclic table up to a factor $3$ in tail probabilities.

  Couple the Poissonized cyclic table to the Poissonized right-unbounded table as in \cref{lem:reduction-to-right-unbounded}.  Let $\mathcal E_0$ be the event that the Poissonized number of keys is at most $n$.  Recall that $\mathcal E_1$ is the event that at most $(1-2\eps/3)n/2$ keys hash to $[1,n/2]$ in the right-unbounded model, and $\mathcal E_2$ is the event that at most $\eps n/6$ keys travel beyond slot $n$.  Let $\mathcal G$ be the good event from that lemma.  Standard Chernoff bounds give
  \[
    \Pr\Bk*{\mathcal E_0^c}+\Pr\Bk*{\mathcal E_1^c} \le \exp(-\Omega(\eps^2 n)).
  \]
  Moreover, the number of keys traveling beyond slot $n$ in the right-unbounded table is $\abs{S_n}$, so \cref{lem:rh-frontier-tail} with $k=\eps^2 n/6$ gives
  \[
    \Pr\Bk*{\mathcal E_2^c}
      \le
      \exp\bk*{-\Omega\bk*{(\eps^2 n)^{\frac{1}{d}}}}.
  \]
  Finally, \cref{lem:reduction-to-right-unbounded} gives
  \[
    \Pr\Bk*{\mathcal G^c\mid \mathcal E_1\cap\mathcal E_2}
      \le
      \exp\bk*{-\Omega\bk*{\eps n^{\frac{1}{d}}}}.
  \]
  Write
  \[
    \eta_n
      \defeq
      \exp(-\Omega(\eps^2 n))
      +\exp\bk*{-\Omega\bk*{(\eps^2 n)^{\frac{1}{d}}}}
      +\exp\bk*{-\Omega\bk*{\eps n^{\frac{1}{d}}}}.
  \]
  On $\mathcal E_0\cap\mathcal E_1\cap\mathcal E_2\cap\mathcal G$, the cyclic displacement is at most four times the right-unbounded displacement.  Hence, for every $T\ge8k_0\eps^{-1}$,
  \begin{equation}
    \label{eq:standard-rh-displacement-tail}
    \Pr\Bk*{\disp(x)>T}
      \le
      O\left(
        \exp\bk*{-\Omega\bk*{(\eps T)^{\frac{1}{d}}}}
        +\eta_n
      \right),
  \end{equation}
  where we used \cref{claim:rh-displacement-tail} with $k=\Theta(\eps T)$ and absorbed the constant factor from Poissonization.  All hidden constants in this proof may depend on the fixed constant $d$.

  Although the comparison with the right-unbounded model is valid for the entire tail, the additive bad-event probability $\eta_n$ in \eqref{eq:standard-rh-displacement-tail} cannot be summed over infinitely many values of $T$.  We therefore use a separate unconditional bound beyond one full cycle.  Insert $x$ last, which is allowed by history independence, and condition on the table formed by all other keys.  There are at least $\eps n$ empty physical slots.  For $T\ge n$, the offsets at most $T$ contain at least $\frac12\eps T$ occurrences of these empty slots.  If any such occurrence is activated, then $x$ is placed no later than that offset.  Since each occurrence has activation probability at least $p(T)$ and the activation bits are independent,
  \begin{equation}
    \label{eq:standard-crude-long-tail}
    \Pr\Bk*{\disp(x)>T}
      \le
      \exp\bk*{-\Omega\bk*{\eps T^{\frac{1}{d}}}}
      \qquad(T\ge n).
  \end{equation}

  Summing the tail probabilities, \eqref{eq:standard-rh-displacement-tail} contributes $O(\eps^{-1})$ over $T\le n$, while the additive term contributes only $n\eta_n=O(\eps^{-1})$ under the assumption $n\ge100\eps^{-10d}$.  A dyadic decomposition in \eqref{eq:standard-crude-long-tail} gives
  \[
      \sum_{T\ge n}
          \exp\bk*{-\Omega\bk*{\eps T^{\frac{1}{d}}}}
      =O\left(
          n\sum_{\ell\ge0}
              2^\ell
              \exp\bk*{-\Omega\bk*{\eps(2^\ell n)^{\frac{1}{d}}}}
        \right)
      =O(\eps^{-1}),
  \]
  again using $n\ge100\eps^{-10d}$.  Therefore $\E\Bk*{\disp(x)}=O(\eps^{-1})$.  Jensen's inequality then gives
  \[
    \E\Bk*{\disp(x)^{\frac{1}{d}}} \le \E\Bk*{\disp(x)}^{\frac{1}{d}} =O(\eps^{-\frac{1}{d}}),
  \]
  completing the proof.
\end{proof}

We are now ready to assemble the upper bounds of \cref{thm:rh-expected-bounds}, restated below for convenience:

\rhexpectedbounds*
\begin{proof}[Proof of \cref{thm:rh-expected-bounds}]
  Fix a key $x$ in the standard Robin Hood smoothed degree-$d$ probing hash table with $n$ slots at load $1-\eps$.  \cref{cor:rh-displacement-upper} gives
  \[
    \E\Bk*{\disp(x)}=O(\eps^{-1})
    \qquad\text{and}\qquad
    \E\Bk*{\disp(x)^{\frac{1}{d}}}=O(\eps^{-\frac{1}{d}}),
  \]
  which is the first of the two claimed bounds.  For the probe-complexity bound, \cref{lem:displacement-and-probe-complexity} gives
  \[
    \E\Bk*{\pc(x)}
    =\Theta\bk*{\E\Bk*{\disp(x)^{\frac{1}{d}}}+1}
    =O(\eps^{-\frac{1}{d}}),
  \]
  as desired.
\end{proof}

\section{Lower Bounds for Both Orderings}
\label{sec:lower-bounds}

In this section, we prove the matching lower bounds for both \cref{thm:anti-robinhood-expected-bounds} (anti-Robin Hood random fixed-offset degree-$d$ probing) and \cref{thm:rh-expected-bounds} (Robin Hood smoothed degree-$d$ probing).  

We begin in Subsections \ref{subsec:rh-displacement-lower} and \ref{subsec:query-time-from-displacement} with a sequence of arguments that gives tight lower bounds both for anti-Robin Hood when $d < 2$ and for Robin Hood for all $d$. An interesting feature of this argument is that, to prove \emph{lower bounds}, we make critical use of the tail bounds in our \emph{upper bounds} from earlier. The argument has two components.  First, we prove an ordering-independent lower bound on the sum of the \emph{displacements} of all keys. Then, we prove that in order for the upper tail bounds from earlier sections to hold, the only way for our displacement lower bounds to be true is if we also get corresponding lower bounds on expected probe complexity. This gives us the desired lower bounds for anti-Robin Hood when $d < 2$ and for Robin Hood for all $d$.

For anti-Robin Hood with $d\ge2$, however, the resulting $\Omega(\eps^{-1})$ lower bound on expected displacement is not tight enough to imply the logarithmic query-time lower bound that we need. To handle that regime, we prove a more general result, showing that the expected query time in \emph{any} fixed-offset anti-Robin Hood probing scheme (no matter whether it uses polynomial offsets or not) is $\Omega(\log \epsilon^{-1})$. This implies that Theorem \ref{thm:anti-robinhood-expected-bounds} is asymptotically tight for all $d\ge2$.

\subsection{Displacement Lower Bound}
\label{subsec:rh-displacement-lower}

We begin with a bound that does not depend on either the probing scheme or the collision-ordering rule.  Orient the table cyclically in the probing direction.  Whenever an interval receives more keys than it has slots, every excess key must cross the outgoing edge of the interval, and these forced crossings contribute directly to total displacement.

\begin{lemma}[Total-displacement lower bound]
  \label{lem:rh-total-displacement-lower}
  Fix a parameter $0<\eps\le1/2$ and an integer $n\ge10\eps^{-2}$.  For any probing scheme and any ordering in a hash table with $n$ slots and load $1-\eps$,
  \[
    \E\Bk*{\sum_{x}\disp(x)}=\Omega(n\eps^{-1}),
  \]
  where the sum is over all keys in the table.
\end{lemma}

\begin{proof}
  We view the slots cyclically and orient every edge from a slot to its successor in the probing direction.  For an edge $e_i$ from slot $i$ to slot $i+1$, let $C_i$ be the number of keys whose displacement path crosses $e_i$, counted with multiplicity if a key wraps around more than once.  Then
  \begin{equation}
    \label{eq:total-displacement-edge-crossings}
    \sum_x \disp(x)=\sum_{i=1}^{n} C_i.
  \end{equation}

  Fix an integer $L\le n/2$, and let $I_i$ be the cyclic interval of $L$ slots ending at $i$.  Let $X_{i,L}$ be the number of keys whose home slots lie in $I_i$.  At most $L$ of these keys can be stored in physical slots of $I_i$. Every remaining key must leave $I_i$ in the probing direction and therefore must cross the edge $e_i$ at least once.  Hence, deterministically,
  \begin{equation}
    \label{eq:edge-overload-lower-bound}
    C_i\ge \bk*{X_{i,L}-L}_+.
  \end{equation}

  Take $L=\floor*{\eps^{-2}}$; the assumption on $n$ ensures $L\le n/2$.  For a fixed edge $i$, the random variable $X_{i,L}$ has binomial distribution
  \[
    X_{i,L}\sim \operatorname{Binomial}\bk*{(1-\eps)n,\frac{L}{n}}.
  \]
  Its mean is $(1-\eps)L=L-\Theta(\eps^{-1})$ and its standard deviation is $\Theta(\eps^{-1})$.  Consequently,
  \begin{equation}
    \label{eq:binomial-overload-positive-part}
    \E\Bk*{\bk*{X_{i,L}-L}_+}=\Omega(\eps^{-1}).
  \end{equation}
  For completeness, when $\eps$ is below a sufficiently small constant, a uniform Berry--Esseen bound gives constants $c_1,c_2>0$ such that
  \[
      \Pr\Bk*{X_{i,L}\ge L+c_1\eps^{-1}}\ge c_2.
  \]
  This implies \eqref{eq:binomial-overload-positive-part}.  When $\eps$ is bounded away from zero, both $L$ and $\eps^{-1}$ are bounded by constants, and the same conclusion follows by decreasing the implicit constant.

  Combining \eqref{eq:total-displacement-edge-crossings}, \eqref{eq:edge-overload-lower-bound}, and \eqref{eq:binomial-overload-positive-part},
  \[
    \E\Bk*{\sum_x\disp(x)} =\sum_{i=1}^{n}\E\Bk*{C_i} \ge n\cdot\Omega(\eps^{-1}) =\Omega(n\eps^{-1}),
  \]
  as claimed.
\end{proof}

\subsection{Query-Time Lower Bounds from Displacement}\label{subsec:query-time-from-displacement}

We next combine \cref{lem:rh-total-displacement-lower} with the upper-tail estimates proved earlier.  Those estimates rule out the possibility that the mean displacement is carried by an extremely small number of exceptionally far-displaced keys.  Truncating at the natural displacement scale then converts the first-moment lower bound into a lower bound on the $1/d$-moment.

\begin{corollary}[Lower bound for Robin Hood smoothed degree-$d$ probing]
  \label{cor:rh-displacement-lower}
  Fix a constant $d\ge1$, a parameter $0<\eps\le1/2$, and an integer $n\ge100\eps^{-10d}$.  In a Robin Hood smoothed degree-$d$ probing hash table with $n$ slots and load $1-\eps$, every fixed key $x$ satisfies
  \[
    \E\Bk*{\disp(x)}=\Omega(\eps^{-1}), \qquad \E\Bk*{\disp(x)^{\frac{1}{d}}}=\Omega(\eps^{-\frac{1}{d}}).
  \]
\end{corollary}

\begin{proof}
  Let $m=(1-\eps)n$ be the number of keys, ignoring harmless rounding.  By exchangeability of the keys and \cref{lem:rh-total-displacement-lower},
  \[
    \E\Bk*{\disp(x)} =\frac1m\E\Bk*{\sum_y\disp(y)} =\Omega(\eps^{-1}).
  \]
  Write $D=\disp(x)$ and let $c>0$ be such that $\E\Bk*{D}\ge c\eps^{-1}$.  The tail estimates in the proof of \cref{cor:rh-displacement-upper} imply that, for a sufficiently large constant $k_0=k_0(d)$,
  \[
    \E\Bk*{D\ind\nolimits_{\BK*{D>k_0\eps^{-1}}}} \le \frac{c}{2}\eps^{-1}.
  \]
  Consequently,
  \[
    \E\Bk*{D\ind\nolimits_{\BK*{D\le k_0\eps^{-1}}}} \ge \frac{c}{2}\eps^{-1}.
  \]
  On the event $D\le k_0\eps^{-1}$,
  \[
    D^{\frac{1}{d}} \ge \frac{D}{(k_0\eps^{-1})^{\frac{d-1}{d}}}.
  \]
  Therefore
  \[
    \E\Bk*{D^{\frac{1}{d}}}
      \ge
      \E\Bk*{D^{\frac{1}{d}}\ind\nolimits_{\BK*{D\le k_0\eps^{-1}}}}
      =\Omega(\eps^{-\frac{1}{d}}),
  \]
  as desired.  By \cref{lem:displacement-and-probe-complexity}, this also gives $\E\Bk*{\pc(x)}=\Omega(\eps^{-1/d})$.
\end{proof}

\begin{corollary}[Lower bound for anti-Robin Hood random fixed-offset degree-$d$ probing when \texorpdfstring{$d<2$}{d < 2}]
  \label{cor:anti-robinhood-displacement-lower}
  Fix a constant $d\in[1,2)$, a parameter $0<\eps\le1/2$, and an integer $n$ satisfying $C_{\mathrm{fin}}\log n\cdot\max\{n^{-1/2},n^{-1/d}\}\le\eps$.  In an anti-Robin Hood random fixed-offset degree-$d$ probing hash table with $n$ slots and load $1-\eps$, every fixed key $x$ satisfies
  \[
    \E\Bk*{\disp(x)^{\frac{1}{d}}}=\Omega\bk*{\eps^{1-\frac{2}{d}}}.
  \]
\end{corollary}

\begin{proof}
  As above, exchangeability and \cref{lem:rh-total-displacement-lower} give
  \[
      \E\Bk*{D}=\Omega(\eps^{-1}),
      \qquad D\defeq\disp(x).
  \]
  Let $c>0$ satisfy $\E\Bk*{D}\ge c\eps^{-1}$.  We claim that, for a sufficiently large constant $k_0=k_0(d,\alpha)$,
  \begin{equation}
      \label{eq:anti-rh-truncated-displacement}
      \E\Bk*{D\ind\nolimits_{\BK*{D>k_0\eps^{-2}}}}
      \le \frac{c}{2}\eps^{-1}.
  \end{equation}
  Indeed, the \todo{what does fixed-size mean here? Jingxun: Edited.}
  tail bounds \eqref{eq:ptheta1} and \eqref{eq:ptheta2} from \cref{thm:anti-robinhood-expected-bounds}, applied to a uniformly random key, imply that for $k_0\eps^{-2}\le K\le n/10$,
  \[
      \Pr[D>K]
      \le
      C\eps\left(e^{-c_1\eps^2K}+e^{-c_1\eps K^{1/d}}\right).
  \]
  Summing this tail from $k_0\eps^{-2}$ to $n/10$ contributes at most
  \[
      C'\eps^{-1}e^{-c_2k_0^{1/d}}.
  \]
  It remains to control displacements beyond $n/10$.  By \eqref{eq:ptheta1} at $K=n/10$, the theorem's lower bound on $\eps$ makes $\Pr[D>n/10]$ smaller than $n^{-M}$ for any prescribed constant $M$, after increasing $C_{\mathrm{fin}}$.  Let $\tau_n$ be the completion threshold from \cref{lem:polynomial-completion}.  Its offsets visit every residue modulo $n$, so the probe prefix of every key contains every physical slot by true displacement $\tau_n$.  The table has an empty slot, and hence $D\le\tau_n$.  By definition, $\tau_n$ is independent of the keys and of the offset variables whose sampling intervals intersect $[0,n/10]$, which determine the event $\{D>n/10\}$.  Consequently,
  \[
      \E\Bk*{D\ind\nolimits_{\BK*{D>n/10}}}
      \le
      \E\Bk*{\tau_n}\Pr[D>n/10]
      =n^{O(1)}n^{-M}
      =o(\eps^{-1}),
  \]
  where the final equality follows by choosing $M$ sufficiently large.
  Choosing $k_0$ large now proves \eqref{eq:anti-rh-truncated-displacement}.

  It follows that
  \[
      \E\Bk*{D\ind\nolimits_{\BK*{D\le k_0\eps^{-2}}}}
      \ge \frac{c}{2}\eps^{-1}.
  \]
  On this event,
  \[
      D^{1/d}
      \ge
      \frac{D}{(k_0\eps^{-2})^{(d-1)/d}}.
  \]
  Taking expectations yields
  \[
      \E\Bk*{D^{1/d}}
      =\Omega\left(
          \frac{\eps^{-1}}{(\eps^{-2})^{(d-1)/d}}
      \right)
      =\Omega\bk*{\eps^{1-2/d}}.
  \]
  For random fixed-offset degree-$d$ probing, $\pc(x)=\Theta(D^{1/d}+1)$ pointwise, so this is also the required query-time lower bound.
\end{proof}

\subsection{A Logarithmic Lower Bound for Fixed-Offset Anti-Robin Hood Probing}
\label{subsec:anti-rh-harmonic-lower}

For $d\ge2$, and for the anti-Robin Hood ordering, the lower bound $\E\Bk*{\disp(x)}=\Omega(\eps^{-1})$ from the preceding argument is not tight enough to give the desired query-time lower bound.  We therefore lower-bound query time directly.  The argument applies to every fixed-offset probing scheme under anti-Robin Hood ordering; it uses only that all keys share an increasing offset sequence, not the spacing or randomness of that sequence.

\begin{theorem}[Logarithmic lower bound for fixed-offset probing]
  \label{thm:anti-rh-harmonic-probe-tail}
  Fix an integer $n\ge2$, a parameter $0<\eps\le1/2$, and any increasing offset sequence
  \[
      0= r_0 < r_1<r_2<\cdots
  \]
  shared by all keys.  Suppose $4\frac{\log n}{\sqrt n}\le\eps$, and let the table contain $m=\floor*{(1-\eps)n}$ keys under anti-Robin Hood ordering.  Then every fixed key $x$ satisfies
  \[
      \Pr\Bk*{\pc(x)>j}\ge\frac{1}{12j}
      \qquad\text{for every integer }1\le j\le\frac{1}{12\eps}.
  \]
  Consequently,
  \[
      \E\Bk*{\pc(x)}=\Omega\bk*{\log\eps^{-1}}.
  \]
\end{theorem}

\begin{proof}
  We first analyze a Poissonized table in which the numbers of keys at the $n$ home positions are independent $\Poi(1-\eps)$ random variables.  Condition throughout on the shared offset sequence.

  Anti-Robin Hood insertion has an exact synchronous description.  In round $j$, every home position that still has an unplaced key tries the slot $a+r_j$.  If the slot is free, one key from that home is placed; otherwise no key from that home is placed in this round.  Indeed, proposals from earlier rounds have smaller displacement and hence higher anti-Robin Hood priority.  Moreover, for fixed $j$, the map $a\mapsto a+r_j\pmod n$ is a permutation, so distinct home positions never compete in the same round.

  Let $U_j$ be the number of keys that remain unplaced after the first $j$ probes, and let $q_j$ be the expected fraction of slots that are free at that time.  If $N$ denotes the total number of Poisson keys, then exactly $N-U_j$ slots are occupied.  Therefore
  \begin{equation}
      \label{eq:anti-rh-free-backlog-relation}
      \E\Bk*{U_j}=n(q_j-\eps).
  \end{equation}
  After the home probe, a slot is free exactly when no key hashes to it, and hence
  \begin{equation}
      \label{eq:anti-rh-free-initial}
      q_1=e^{-(1-\eps)}.
  \end{equation}
  In particular, $1/3<q_1<2/3$ for $0<\eps\le1/2$.

  We next derive a recurrence that no longer contains $\eps$.  Fix a home position $a$.  As functions of the independent Poisson home counts, the event that $a$ still has an unplaced key after $j$ probes is increasing, whereas the event that $a+r_{j+1}$ is free is decreasing.  These monotonicities follow by coupling two inputs that differ by added keys: after every round, each home has at least as many unplaced keys in the larger input, and every slot occupied in the smaller input is also occupied in the larger input. With these monotonicities, we can apply the following Harris's inequality.

  \begin{theorem}[Harris's Inequality \cite{harris1960lower}]
    \label{thm:harris-inequality}
    Let $Z_1,\ldots,Z_N$ be independent random variables taking values in totally ordered spaces.  If $f$ and $g$ are integrable coordinatewise nondecreasing functions of $(Z_1,\ldots,Z_N)$, then
    \[
        \E\Bk*{fg}\ge\E\Bk*{f}\,\E\Bk*{g}.
    \]
    Consequently, if $f$ is coordinatewise nondecreasing and $g$ is coordinatewise nonincreasing, then
    \[
        \E\Bk*{fg}\le\E\Bk*{f}\,\E\Bk*{g}.
    \]
  \end{theorem}
  
  By \cref{thm:harris-inequality} and translation invariance, the expected number of placements in round $j+1$ is at most
  \begin{align*}
      & \E\Bk*{\sum_{a\in[n]} \ind\nolimits_{\BK*{a\text{ has an unplaced key after }j\text{ probes}}} \cdot \ind\nolimits_{\BK*{a+r_{j+1}\text{ is free after }j\text{ probes}}}}\\
      \le{}&\sum_{a\in[n]}
      \Pr\Bk*{a\text{ has an unplaced key after }j\text{ probes}}
      \Pr\Bk*{a+r_{j+1}\text{ is free after }j\text{ probes}}\\
      \le{}&
      \E\Bk*{U_j}\,q_j
      =n(q_j-\eps)q_j
      \le nq_j^2.
  \end{align*}
  Every placement fills one free slot, so
  \begin{equation}
      \label{eq:anti-rh-free-recurrence}
      q_{j+1}\ge q_j-q_j^2=q_j(1-q_j).
  \end{equation}

  The sequence $(q_j)$ is nonincreasing, and hence $q_j<2/3$ by \eqref{eq:anti-rh-free-initial}.  Taking reciprocals in \eqref{eq:anti-rh-free-recurrence} gives
  \[
      \frac1{q_{j+1}}
      \le
      \frac1{q_j(1-q_j)}
      =
      \frac1{q_j}+\frac1{1-q_j}
      \le
      \frac1{q_j}+3.
  \]
  Since $1/q_1<3$, induction yields
  \begin{equation}
      \label{eq:anti-rh-harmonic-free-tail}
      q_j\ge\frac1{3j}
      \qquad(j\ge1).
  \end{equation}
  Combining \eqref{eq:anti-rh-free-backlog-relation} and \eqref{eq:anti-rh-harmonic-free-tail}, for every $j\le1/(6\eps)$ we obtain
  \begin{equation}
      \label{eq:anti-rh-harmonic-poisson-tail}
      \E\Bk*{U_j}
      =n(q_j-\eps)
      \ge\frac{n}{6j}.
  \end{equation}

  It remains to return to exactly $m$ keys.  Couple the Poissonized and fixed-load tables by taking the first $N$ and the first $m$ keys, respectively, from one infinite i.i.d. key sequence.  Adding or deleting one key changes $U_j$ by at most one.  Indeed, the round-by-round coupling above shows that adding a key cannot reduce any home backlog and cannot make a previously occupied slot free.  Since $U_j$ equals the number of keys minus the number of occupied slots, its increase is therefore between $0$ and $1$.  Since $\E\Bk*{N}=(1-\eps)n$ and $\abs{m-(1-\eps)n}<1$,
  \[
      \E\Bk*{\abs{N-m}}\le\sqrt n+1.
  \]
  Therefore, for $j\le1/(6\eps)$,
  \[
      \E\Bk*{U_j^{\mathrm{fix}}}
      \ge
      \frac{n}{6j}-\sqrt n-1.
  \]
  The theorem's assumption and $n\ge2$ imply
  \[
      \sqrt n+1\le\eps n.
  \]
  Hence, for every integer $1\le j\le1/(12\eps)$,
  \[
      \sqrt n+1
      \le
      \eps n
      \le
      \frac{n}{12j},
  \]
  and therefore
  \[
      \E\Bk*{U_j^{\mathrm{fix}}}
      \ge
      \frac{n}{12j}.
  \]
  The $m$ fixed-load keys are exchangeable, and $U_j^{\mathrm{fix}}$ counts exactly the keys with probe complexity greater than $j$.  Thus every fixed key $x$ satisfies
  \[
      \Pr\Bk*{\pc(x)>j}
      =\frac{\E\Bk*{U_j^{\mathrm{fix}}}}{m}
      \ge
      \frac{1}{12j}
  \]
  throughout this range, proving the stated tail.  Finally, the tail-sum formula and the trivial bound $\pc(x)\ge1$ give
  \[
      \E\Bk*{\pc(x)}
      \ge
      1+\frac1{12}
        \sum_{j=1}^{\floor*{1/(12\eps)}}\frac1j
      =\Omega\bk*{\log\eps^{-1}}. \qedhere
  \]
\end{proof}

\begin{corollary}[Lower bound for anti-Robin Hood random fixed-offset degree-$d$ probing when \texorpdfstring{$d\ge2$}{d >= 2}]
  \label{cor:anti-robinhood-displacement-lower-d-ge-2}
  Fix a constant $d\ge2$, a parameter $0<\eps\le1/2$, and an integer $n$ satisfying $C_{\mathrm{fin}}\log n\cdot\max\{n^{-1/2},n^{-1/d}\}\le\eps$.  In an anti-Robin Hood random fixed-offset degree-$d$ probing hash table with $n$ slots and load $1-\eps$, every fixed key $x$ satisfies
  \[
    \E\Bk*{\pc(x)}=\Omega\bk*{\log \eps^{-1}}.
  \]
\end{corollary}

\begin{proof}
  The parameter range of the corollary implies the hypothesis of \cref{thm:anti-rh-harmonic-probe-tail}, whose expectation conclusion gives the claim directly.
\end{proof}


\section{AI Acknowledgment}

Several of the sections in this paper were written with the assistance of GPT 5.5 Pro, which was used to convert detailed proof outlines into first drafts for sections (which were then heavily edited by the authors). Additionally, GPT 5.5 Pro suggested the use of Harris's inequality in Section \ref{sec:lower-bounds} as a way to simplify the anti-Robin Hood lower bound proof.

\section{Funding Acknowledgments}

William Kuszmaul, Jingxun Liang, and Renfei Zhou were partially supported by NSF grant CNS-2504471 and by a Jane Street Research Grant. Huacheng Yu was supported by NSF CAREER award CCF-2339942. Renfei Zhou was partially supported by the Jane Street Graduate Research Fellowship and the MongoDB PhD Fellowship. Part of this work was initiated by Dagstuhl Seminar 25191``Adaptive and
Scalable Data Structures''.

\bibliography{reference}

\appendix

\section{Completion Bound for the Offset Sequence}
\label{app:completion}
The query-time bounds invoke \cref{lem:polynomial-completion} only after a probe sequence has reached a displacement on the order of the table size.  The expected completion threshold is enough: once a sufficiently far tail of the offset sequence has represented every residue class, a query has inspected every physical slot and must have stopped.  We restate the lemma for convenience.

\polycompletion*

\begin{proof}[Proof of \cref{lem:polynomial-completion}]
  We first consider random fixed-offset degree-$d$ probing.  Let
  \[
      m=\lceil Cn\log(2n)\rceil,
  \]
  where $C$ is a sufficiently large constant, and partition the probe indices beginning at $m$ into consecutive blocks of length $m$.  Consider a block
  \[
      J,J+1,\ldots,J+m-1,
      \qquad J\ge m.
  \]
  The sampling intervals in this block have comparable lengths, namely $|I_J|=\Theta(J^{d-1})$, with this quantity interpreted as a constant when $d=1$.

  Fix a residue class $a$ modulo $n$.  The union of the sampling intervals in the block has length $\Theta(J^{d-1}m)$ and hence contains $\Theta(J^{d-1}m/n)$ integer offsets congruent to $a$ modulo $n$.  Each such integer belongs to one sampling interval, which selects it with probability $\Omega(J^{1-d})$.  Therefore
  \[
      \sum_{j=J}^{J+m-1}\Pr[r_j\equiv a\pmod n]
      \ge
      c\frac{m}{n}
      \ge
      cC\log(2n)
  \]
  for a constant $c=c(d,\alpha)>0$.  Independence within the block shows that the probability of missing $a$ is at most $(2n)^{-cC}$.  After a union bound over all residue classes and an appropriate choice of $C$, every block covers all residues with probability at least $1/2$, independently of the other blocks.

  Let $G$ be the index of the first successful block.  Then $G$ is stochastically dominated by a geometric random variable of constant mean.  
  \todo{What is $s_n$? Do you just mean $n$? Jingxun: Edited.}
  By the end of the first successful block, the offsets allowed in the definition of $\tau_n$ have visited every residue class, and the largest exposed offset is $O(((G+1)m)^d)$.  Therefore
  \[
      \tau_n
      =
      O(((G+1)m)^d).
  \]
  A geometric random variable has a finite $d$-th moment, and hence
  \[
      \E\Bk*{\tau_n}
      =
      O(m^d)
      =
      n^{O(1)}.
  \]
  By its definition through indices $j\ge s_n$, $\tau_n$ is measurable only with respect to offset variables whose sampling intervals are disjoint from $[0,n/10]$.

  Now consider smoothed degree-$d$ probing for a fixed key $x$.  Choose
  \[
      T=\lceil(C'n\log(2n))^d\rceil
  \]
  for a sufficiently large constant $C'$, and partition the offsets above $T$ into dyadic blocks $[2^\ell T,2^{\ell+1}T)$, $\ell\ge0$.  For any residue class $a$ modulo $n$,
  \[
      \sum_{\substack{2^\ell T\le k<2^{\ell+1}T\\ k\equiv a\pmod n}}p(k)
      \ge
      c'2^{\ell/d}\frac{T^{1/d}}{n}
      =
      \Omega\Bk*{2^{\ell/d}C'\log(2n)}.
  \]
  Thus the probability that the $\ell$-th block fails to contain an activated representative of every residue class is at most
  \[
      \exp\Bk*{-\Omega(2^{\ell/d}\log(2n))}.
  \]
  Let $L$ be the first successful dyadic block.  The blocks use disjoint activation bits, and the preceding estimate implies
  \[
      \E\Bk*{2^L}
      \le
      1+
      \sum_{\ell\ge0}2^{\ell+1}\Pr[L>\ell]
      =
      O(1).
  \]
  Since $T>n$ and $\tau_n(x)\le2^{L+1}T$,
  \[
      \E\Bk*{\tau_n(x)}
      =
      O(T)
      =
      n^{O(1)}.
  \]
  Taking a common exponent $B$ large enough proves both assertions.
\end{proof}

\section{Proofs of Facts from Section~\ref{sec:prelim}}
\label{app:reductions}
In this section, we provide proofs for the facts \cref{lem:history-independence-of-rh-and-arh,lem:displacement-and-probe-complexity} from \cref{sec:prelim}.

\subsection{History Independence of the Two Orderings}

History independence for priority-based open-addressing tables, including Robin Hood-style rules, has been studied in the hashing literature; see, for example, \cite{attiya2025historyindependent} and the references therein.  Our setting differs in one bookkeeping detail: in a cyclic table, the same key may reach the same physical slot at several true offsets that differ by multiples of $n$, and these occurrences can have different priorities.  The standard deferred-acceptance proof otherwise applies unchanged.  For completeness, we spell it out while retaining the true offset in every proposal contract.

\historyindependence*

\begin{proof}[Proof of \cref{lem:history-independence-of-rh-and-arh}]
  We use key-proposing deferred acceptance in the sense of Gale and Shapley~\cite{gale1962college}, expressed in the matching-with-contracts framework of Hatfield and Milgrom~\cite{hatfield2005matching}.  Fix the set of keys, their base hashes, their realized probe sequences, and their tie-breaking values.  For every occurrence of a true offset $k$ in the probe sequence of a key $x$, introduce a contract $(x,k)$ with the physical slot
  \[
      h(x)+k \pmod n.
  \]
  Contracts $(x,k)$ and $(x,k')$ remain distinct even if $k\equiv k'\pmod n$.  Each key ranks its contracts in probe order, equivalently by increasing true offset.  Each slot ranks the contracts pointing to it according to the insertion rule: Robin Hood ordering prefers larger true offset, anti-Robin Hood ordering prefers smaller true offset, and equal offsets are ordered by the fixed tie-breaking values.  All preference orders are therefore strict.

  Run key-proposing deferred acceptance.  An unmatched key proposes its next contract; the receiving slot keeps the contract it prefers and rejects the other one, if any; and a rejected key continues with its next contract.  This is exactly the insertion-and-kicking algorithm.  Choosing an insertion order merely determines when each key makes its first proposal and hence selects an asynchronous proposal schedule for the same deferred-acceptance instance.

  We use the standard schedule-independence property of deferred acceptance with finite strict preference lists.  Its proof is the usual Gale--Shapley proof, which has two steps.  First, the rejection lemma says that no rejected contract belongs to a stable matching.  Indeed, let $(x,s)$ be the earliest rejected contract in a purported stable matching.  Slot $s$ rejects it for a key $y$ that it prefers to $x$.  Since $y$ has proposed to $s$, it has already been rejected from every contract it prefers to $s$; by minimality, its purported match is not one of them.  Thus $y$ prefers $s$ to its match, and $(y,s)$ blocks the purported matching.  Second, a terminating fair schedule is stable, and each key has been rejected from every contract it prefers to its terminal one.  The rejection lemma makes the terminal matching key-optimal; strict preferences make it unique.  Hence the outcome is schedule independent.

  Now compare two completed insertion executions.  For each key, truncate its contract list after the largest index proposed in either execution.  Both executions are terminating fair schedules for this finite instance, so both produce its unique key-optimal matching.  They therefore accept the same contract $(x,k)$ for every key.  Since a contract records both the physical slot and the true offset, the two executions have the same table state and displacements.  Hence the table is history independent.
\end{proof}

History independence allows us to choose any convenient insertion order in the analysis.  In the random models considered here, exchangeability of the keys additionally implies that identically distributed inserted keys have the same expected displacement and probe complexity.

\subsection{Reducing Probe Complexity to Displacement}
Suppose that a key $x$ is stored on its $i$-th probe, at
\[
    h_i(x)=h(x)+r_i(x)\pmod n.
\]
Then $\disp(x)=r_i(x)$ and $\pc(x)=i$.  For random fixed-offset degree-$d$ probing, the sampling-interval bounds give the pointwise relation $\pc(x)=\Theta(\disp(x)^{1/d}+1)$.  In smoothed degree-$d$ probing, the number of activated offsets below a given displacement fluctuates, so no analogous pointwise relation holds in general.  The following expectation identity is the appropriate substitute.

\displacementprobecomplexity*

\begin{proof}[Proof of \cref{lem:displacement-and-probe-complexity}]
  By history independence, we may insert $x$ last.  Let $D=\disp(x)$, and let $A_{x,r}$ indicate that offset $r$ is activated for $x$, with $A_{x,0}=1$.  The probe complexity is exactly the number of activated offsets up to and including the final displacement, so
  \begin{equation}
    \label{eq:pc-from-displacement}
    \pc(x)=\sum_{r\ge 0} A_{x,r}\ind\nolimits_{\BK*{D\ge r}}.
  \end{equation}
  Fix $r\ge1$, condition on the table formed by the other keys, and expose all randomness of $x$ except $A_{x,r}$.  Whether $x$ advances past every offset smaller than $r$ is decided before the bit $A_{x,r}$ is examined, and this event is precisely $\{D\ge r\}$.  Hence $A_{x,r}$ is independent of $\{D\ge r\}$.  Taking expectations in \eqref{eq:pc-from-displacement} and using $p(0)=1$ gives
  \begin{align*}
    \E\Bk*{\pc(x)}
      &= \sum_{r\ge 0} p(r)\Pr\Bk*{D\ge r} \\
      &= \E\Bk*{\sum_{r=0}^{D}p(r)}.
  \end{align*}
  Since $p(r)=\Theta(r^{-(d-1)/d})$ for $r\ge1$, the claimed identity follows by substituting $m=D$ into
  \[
    \sum_{r=0}^{m}p(r)=\Theta\bk*{m^{1/d}+1}
    \qquad\text{for every integer }m\ge0.
    \qedhere
  \]
\end{proof}

\end{document}